\documentclass[conference]{IEEEtran}
\usepackage{amsmath}
\usepackage{amsfonts}
\usepackage{epsfig}
\usepackage{amssymb}
\usepackage{cite}
\usepackage{multicol}
\usepackage{graphicx}
\usepackage{adjustbox}
\usepackage{xpatch,amsthm}
\usepackage{mathtools}
\usepackage[numbers,sort&compress]{natbib}

\begin{document}
\title{Are Generalized Cut-Set Bounds Tight for the Deterministic Interference Channel?}
\author{\IEEEauthorblockN{Mehrdad Kiamari and A. Salman Avestimehr}\\
\IEEEauthorblockA{Department of Electrical Engineering, \\University of Southern California\\
Emails: kiamari@usc.edu and avestimehr@ee.usc.edu}}

\maketitle
\begin{abstract}
We propose the idea of \emph{extended networks}, which is constructed by replicating the users in the two-user deterministic interference channel (DIC) and designing the interference structure among them, such that any rate that can be achieved by each user in the original network can also be achieved \emph{simultaneously} by all replicas of that user in the extended network. We demonstrate that by carefully designing extended networks and applying the generalized cut-set (GCS) bound to them, we can derive a tight converse for the two-user DIC. Furthermore, we generalize our techniques to the three-user DIC, and demonstrate that the proposed approach also results in deriving a tight converse for the three-user DIC in the symmetric case.
\end{abstract}

\section{Introduction} \label{sec1}
\IEEEPARstart{A}
 ~common tool for deriving outer bounds on the capacity region of communication networks is the cut-set bound \cite{cover}. It is well-known that the cut-set bound is tight in several general problems, such as the unicast and multicast wireline networks \cite{adet}-\cite{netcod}. Furthermore, it can also approximate the capacity of Gaussian relay networks to within a constant gap \cite{Gaussian}. However, once we go beyond the unicast and multicast problems, the cut-set bound is typically loose, even for the simplest case of the two-user deterministic interference channel (DIC) \cite{elgamal}. Therefore, to make progress in these problems, researchers have focused on obtaining new techniques to develop tighter outer bounds, such as the genie-aided bounds (see, e.g., \cite{genie}). Although these techniques enable progress, they are typically tailored to the specific problem setting that is considered, and unlike the cut-set bound they cannot be systematically applied to other problems.

As a result, there have recently been several efforts to generalize the cut-set bound in order to systematically obtain tighter outer bounds on the capacity region of communication networks, in particular the generalized network sharing (GNS) bound for wireline networks \cite{tse_sharing}, the generalized cut-set (GCS) bound for deterministic networks \cite{gcs}, and its extension to noisy networks \cite{yang_sharing}. The key idea behind the GCS bound \cite{gcs}, which is the main focus of this paper, is to construct a serial concatenation of a network and derive a tighter bound by applying the cut-set bound to the concatenated network, with a special restriction on its allowed transmit signal distributions. This bound has resulted in development of several new outer bounds on the deterministic $k$-unicast two-hop networks \cite{gcs}.

In this paper, we focus on the two-user DIC and investigate whether GCS bound can be systematically utilized for deriving a tight converse for this problem. While directly applying GCS bound on the two-user DIC does not lead to tight outer bounds, we introduce the idea of \emph{extended networks}, and show that, quite interestingly, by carefully designing extended networks and applying the GCS bound to them, we are able to derive tight outer bounds on the capacity region of the two-user DIC. The main idea behind extended networks is to create several copied versions of each user in the two-user DIC, with an interference structure among them, such that all copied users in the extended network can \emph{simultaneously} achieve the same rate as the original user by deploying the same coding scheme as the one used in the original network. Therefore, any sum-rate bound on the extended network can lead to a weighted sum-rate bound on the original network, where the coefficient of the corresponding rate weighted sum-rate bound depends on the number of copied versions of each user in the extended network. We demonstrate that by carefully designing the structure of extended networks and applying GCS bound to them, we can systematically derive all outer bounds not the capacity region of the two-user DIC.

The idea of extended networks is general, and can also be utilized to derive outer bounds on interference channels with more than two users. We consider the symmetric three-user DIC \cite{jaafaar}, and demonstrate that all outer bounds on the capacity region of this problem can also be recovered by carefully constructing the corresponding extended network for each outer bounds and applying the GCS to it.

We start the paper by describing the system model of the two-user DIC and its capacity region in Section~\ref{sec2}, and providing an overview of the generalized cut-set (GCS) bound \cite{gcs} in Section~\ref{sec3}. We then formally introduce the idea of extended networks in Section~\ref{sec4}, and demonstrate how all bounds on the capacity region of the  two-user DIC can be derived by constructing extended networks and applying the GCS bound to them. We finally consider the three-user DIC in Section~\ref{sec5}, and demonstrate how the ideas proposed in this paper can be generalized beyond two users to derive all outer bounds on the capacity region of this problem in the symmetric case.
\section{System Model of the Two-User DIC} \label{sec2}
The general system model of the two-user DIC is depicted in Fig. \ref{fig1}. This network, which is denoted by $\mathcal N$, consists of two source nodes $S_1$ and $S_2$ producing $W_1 \in {\{1,2,...,M_1\}}$ and $W_2 \in {\{1,2,...,M_2\}}$ respectively. Then, encoder 1 and encoder 2 maps their corresponding messages, i.e. $W_1$ and $W_2$, into ${\bf X}_{1}^n=(X_{1}[1],...,X_{1}[n])$ and ${\bf X}_{2}^n=(X_{2}[1],...,X_{2}[n])$ respectively. The received signals at destination 1 and destination 2, i.e. $Y_1$ and $Y_2$ are deterministic functions of transmitted signals $X_1$, $X_2$, and interference signals $V_1$ and $V_2$:
\begin{figure}[htb]
\centering
\begin{center}
\includegraphics[width=90 mm]{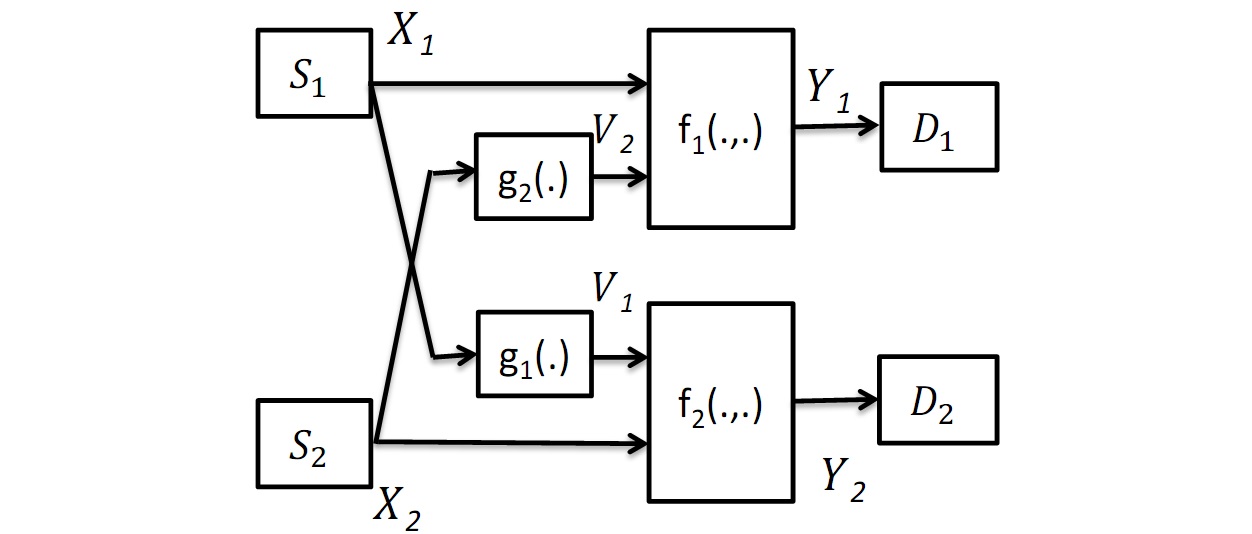} 
\caption{System model of the two-user DIC.}
\label{fig1}
\end{center}
\end{figure}
\small
\begin{equation} \label{eq0203}
\begin{aligned}
Y_1&=f_1(X_1,V_2)
\\Y_2&=f_2(X_2,V_1)
\\V_1&=g_1(X_1)
\\V_2&=g_2(X_2)
\end{aligned}
\end{equation}
\normalsize
where $g_1(.)$ and $g_2(.)$ are not generally invertible functions and $f_1(.,.)$ and $f_2(.,.)$ satisfy the following conditions which is the same as \cite{elgamal}
\small
\begin{equation} \label{eq0204}
\begin{aligned}
H(Y_1|X_1)&=H(V_2)
\\H(Y_2|X_2)&=H(V_1)
\end{aligned}
\end{equation}
\normalsize
these conditions are equivalent to
\small
\begin{equation} \label{eq0205}
\begin{aligned}
V_1&=h_2(X_2,Y_2)
\\V_2&=h_1(X_1,Y_1)
\end{aligned}
\end{equation}
\normalsize

\makeatletter
   \xpatchcmd{\@thm}{\fontseries\mddefault\upshape}{}{}{} 
\makeatother
\newtheorem{theorem}{Theorem}

\begin{theorem}
The capacity region $C$ of the two-user DIC is the union of all set of rate tuple $(R_1,R_2)$ satisfying
\small
\begin{subequations} \label{eq0206}
\begin{align}
R_1&\leq H(Y_1|V_2)
\\R_2&\leq H(Y_2|V_1)
\\R_1+R_2 &\leq H(Y_1)+H(Y_2|V_1,V_2)
\\R_1+R_2 &\leq H(Y_2)+H(Y_1|V_1,V_2)
\\R_1+R_2 &\leq H(Y_1|V_1)+H(Y_2|V_2)
\\2R_1+R_2 &\leq H(Y_1)+H(Y_2|V_2)+H(Y_1|V_1,V_2)
\\2R_2+R_1 &\leq H(Y_2)+H(Y_1|V_1)+H(Y_2|V_1,V_2)
\end{align}
\end{subequations}
\normalsize
over all product probability distributions $p(x_1)p(x_2)$.
\end{theorem}
\begin{proof}[{Proof}]
The achievable scheme for this problem is based on Han-Kobayashi (HK) \cite{achiev}, and the converse has been derived by El Gamal and Costa \cite{elgamal}.
\end{proof}
As it is mentioned in introduction, our goal is to demonstrate all outer bounds (4a)-(4g) can be systematically derived by the idea of extended network and applying GCS bound to it. To that end, we next introduce GCS bound and extended networks.
A simple version of extended networks idea, called \emph{side-way concatenation}, was first introduced in \cite{thesis}. However, with this side-way concatenation, one can only recover bounds (4f) and (4g)(See section 7.3.1 of \cite{thesis}). The main contribution of this paper is to demonstrate that all bounds (4a)-(4f) can be recovered by the idea of extended networks and applying the GCS bound to them.
Not only can this technique be utilized for the two-user DIC, but it can be applied for three-user DIC as well.

\section{Overview of Generalized Cut Set Bound} \label{sec3}
Generalized Cut-Set bound is a generalization of classical cut-set bound introduced by \cite{gcs} which allows for deriving upper bound on capacity region.

In particular, consider a K-unicast memoryless network $\mathcal N$ including a set of nodes $\mathcal V$ which has $K$ sources $\mathcal S=[S_1,....,S_K]$ and $K$ destinations $\mathcal D=[D_1,....,D_K] $. At each time step $t$, $X_v [t]$ is the transmitted symbol from $v \in \mathcal V$ and $Y_v [t]$ is the received symbol at $v \in \mathcal V$. In a deterministic network, the received signals at time step $t$ is a function of transmitted signals which means $Y_{\mathcal V}[t]=F(X_{\mathcal V}[t])$. We assume that source nodes do not receive any symbols and destination nodes do not transmit any symbols.

The coding scheme $\mathcal C_n$ with rate tuple $(R_1,...,R_K)$ can be expressed as follow
\begin{enumerate}
  \item Encoder $E_i : \{1,2,...,2^{nR_{i}}\} \longrightarrow {\mathcal X}_{S_i}^{n}$ for $i=1,...,K $
  \item Relaying node $r_v^{(t)} : r_v^{(t-1)} \longrightarrow {\mathcal X}_v$ for $v \in \mathcal V \setminus \{ \mathcal S \cup \mathcal D \}$ and $t=1,...,n$
  \item Decoder $D_i : {\mathcal Y}_{D_i}^{n}  \longrightarrow \{1,2,...,2^{nR_{i}}\}$ for $i=1,...,K $
\end{enumerate}
and error probability $P_e(\mathcal C_n)$ can be written as
\begin{equation} \label{eq01001}
\begin{aligned}
P_e(\mathcal C_n)&= Prob\{W_i\neq g_i(Y_{D_i}[1],...,Y_{D_i}[n])\\&~~~~~~~~~~ for~some~i~ \in {1,...,K}\}.
\end{aligned}
\end{equation}

A rate tuple $(R_1,...,R_K)$ is said to be achievable if for any $\epsilon >0$, there exists a coding scheme such that $P_e(\mathcal C_n) \leq \epsilon$ for some large $n$. The closure of all achievable rate tuples would specify the capacity region $C \subset {\mathbb{R}}_{+}^K$.
\cite{gcs} has derived an upper bound for sum-rate on $\mathcal N$ in the case of deterministic networks as follow

\begin{theorem}[{\cite{gcs}}]
If a rate tuple $(R_1,...,R_K)$ is achievable on a K-unicast deterministic network, then it can be shown that there exists a joint distribution $p(x_V)$ on sources such that
\begin{equation} \label{eq0201}
\begin{aligned}
\sum_{i=1}^K {R_i} \leq \sum_{j=1}^l {I(X_{{\Omega}_j};Y_{{\Omega}_j^c}|X_{{\Omega}_j^c},Y_{{\Omega}_{j-1}^c})}
\end{aligned}
\end{equation}
for all choices of l node subsets ${\Omega}_1,....,{\Omega}_l$ such that $\mathcal V ={\Omega}_{0} \supseteq {\Omega}_1 \supseteq ... \supseteq{\Omega}_{l+1}=\emptyset$, and $d_i \in {\Omega}_j \Longleftrightarrow s_i \in {\Omega}_{j+1}$ for $i=1,...,K$ and $j=0,...,l$.
\label{th2}
\end{theorem}

It should be noted that in the case of deterministic network, (\ref{eq0201}) can be written as follow
\begin{equation} \label{eq0202}
\begin{aligned}
\sum_{i=1}^K {R_i} \leq \sum_{j=1}^l {H(Y_{{\Omega}_j^c}|X_{{\Omega}_j^c},Y_{{\Omega}_{j-1}^c})}.
\end{aligned}
\end{equation}

As an example of GCS bound, we demonstrate that (4c) and (4d) can be derived by applying Theorem 2 on the two-user DIC. It is easy to verify that (4c) and (4d) cannot be found through classical cut-set bounds. If a rate tuple $(R_1,R_2)$ is achievable on $\mathcal N$, then there exists a joint distribution on the sources of $\mathcal N$ such that
\small
\begin{subequations} \label{eq01001}
\begin{align}
R_1+R_2 &\leq H(Y_1)+H(Y_2|V_1,V_2)
\\R_1+R_2 &\leq H(Y_2)+H(Y_1|V_1,V_2)
\end{align}
\end{subequations}
\normalsize

By utilizing GCS bound on the original network with the cuts depicted in Fig. \ref{fig2}, we have
\small
\begin{equation} \label{eq402}
\begin{aligned}
R_1+R_2 &\leq H(Y_{2})+H(Y_{1}|X_2,Y_2)
\\&=H(Y_{2})+H(Y_1|X_2,Y_2,g_2(X_2),h_2(X_2,Y_2))
\\&=H(Y_{2})+H(Y_1|X_2,Y_2,V_2,V_1)
\\&\leq H(Y_2)+H(Y_1|V_1,V_2).
\end{aligned}
\end{equation}
\normalsize

Regarding bound (4d), it suffices to consider the original network. By applying GCS bound on the original network with the cuts depicted in Fig. \ref{fig3}, we have
\small
\begin{equation} \label{eq403}
\begin{aligned}
R_1+R_2 &\leq H(Y_{1})+H(Y_{2}|X_1,Y_1)
\\&=H(Y_{1})+H(Y_2|X_1,Y_1,g_1(X_1),h_1(X_1,Y_1))
\\&\leq H(Y_1)+H(Y_2|V_1,V_2).
\end{aligned}
\end{equation}
\normalsize

\begin{figure}[h]
\centering
\includegraphics[width=90 mm]{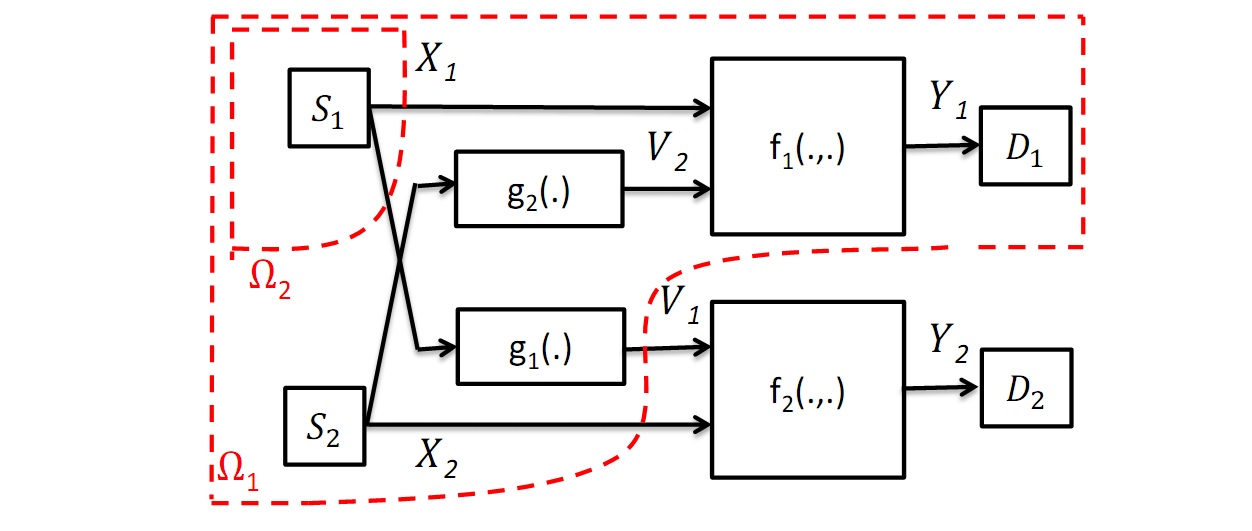}
\caption{Depicting the cuts for deriving bound (4c) of the two-user DIC}
\label{fig2}
\end{figure}

\begin{figure}[htb]
\centering
\begin{center}
\includegraphics[width=90 mm]{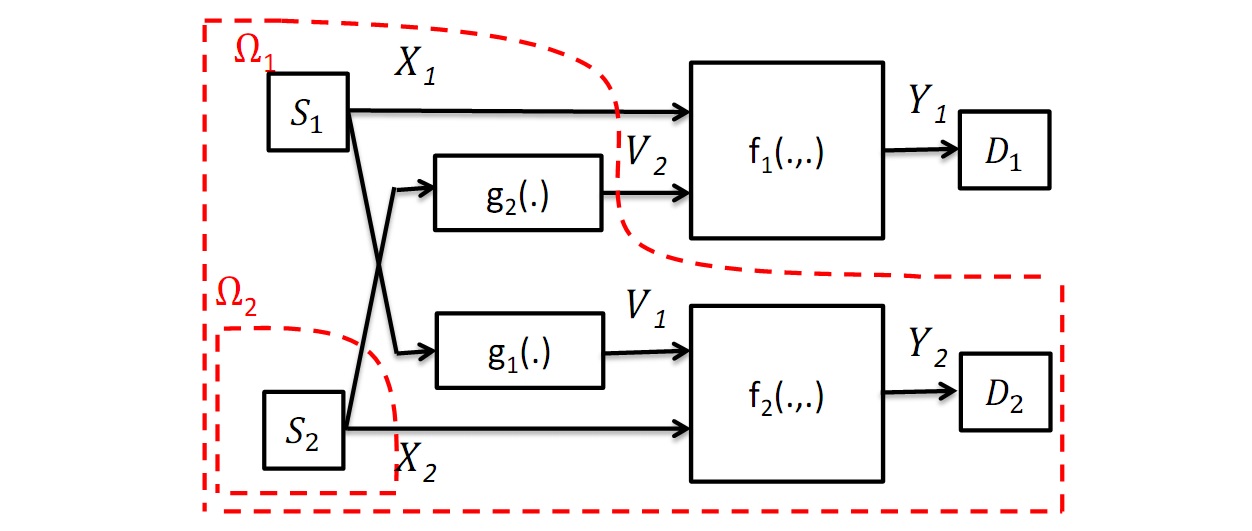} 
\caption{Depicting the cuts for deriving bound (4d) of the two-user DIC}
\label{fig3}
\end{center}
\end{figure}

Although GCS method recovers two bounds, i.e. (4c)-(4d), it fails to recover the remaining five bounds. The question raised here is that whether we can obtain the rest bounds or not. We will show that performing GCS on extended network rather than the original network is a solution for finding the remaining bounds.
\section{Converse for Two-User DIC} \label{sec4}
We start by formally defining an \emph{extended network}.
\newenvironment{definition}[1][Definition]{\begin{trivlist}
\item[\hskip \labelsep {\bfseries #1}]}{\end{trivlist}}
\begin{definition}
Consider a two-user DIC, denoted by $\mathcal N$, as defined in section II. A $(k_1,k_2)$ extended version of $\mathcal N$, denoted by $N_{k_1,k_2}$ is constructed by considering $k_1$ and $k_2$ copied versions of user 1 and user 2, respectively, as depicted in Fig. \ref{fig4}. In this extended network, each receiver gets interference only from one of the copied versions of the other user. It should be noted that interference function, i.e. $g_i(.)$, and function $f_i(,.,)$ for the copied versions of user $i$ are the same as the functions considered for the original user $i$ on the two-user DIC. The choice of interference pattern depends on the outer bound which is supposed to be derived. In Fig. \ref{fig4}, $S_{i}^{(j)}$ represents the source node of $j$th copied version of $i$th user. The message corresponded to source $S_{{i}^{(j)}}$ is denoted by $W_{{i}^{(j)}}\in [1:2^{nR_i}]$. It should be noted that the messages of different sources are independent from each other even if the same coding scheme is utilized for them. In this extended network, encoder function $E_i: {[1:2^{nR_i}]}\rightarrow {\mathcal X}_{S_i}^{n}$ and decoder function $D_i: {\mathcal Y}_{D_i}^{n} \rightarrow {[1:2^{nR_i}]}$ for $i=1,2$ are applied to the copied versions of $i$th user. Then, encoder function $E_i$ maps the corresponding message $W_{{i}^{(j)}}$ into $X_{i^{(j)}}^n=[X_{i^{(j)}}{[1]},...,X_{i^{(j)}}{[n]}]$. Furthermore, the received signal at each copy is a deterministic function of transmitted signal of its corresponding source and interference signal coming from the copied versions of the other user.
\end{definition}

\newtheorem{lem}{Lemma}
\begin{lem}
\label{lemext}
If a rate tuple $(R_1,R_2)$ is achievable on the two-user DIC, $\mathcal N$, then rate tuple $\{{\underbrace{R_1,...,R_1}_{k_1}},{\underbrace{R_2,...,R_2}_{k_2}}\}$ is achievable on any $(k_1,k_2)$ extended version of the network, $N_{k_1,k_2}$ by applying the same coding scheme as $S_i$ in $\mathcal N (i=1,2)$ to all $S_i^{(j)}$'s $(i=1,2, j=1,...,k_i)$ in $\mathcal N_{k_1,k_2}$.
\end{lem}
\begin{proof}[{Proof}]
Consider a coding scheme that results in transmission of $X_1^n, X_2^n$ and reception of $Y_1^n, Y_2^n$ in $\mathcal N$. Consider a $(k_1,k_2)$ extended version of the network, $\mathcal N_{k_1,k_2}$, and assume the same coding scheme as $S_i$ in $\mathcal N (i=1,2)$ is applied to all $S_i^{(j)}$'s $(i=1,2, j=1,...,k_i)$ in $\mathcal N_{k_1,k_2}$. Now, note the following three points;
\begin{enumerate}
  \item The coding scheme utilized on the $j$th copied version of user $i$ on the extended network is the same as the coding scheme used on user $i$ on the original network
  \item The interference function, i.e. $g_i(.)$, and function $f_i(.,.)$ in the extended network are the same as the functions considered for user $i$ in the original network
  \item user $i$ gets interference \emph{only} from one of copied versions of user $\bar i$ for $i=1,2$($\bar i =3-i$)
\end{enumerate}
It is easy to see that these three points imply $I(X_{i^{(j)}}^n,Y_{i^{(j)}}^n)=I(X_{i}^n,Y_{i}^n) ~~~\forall j=1,...,k_i , i=1,2$. Therefore, any rate tuple $\{{\underbrace{R_1,...,R_1}_{k_1}},{\underbrace{R_2,...,R_2}_{k_2}}\}$ is achievable on any $(k_1,k_2)$ extended version of the network, $N_{k_1,k_2}$.
\end{proof}
Therefore, any sum-rate upper bound on $\mathcal N_{k_1,k_2}$ would result in an upper bound on $k_1R_1+k_2R_2$ in the original network. In particular, by applying GCS bound on $N_{k_1,k_2}$, we can systematically derive a bound for $k_1R_1+k_2R_2$ on $\mathcal N$.

In the remaining parts, we derive all remaining outer bounds, i.e. (4a)$, $(4b)$, $(4e)$, $(4f), and (4g), by designing specific extended network for each of them and applying GCS to them.
\begin{figure}[htb]
\centering
\begin{center}
\includegraphics[width=90 mm]{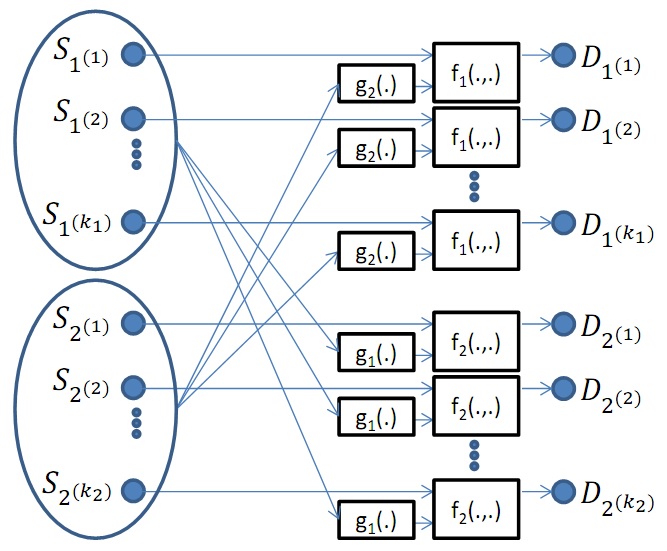} 
\caption{System model of the extended network $\mathcal N_{k_1,k_2}$.}
\label{fig4}
\end{center}
\end{figure}
\subsection{Deriving Bounds (4a) and (4b)} \label{sec4_1}
In this section, we demonstrate how to derive individual rate bound (4a) by applying GCS bound to an appropriate extended network. Due to symmetry, we only need to find the bound for user 1, i.e. (4a).

We utilize extended network $\mathcal N_{k,1}$, as depicted in Fig. \ref{fig5} for this bound. As you can note, in this network, second user imposes interference on all copied versions of first user.

According to Lemma\ref{lemext}, if a rate tuple $(R_1,R_2)$ is achievable on $\mathcal N$, then rate tuple $({\underbrace{R_1,...,R_1}_{k}},R_2)$ is achievable on $\mathcal N_{k,1}$. Consider ${\{\mathcal C_n\}}$ as a sequence of coding scheme with block length $n$ that achieves sum rate $R_{\Sigma}$ on $\mathcal N_{k,1}$. Therefore, by applying GCS bound on the extended network $\mathcal N_{k,1}$ and picking cuts $\Omega_{1},...,\Omega_{k+1}$ depicted in Fig. \ref{fig5}, we have
\small
\begin{equation} \label{eq0205}
\begin{aligned}
nR_{\Sigma}&=n(R_1+R_2+(k-1)R_1)
\\& \overset{(a)}{\leq} I(W_{\mathcal S};Y_{{\Omega}_1^c}^n)\\&+\sum_{l=2}^{k+1}{I(W_{\mathcal S\cap{{\Omega}_l}};Y_{{\Omega}_l^c\bigcap {\Omega}_{l-1}}^n|W_{\mathcal S\backslash{{\Omega}_l}},Y_{{\Omega}_{l-1}^c}^n)}+n\epsilon_n
\\&\overset{(b)}{=}H(Y_{{\Omega}_1^c}^n)+\sum_{l=2}^{k+1}{H(Y_{{\Omega}_l^c\bigcap {\Omega}_{l-1}}^n|W_{\mathcal S\backslash{{\Omega}_l}},Y_{{\Omega}_{l-1}^c}^n)}+n\epsilon_n
\\& \overset{(c)}{\leq}\sum_{i=1}^{n}{H(Y_{1^{(1)}}[i])}+(k-1)\sum_{i=1}^{n}{H(Y_{1^{(1)}}[i]|V_{2^{(1)}}[i])}
\\&~~~+\sum_{i=1}^{n}{H(Y_{2^{(1)}}[i]|V_{2^{(1)}}[i],V_{1^{(1)}}[i])}+n\epsilon_n
\end{aligned}
\end{equation}
\normalsize
where $(a)$ follows from Theorem 1 of \cite{gcs} (in particular, inequality (6) in \cite{gcs} in the proof of Theorem 1) on the extended network shown in Fig. \ref{fig5}. Step $(b)$ follows from considering deterministic model. Finally, step $(c)$ follows from proof the presented in Appendix A.

By letting $k$ goes to infinity, we have
\small
\begin{equation} \label{eq0205}
\begin{aligned}
n(R_1) &\leq \sum_{i=1}^{n}{H(Y_{1^{(1)}}[i]|V_{2^{(1)}}[i])}+n\epsilon_n
\end{aligned}
\end{equation}
\normalsize
\begin{figure}[h]
\centering
\includegraphics[trim = 2in 0in 2.2in 0in, clip,width=0.4\textwidth]{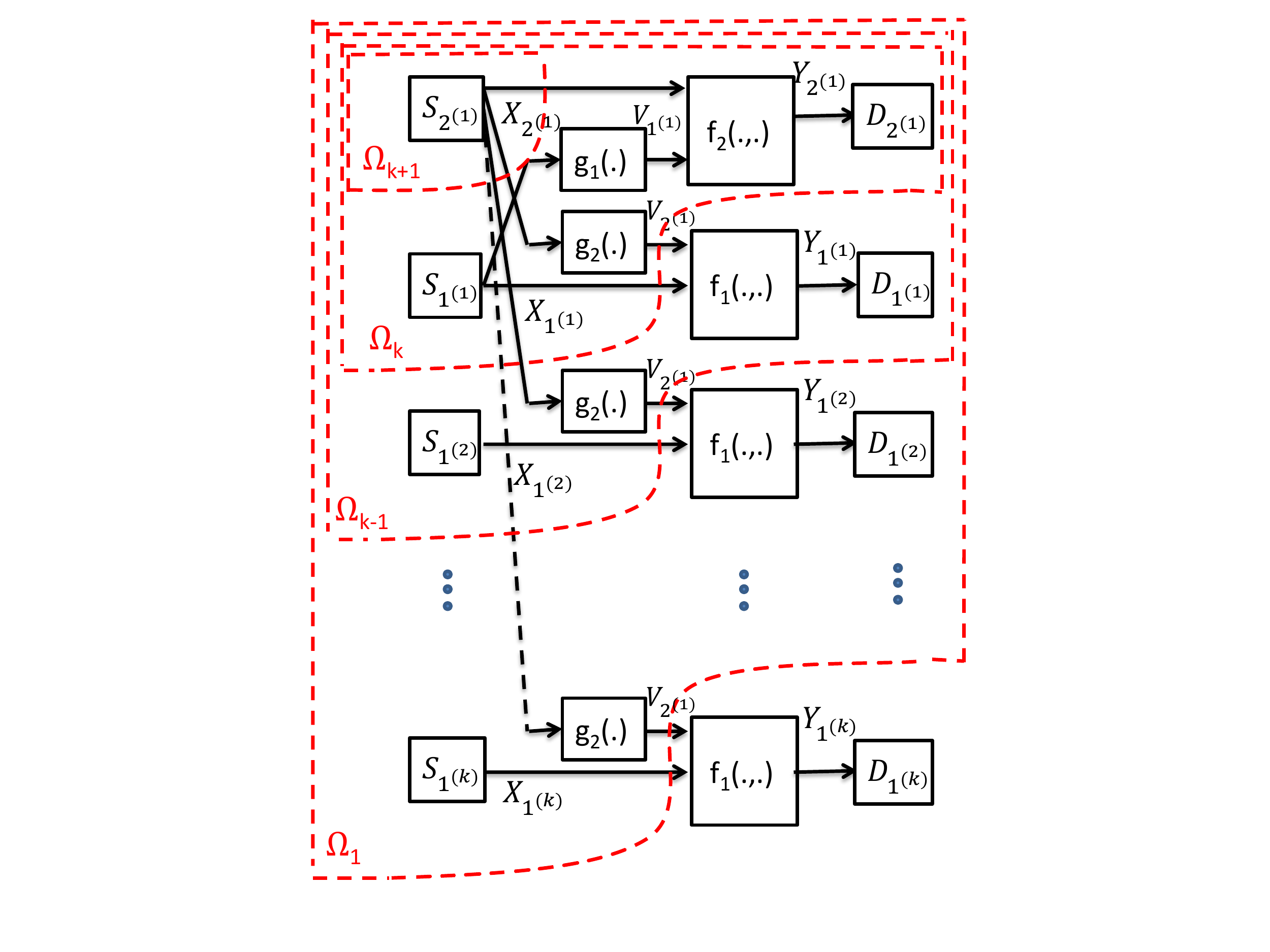}
\caption{Extended network and the cuts for deriving bound (4a) of the two-user DIC}
\label{fig5}
\end{figure}
which would result in bound (4a).

\subsection{Deriving Bound (4e)} \label{sec4_2}
In this section, we design the extended network $\mathcal N_{k,k}$ illustrated in Fig. \ref{fig6} to derive bound (4e). Based on Lemma \ref{lemext}, we know that if a rate tuple $(R_1,R_2)$ is achievable on $\mathcal N$, then rate tuple $({\underbrace{R_1,...,R_1}_{k}},{\underbrace{R_2,...,R_2}_{k}})$ is achievable on $\mathcal N_{k,k}$. Assume ${\{\mathcal C_n\}}$ as a sequence of coding scheme with block length $n$ that achieves sum rate $R_{\Sigma}$ on $\mathcal N_{k,k}$. Therefore, by applying GCS bound on the extended network $\mathcal N_{k,k}$ and picking cuts $\Omega_{1},...,\Omega_{2k}$ depicted in Fig. \ref{fig6}, we have
\small
\begin{equation} \label{eq0205}
\begin{aligned}
nR_{\Sigma}&=n(R_1+R_2+(k-1)(R_1+R_2))
\\& \overset{(a)}{\leq} I(W_{\mathcal S};Y_{{\Omega}_1^c}^n)\\&+\sum_{l=2}^{2k}{I(W_{\mathcal S\cap{{\Omega}_l}};Y_{{\Omega}_l^c\bigcap {\Omega}_{l-1}}^n|W_{\mathcal S\backslash{{\Omega}_l}},Y_{{\Omega}_{l-1}^c}^n)}+n\epsilon_n
\\&\overset{(b)}{=}H(Y_{{\Omega}_1^c}^n)+\sum_{l=2}^{2k}{H(Y_{{\Omega}_l^c\bigcap {\Omega}_{l-1}}^n|W_{\mathcal S\backslash{{\Omega}_l}},Y_{{\Omega}_{l-1}^c}^n)}+n\epsilon_n \nonumber
\end{aligned}
\end{equation}
\begin{equation} \label{eq343413}
\begin{aligned}
\\&\overset{(c)}{\leq} \sum_{i=1}^{n}{H(Y_{2^{(1)}}[i])}
\\&~~~+k\sum_{i=1}^{n}{H(Y_{1^{(1)}}[i]|V_{1^{(1)}}[i])}
\\&~~~+(k-1)\sum_{i=1}^{n}{H(Y_{2^{(1)}}[i]|V_{2^{(1)}}[i])}+n\epsilon_n
\end{aligned}
\end{equation}
\normalsize
where $(a)$ follows from Theorem 1 of \cite{gcs} (in particular, inequality (6) in \cite{gcs}) on the extended network depicted in Fig. \ref{fig6}. Step $(b)$ follows from considering deterministic model. Finally, step $(c)$ follow from the proof presented in Appendix A.

As $k$ goes to infinity, we have
\small
\begin{equation} \label{eq0205}
\begin{aligned}
n(R_1+R_2) &\leq \sum_{i=1}^{n}{H(Y_{1^{(1)}}[i]|V_{1^{(1)}}[i])}
\\&~~~+\sum_{i=1}^{n}{H(Y_{2^{(1)}}[i]|V_{2^{(1)}}[i])}+n\epsilon_n
\end{aligned}
\end{equation}
\normalsize
which would result in bound (4e).
\begin{figure}[h]
\centering
\includegraphics[trim = 2.5in 0in 2.2in 0in, clip,width=0.4\textwidth]{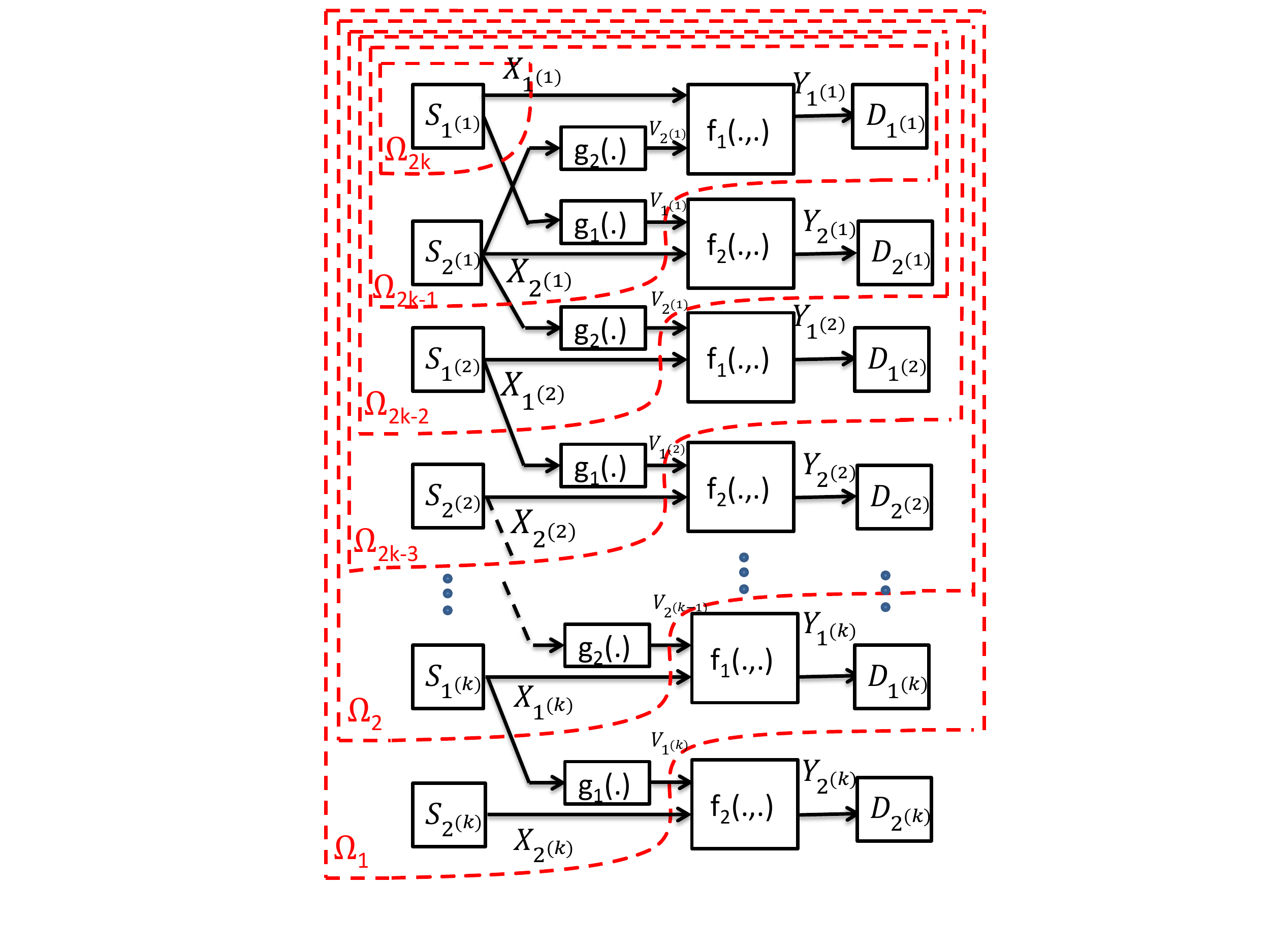}
\caption{Extended network and the cuts for deriving bound (4e) of the two-user DIC}
\label{fig6}
\end{figure}


\subsection{Deriving Bounds (4f) and (4g)} \label{sec4_3}
In this part, the extended network $\mathcal N_{2,1}$, depicted in Fig. \ref{fig7}, is utilized to derive bound (4f). By applying GCS bound on this extended network and picking cuts $\Omega_{1}$, $\Omega_{2}$ and $\Omega_{3}$ depicted in Fig. \ref{fig7}, we have
\small
\begin{equation} \label{eq0205}
\begin{aligned}
nR_{\Sigma}&=n(R_1+R_2+R_1)
\\& \overset{(a)}{\leq} I(W_{\mathcal S};Y_{{\Omega}_1^c}^n)+I(W_{\mathcal S\cap{{\Omega}_2}};Y_{{\Omega}_2^c\bigcap {\Omega}_1}^n|W_{\mathcal S\backslash{{\Omega}_2}},Y_{{\Omega}_1^c}^n)
\\&~~~+I(W_{\mathcal S\cap{{\Omega}_3}};Y_{{\Omega}_3^c\bigcap {\Omega}_2}^n|W_{\mathcal S\backslash{{\Omega}_3}},Y_{{\Omega}_2^c}^n)+n{\epsilon}_n
\\&\overset{(b)}{=}H(Y_{{\Omega}_1^c}^n)+H(Y_{{\Omega}_2^c\bigcap {\Omega}_1}^n|W_{\mathcal S\backslash{{\Omega}_2}},Y_{{\Omega}_1^c}^n)
\\&~~~+H(Y_{{\Omega}_3^c\bigcap {\Omega}_2}^n|W_{\mathcal S\backslash{{\Omega}_3}},Y_{{\Omega}_2^c}^n)+n\epsilon_n
\\& \overset{(c)}{\leq}\sum_{i=1}^{n}{H(Y_{1^{(1)}}[i])}+\sum_{i=1}^{n}{H(Y_{2^{(1)}}[i]|V_{2^{(1)}}[i])}
\\&~~~+\sum_{i=1}^{n}{H(Y_{1^{(1)}}[i]|V_{2^{(1)}}[i],V_{1^{(1)}}[i])}+n\epsilon_n
\end{aligned}
\end{equation}
\normalsize
where $(a)$ follows from Theorem 1 of \cite{gcs} (in particular, inequality (6) in \cite{gcs}) on the extended network depicted in Fig. \ref{fig7}. Step $(b)$ follows from considering deterministic model. Finally, step $(c)$ follow from the proof presented in Appendix A.
\begin{figure}[htb]
\centering
\begin{center}
\includegraphics[width=90 mm]{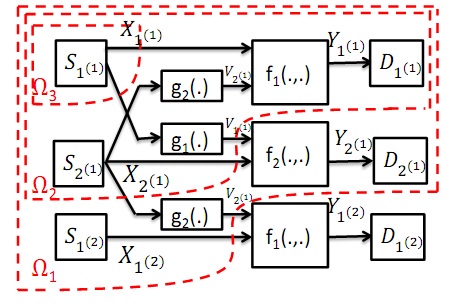} 
\caption{Extended network and the cuts for deriving bound (4f) of the two-user DIC}
\label{fig7}
\end{center}
\end{figure}

Bound (4g) can be derived by considering its corresponding extended network, i.e. $\mathcal N_{1,2}$, and following the similar procedure.

Now, we focus on obtaining the outer bounds of the symmetric three-user DIC.
\section{Extension to the Symmetric Three-user DIC} \label{sec5}
In this section, we consider the symmetric three-user DIC whose capacity region is found in \cite{jaafaar}. We demonstrate that all outer bounds can be systematically derived by the idea of the extended network and applying GCS bound to it. We start the derivation of outer bounds by defining the system model of the symmetric three-user DIC.
\subsection{System Model of the Symmetric Three-user DIC} \label{sub_5_1}
 The system model of the symmetric three-user DIC is illustrated in Fig. \ref{fig8}. In this section, the outer bounds of the symmetric three-user DIC are obtained by applying GCS on different extended networks. Since the same type of choosing cuts have been utilized, the procedure of finding outer bounds are similar to the case of finding the outer bounds of the two-user DIC. The received signals at destination 1, destination 2, and destination 3, i.e. $Y_1$, $Y_2$, and $Y_3$ are deterministic functions of transmitted signals $X_1$, $X_2$, $X_3$, and interference signals $V_1$, $V_2$ and $V_3$ as follow:
\begin{figure}[h10]
\centering
\includegraphics[trim = 1.5in 4.5in 2.5in 0in, clip,width=0.4\textwidth]{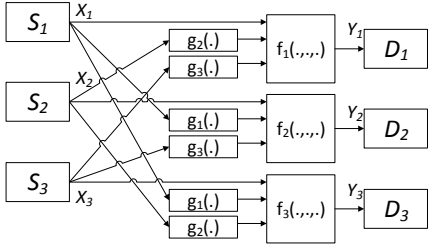}
\caption{System model of the symmetric three-user DIC}
\label{fig8}
\end{figure}
\small
\begin{equation} \label{eq0203}
\begin{aligned}
Y_1&=f_1(X_1,V_2,V_3)
\\Y_2&=f_2(X_2,V_1,V_3)
\\Y_3&=f_3(X_3,V_1,V_2)
\\V_1&=g_1(X_1)
\\V_2&=g_2(X_2)
\\V_3&=g_3(X_3)
\end{aligned}
\end{equation}
\normalsize
where $g_1(.)$, $g_2(.)$, and $g_3(.)$ are not generally invertible functions and $f_1(.,.)$, $f_2(.,.)$, and $f_3(.,.)$ satisfy the following conditions which is the same as \cite{jaafaar}
\small
\begin{equation} \label{eq0204}
\begin{aligned}
H(Y_1|X_1)&=H(V_2,V_3)
\\H(Y_2|X_2)&=H(V_1,V_3)
\\H(Y_3|X_3)&=H(V_2,V_3)
\end{aligned}
\end{equation}
\normalsize
these conditions are equivalent to
\small
\begin{equation} \label{eq0205}
\begin{aligned}
(V_2,V_3)&=h_1(X_1,Y_1)
\\(V_1,V_3)&=h_2(X_2,Y_2)
\\(V_1,V_2)&=h_3(X_3,Y_3)
\end{aligned}
\end{equation}
\normalsize

The capacity region of the symmetric three-user DIC is characterized in \cite{jaafaar}. It was shown by \cite{jaafaar} that the capacity region is the union of six regions each of which is characterized by 28 linear bounds as follow
\small
\begin{equation} \label{ineq1}
\begin{aligned}
R_1 &\leq H(Y_{1}|V_{2}V_{{3}})
\end{aligned}
\end{equation}
\begin{equation} \label{ineq2}
\begin{aligned}
R_1+R_2 &\leq H(Y_{1}|V_{1}V_{2}V_{{3}})+H(Y_{2}|V_{{3}})
\end{aligned}
\end{equation}
\begin{equation} \label{ineq3}
\begin{aligned}
R_1+R_2 &\leq H(Y_{1}|V_{1}V_{{3}})+H(Y_{2}|V_{2}V_{{3}})
\end{aligned}
\end{equation}
\begin{equation} \label{ineq4}
\begin{aligned}
2R_1+R_2 &\leq H(Y_{1}|V_{3})+H(Y_{1}|V_{1}V_{2}V_{{3}})+H(Y_{2}|V_{2}V_{{3}})
\end{aligned}
\end{equation}
\begin{equation} \label{ineq5}
\begin{aligned}
R_1+R_2+R_3 &\leq H(Y_{1}|V_{1})+H(Y_{2}|V_{2}V_{{3}})+H(Y_{3}|V_{1}V_{2}V_{{3}})
\end{aligned}
\end{equation}
\begin{equation} \label{ineq6}
\begin{aligned}
R_1+R_2+R_3 &\leq H(Y_{1}|V_{1}V_{3})+H(Y_{2}|V_{1}V_{{2}})+H(Y_{3}|V_{2}V_{{3}})
\end{aligned}
\end{equation}
\begin{equation} \label{ineq7}
\begin{aligned}
R_1+R_2+R_3 &\leq H(Y_{1}|V_{1}V_{2}V_{3})+H(Y_{2}|V_{1})+H(Y_{3}|V_{2}V_{{3}})
\end{aligned}
\end{equation}
\begin{equation} \label{ineq8}
\begin{aligned}
R_1+R_2+R_3 &\leq H(Y_{1}|V_{1}V_{2}V_{{3}})+H(Y_{2}|V_{1}V_{2}V_{{3}})+H(Y_{3})
\end{aligned}
\end{equation}
\begin{equation} \label{ineq9}
\begin{aligned}
2R_1+R_2+R_3 &\leq H(Y_{{1}})+H(Y_{1}|V_{1}V_{2}V_{{3}})+H(Y_{{2}}|V_{{2}}V_{{3}})\\&+H(Y_{3}|V_{1}V_{2}V_{{3}})
\end{aligned}
\end{equation}
\begin{equation} \label{ineq10}
\begin{aligned}
2R_1+R_2+R_3 &\leq H(Y_{1}|V_{1})+H(Y_{1}|V_{1}V_{2}V_{3})+H(Y_{2}|V_{2}V_{{3}})\\&+H(Y_{3}|V_{2}V_{{3}})
\end{aligned}
\end{equation}
\begin{equation} \label{ineq11}
\begin{aligned}
2R_1+R_2+R_3 &\leq H(Y_{1}|V_{1}V_{3})+H(Y_{1}|V_{2})+H(Y_{2}|V_{1}V_{2}V_{{3}})\\&+H(Y_{3}|V_{2}V_{{3}})
\end{aligned}
\end{equation}
\begin{equation} \label{ineq12}
\begin{aligned}
2R_1+R_2+R_3 &\leq H(Y_{1}|V_{1}V_{2}V_{3})+H(Y_{1}|V_{3})+H(Y_{2}|V_{1}V_{2})\\&+H(Y_{3}|V_{2}V_{{3}})
\end{aligned}
\end{equation}
\begin{equation} \label{ineq13}
\begin{aligned}
2R_1+R_2+R_3 &\leq H(Y_{1}|V_{1}V_{2}V_{3})+H(Y_{1}|V_{3})+H(Y_{2}|V_{2})\\&+H(Y_{3}|V_{1}V_{2}V_{{3}})
\end{aligned}
\end{equation}
\begin{equation} \label{ineq14}
\begin{aligned}
2R_1+R_2+R_3 &\leq H(Y_{1}|V_{1}V_{2}V_{3})+H(Y_{1}|V_{1}V_{2})+H(Y_{2}|V_{2}V_{3})\\&+H(Y_{3}|V_{{3}})
\end{aligned}
\end{equation}
\begin{equation} \label{ineq15}
\begin{aligned}
2R_1+R_2+R_3 &\leq 2H(Y_{1}|V_{1}V_{2}V_{{3}})+H(Y_{2})+H(Y_{3}|V_2V_3)
\end{aligned}
\end{equation}
\begin{equation} \label{ineq16}
\begin{aligned}
2R_1+R_2+R_3 &\leq 2H(Y_{1}|V_{1}V_{2}V_{3})+H(Y_{2}|V_{2})+H(Y_{3}|V_{{3}})
\end{aligned}
\end{equation}
\begin{equation} \label{ineq17}
\begin{aligned}
3R_1+R_2+R_3 &\leq 2H(Y_{1}|V_{1}V_{2}V_{{3}})+H(Y_1)+H(Y_{2}|V_2V_3)\\&+H(Y_{3}|V_2V_3)
\end{aligned}
\end{equation}
\begin{equation} \label{ineq18}
\begin{aligned}
3R_1+R_2+R_3 &\leq 2H(Y_{1}|V_{1}V_{2}V_{3})+H(Y_{1}|V_{2})+H(Y_{2}|V_{2}V_{3})\\&+H(Y_{3}|V_{{3}})
\end{aligned}
\end{equation}
\begin{equation} \label{ineq19}
\begin{aligned}
2R_1+2R_2+R_3 &\leq H(Y_1)+H(Y_1|V_1V_3)+2H(Y_2|V_1V_2V_3)\\&+H(Y_3|V_2V_3)
\end{aligned}
\end{equation}
\begin{equation} \label{ineq20}
\begin{aligned}
2R_1+2R_2+R_3 &\leq H(Y_1)+H(Y_{1}|V_{1}V_{2}V_{{3}})+H(Y_{2}|V_1V_2V_3)\\&+H(Y_{2}|V_2V_3)+H(Y_{3}|V_1V_3)
\end{aligned}
\end{equation}
\begin{equation} \label{ineq21}
\begin{aligned}
2R_1+2R_2+R_3 &\leq H(Y_1)+H(Y_{1}|V_{1}V_{2}V_{{3}})+2H(Y_{2}|V_1V_2V_3)\\&+H(Y_{3}|V_3)
\end{aligned}
\end{equation}
\begin{equation} \label{ineq22}
\begin{aligned}
2R_1+2R_2+R_3 &\leq H(Y_{1}|V_{1})+H(Y_{1}|V_{1}V_{2}V_{3})+H(Y_{2}|V_{1}V_{2}V_{3})\\&+H(Y_{2}|V_{2}V_{3})+H(Y_{3}|V_{{3}})
\end{aligned}
\end{equation}
\begin{equation} \label{ineq23}
\begin{aligned}
2R_1+2R_2+R_3 &\leq 2H(Y_{1}|V_{1}V_{3})+H(Y_{2}|V_{1}V_{2}V_{3})+H(Y_{2}|V_{2})\\&+H(Y_{3}|V_{{2}}V_{{3}})
\end{aligned}
\end{equation}
\begin{equation} \label{ineq24}
\begin{aligned}
3R_1+2R_2+R_3 &\leq 2H(Y_{1}|V_{1}V_{2}V_{{3}})+H(Y_1)+H(Y_{2}|V_1V_2V_3)\\&+H(Y_{2}|V_2V_3)+H(Y_{3}|V_3)
\end{aligned}
\end{equation}
\begin{equation} \label{ineq25}
\begin{aligned}
3R_1+2R_2+R_3 &\leq 2H(Y_{1}|V_{1}V_{2}V_{{3}})+H(Y_1)+2H(Y_{2}|V_2V_3)\\&+H(Y_{3}|V_1V_3)
\end{aligned}
\end{equation}
\begin{equation} \label{ineq26}
\begin{aligned}
3R_1+2R_2+R_3 &\leq 2H(Y_{1}|V_{1}V_{2}V_{3})+H(Y_{1}|V_{1})+2H(Y_{2}|V_{2}V_{3})\\&+H(Y_{3}|V_{{3}})
\end{aligned}
\end{equation}
\begin{equation} \label{ineq27}
\begin{aligned}
3R_1+2R_2+R_3 &\leq 3H(Y_{1}|V_{1}V_{2}V_{{3}})+H(Y_{2}|V_2V_3)+H(Y_{2})\\&+H(Y_{3}|V_3)
\end{aligned}
\end{equation}
\begin{equation} \label{ineq28}
\begin{aligned}
4R_1+2R_2+R_3 &\leq 3H(Y_{1}|V_{1}V_{2}V_{{3}})+H(Y_{1})+2H(Y_{2}|V_2V_3)\\&+H(Y_{3}|V_3)
\end{aligned}
\end{equation}
\normalsize

In the rest of this section, we prove bounds (\ref{ineq1})-(\ref{ineq28}) by constructing appropriate extended networks and applying GCS on them.
\\

\noindent \textbf{Derivation of bound (\ref{ineq1}).}
To derive this bound, we design the extended network considered in Fig. \ref{figb10}. By applying GCS bound and picking cuts depicted in Fig. \ref{figb10}, we have
\small
\begin{equation} \label{eq3434}
\begin{aligned}
nR_{\Sigma}&=n(R_1+R_2+R_3+kR_1)
\\& \overset{(a)}{\leq}H(Y_{{\Omega}_1^c}^n)+\sum_{l=2}^{k+2}{H(Y_{{\Omega}_l^c\bigcap {\Omega}_{l-1}}^n|W_{\mathcal S\backslash{{\Omega}_l}},Y_{{\Omega}_{l-1}^c}^n)}+n\epsilon_n
\\& \overset{(b)}{\leq}\sum_{i=1}^{n}{kH(Y_{1}[i]|V_{2}[i]V_{3}[i])}+C+n\epsilon_n
\end{aligned}
\end{equation}
\normalsize
where $C\triangleq \sum_{i=1}^{n}{H(Y_{1}[i])+H(Y_{2}[i])+H(Y_{3}[i])}$ and step $(a)$ follows from GCS \cite{gcs} for deterministic networks. Step $(b)$ follows from the proof presented in Appendix B.

By letting $k$ goes to infinity, bound (\ref{ineq1}) will be recovered.
\begin{figure}[h10]
\centering
\includegraphics[trim = 1.5in 5.5in 2.5in 5.5in, clip,width=0.4\textwidth]{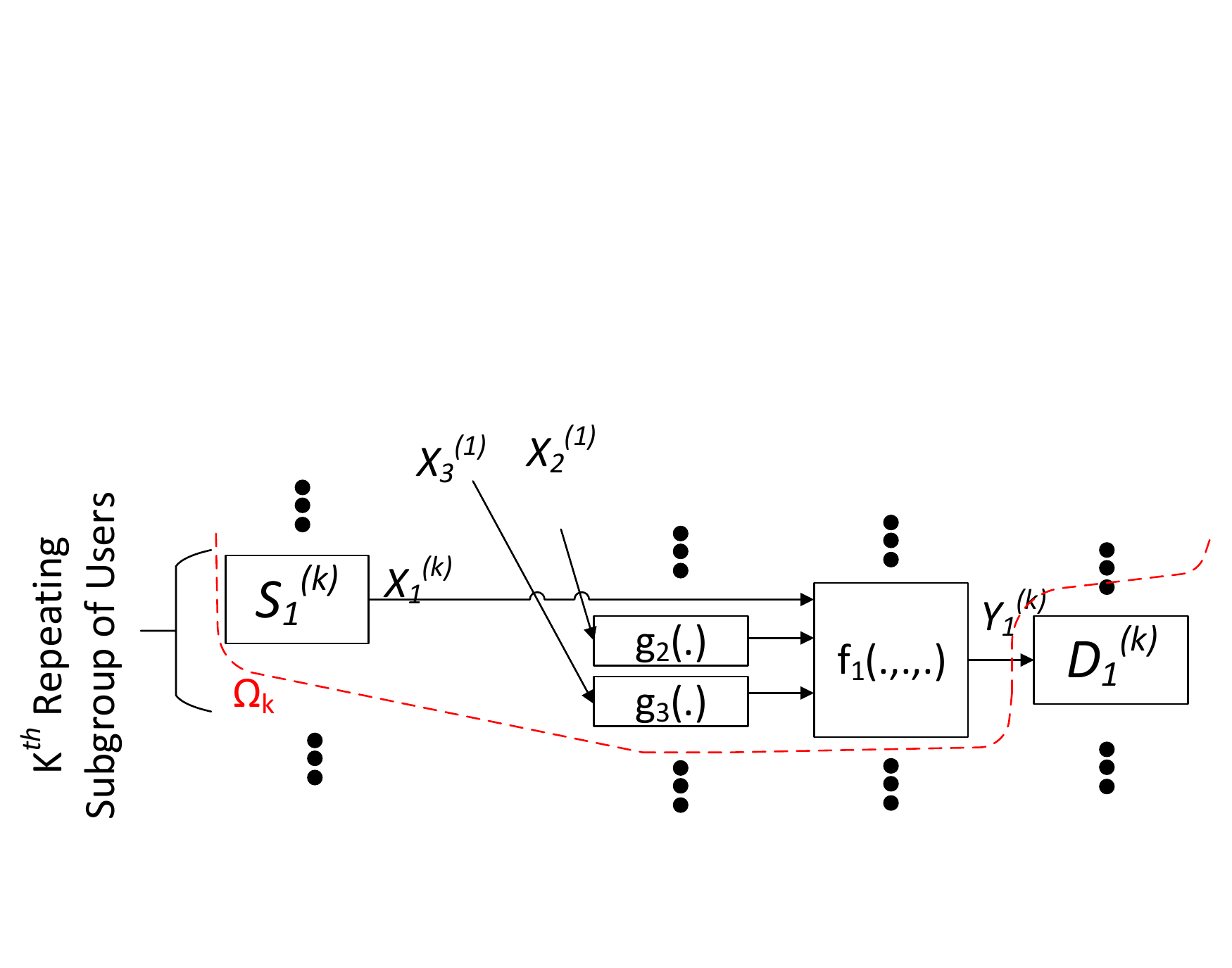}
\caption{Extended network and the cuts for deriving bound (\ref{ineq1}) for the three-user DIC}
\label{figb10}
\end{figure}
\\
\noindent \textbf{Derivation of bound (\ref{ineq2}).}
To derive this bound, we design the extended network considered in Fig. \ref{figb11}. By applying GCS bound and picking cuts depicted in Fig. \ref{figb11}, we have
\small
\begin{equation} \label{eq3434}
\begin{aligned}
nR_{\Sigma}&=n(R_1+R_2+R_3+kR_1+kR_2)
\\& \overset{(a)}{\leq}H(Y_{{\Omega}_1^c}^n)+\sum_{l=2}^{2k+2}{H(Y_{{\Omega}_l^c\bigcap {\Omega}_{l-1}}^n|W_{\mathcal S\backslash{{\Omega}_l}},Y_{{\Omega}_{l-1}^c}^n)}+n\epsilon_n
\\& \overset{(b)}{\leq}\sum_{i=1}^{n}{kH(Y_{1}[i]|V_{1}[i]V_{2}[i]V_{3}[i])+kH(Y_{2}[i]|V_{3}[i])}\\&+C+n\epsilon_n
\end{aligned}
\end{equation}
\normalsize
where $(a)$ follows from GCS \cite{gcs} for deterministic networks and step $(b)$ follows from the proof presented in Appendix B.

By letting $k$ goes to infinity, bound (\ref{ineq2}) will be recovered.
\begin{figure}[h10]
\centering
\includegraphics[trim = 1.5in 5in 2.5in 0in, clip,width=0.4\textwidth]{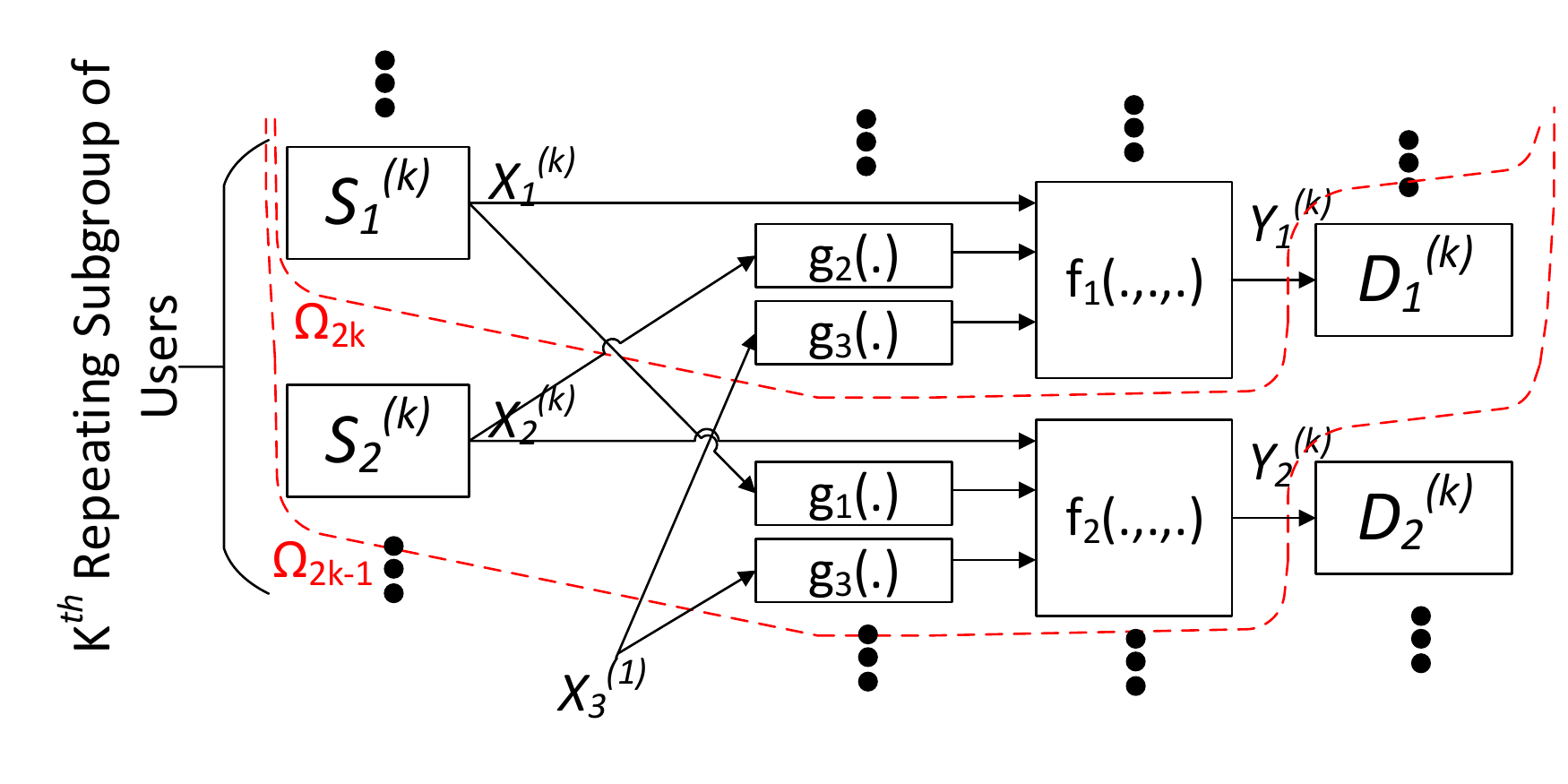}
\caption{Extended network and the cuts for deriving bound (\ref{ineq2}) for the three-user DIC}
\label{figb11}
\end{figure}
\\
\noindent \textbf{Derivation of bound (\ref{ineq3}).}
To derive this bound, we design the extended network considered in Fig. \ref{figb12}. By applying GCS bound and picking cuts depicted in Fig. \ref{figb12}, we have
\small
\begin{equation} \label{eq3434}
\begin{aligned}
nR_{\Sigma}&=n(R_1+R_2+R_3+kR_1+kR_2)
\\& \overset{(a)}{\leq}H(Y_{{\Omega}_1^c}^n)+\sum_{l=2}^{2k+2}{H(Y_{{\Omega}_l^c\bigcap {\Omega}_{l-1}}^n|W_{\mathcal S\backslash{{\Omega}_l}},Y_{{\Omega}_{l-1}^c}^n)}+n\epsilon_n
\\& \overset{(b)}{\leq}\sum_{i=1}^{n}{kH(Y_{1}[i]|V_{1}[i]V_{3}[i])+kH(Y_{2}[i]|V_{2}[i]V_{3}[i])}\\&+C+n\epsilon_n
\end{aligned}
\end{equation}
\normalsize
where $(a)$ follows from GCS \cite{gcs} for deterministic networks and step $(b)$ follows from the proof presented in Appendix B.

By letting $k$ goes to infinity, bound (\ref{ineq3}) will be recovered.
\begin{figure}[h10]
\centering
\includegraphics[trim =1in 4in 2.5in 6in, clip,width=0.4\textwidth]{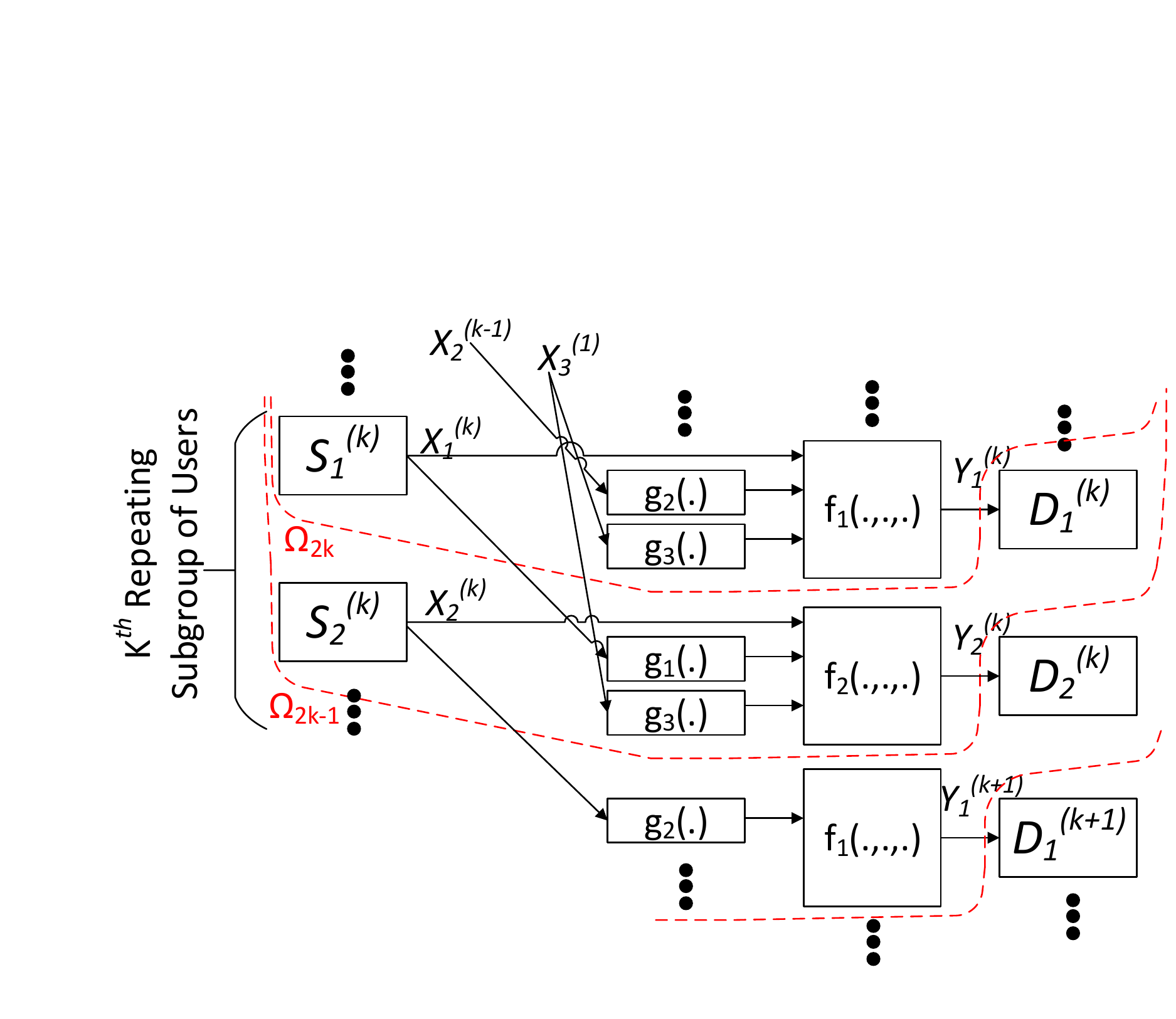}
\caption{Extended network and the cuts for deriving bound (\ref{ineq3}) for the three-user DIC}
\label{figb12}
\end{figure}
\\
\noindent \textbf{Derivation of bound (\ref{ineq4}).}
To derive this bound, we design the extended network considered in Fig. \ref{figb13}. By applying GCS bound and picking cuts depicted in Fig. \ref{figb13}, we have
\small
\begin{equation} \label{eq3434}
\begin{aligned}
nR_{\Sigma}&=n(R_1+R_2+R_3+2kR_1+kR_2)
\\& \overset{(a)}{\leq}H(Y_{{\Omega}_1^c}^n)+\sum_{l=2}^{3k+2}{H(Y_{{\Omega}_l^c\bigcap {\Omega}_{l-1}}^n|W_{\mathcal S\backslash{{\Omega}_l}},Y_{{\Omega}_{l-1}^c}^n)}+n\epsilon_n
\\& \overset{(b)}{\leq}\sum_{i=1}^{n}{kH(Y_{1}[i]|V_{3}[i])+kH(Y_{1}[i]|V_{1}[i]V_{2}[i]V_{3}[i])}\\&+\sum_{i=1}^{n}{kH(Y_{2}[i]|V_{2}[i]V_{3}[i])}\\&+C+n\epsilon_n
\end{aligned}
\end{equation}
\normalsize
where $(a)$ follows from GCS \cite{gcs} for deterministic networks and step $(b)$ follows from the proof presented in Appendix B.

By letting $k$ goes to infinity, bound (\ref{ineq4}) will be recovered.
\begin{figure}[h10]
\centering
\includegraphics[trim = .7in 1.5in 2.5in 0in, clip,width=0.4\textwidth]{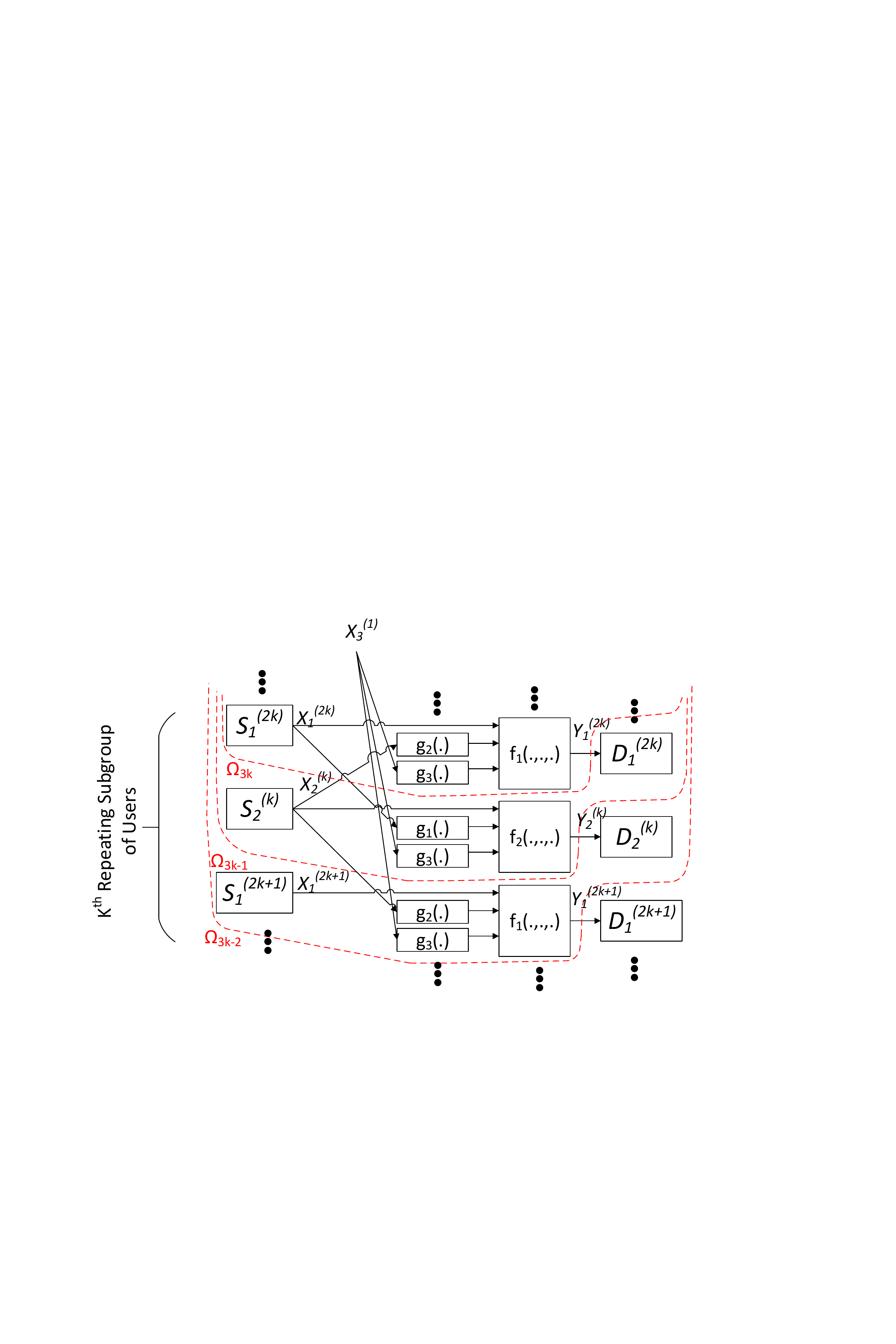}
\caption{Extended network and the cuts for deriving bound (\ref{ineq4}) for the three-user DIC}
\label{figb13}
\end{figure}
\\
\noindent \textbf{Derivation of bound (\ref{ineq5}).}
To derive this bound, we design the extended network considered in Fig. \ref{figb14}. By applying GCS bound and picking cuts depicted in Fig. \ref{figb14}, we have
\small
\begin{equation} \label{eq3434}
\begin{aligned}
nR_{\Sigma}&=n(R_1+R_2+R_3+kR_1+kR_2+kR_3)
\\& \overset{(a)}{\leq}H(Y_{{\Omega}_1^c}^n)+\sum_{l=2}^{3k+2}{H(Y_{{\Omega}_l^c\bigcap {\Omega}_{l-1}}^n|W_{\mathcal S\backslash{{\Omega}_l}},Y_{{\Omega}_{l-1}^c}^n)}+n\epsilon_n
\\& \overset{(b)}{\leq}\sum_{i=1}^{n}{kH(Y_{1}[i]|V_{1}[i])+kH(Y_{2}[i]|V_{2}[i]V_{3}[i])}\\&+\sum_{i=1}^{n}{kH(Y_{3}[i]|V_{1}[i]V_{2}[i]V_{3}[i])}\\&+C+n\epsilon_n
\end{aligned}
\end{equation}
\normalsize
where $(a)$ follows from GCS \cite{gcs} for deterministic networks and step $(b)$ follows from the proof presented in Appendix B.

By letting $k$ goes to infinity, bound (\ref{ineq5}) will be recovered.
\begin{figure}[h10]
\centering
\includegraphics[trim = .7in 2.5in 2.5in 6in, clip,width=0.4\textwidth]{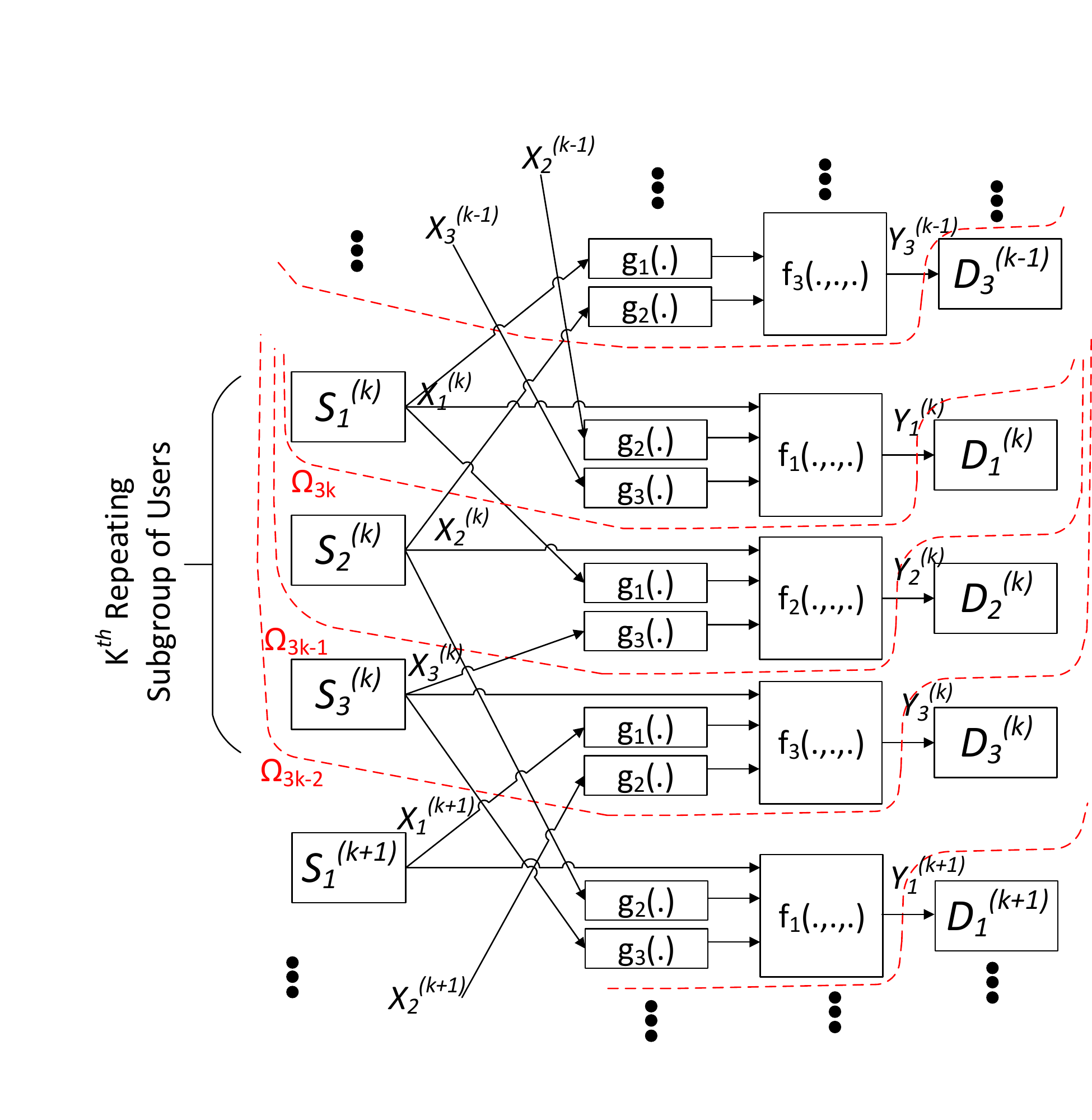}
\caption{Extended network and the cuts for deriving bound (\ref{ineq5}) for the three-user DIC}
\label{figb14}
\end{figure}
\\
\noindent \textbf{Derivation of bound (\ref{ineq6}).}
To derive this bound, we design the extended network considered in Fig. \ref{figb15}. By applying GCS bound and picking cuts depicted in Fig. \ref{figb15}, we have
\small
\begin{equation} \label{eq3434}
\begin{aligned}
nR_{\Sigma}&=n(R_1+R_2+R_3+kR_1+kR_2+kR_3)
\\& \overset{(a)}{\leq}H(Y_{{\Omega}_1^c}^n)+\sum_{l=2}^{3k+2}{H(Y_{{\Omega}_l^c\bigcap {\Omega}_{l-1}}^n|W_{\mathcal S\backslash{{\Omega}_l}},Y_{{\Omega}_{l-1}^c}^n)}+n\epsilon_n
\\& \overset{(b)}{\leq}\sum_{i=1}^{n}{kH(Y_{1}[i]|V_{1}[i]V_{3}[i])+kH(Y_{2}[i]|V_{2}[i]V_{1}[i])}\nonumber
\end{aligned}
\end{equation}
\begin{equation} \label{eq3434}
\begin{aligned}
&+\sum_{i=1}^{n}{kH(Y_{3}[i]|V_{2}[i],V_{3}[i])}+C+n\epsilon_n
\end{aligned}
\end{equation}
\normalsize
where $(a)$ follows from GCS \cite{gcs} for deterministic networks and step $(b)$ follows from the proof presented in Appendix B.

By letting $k$ goes to infinity, bound (\ref{ineq6}) will be recovered.
\begin{figure}[h10]
\centering
\includegraphics[trim = .7in 2.5in 2.5in 6in, clip,width=0.4\textwidth]{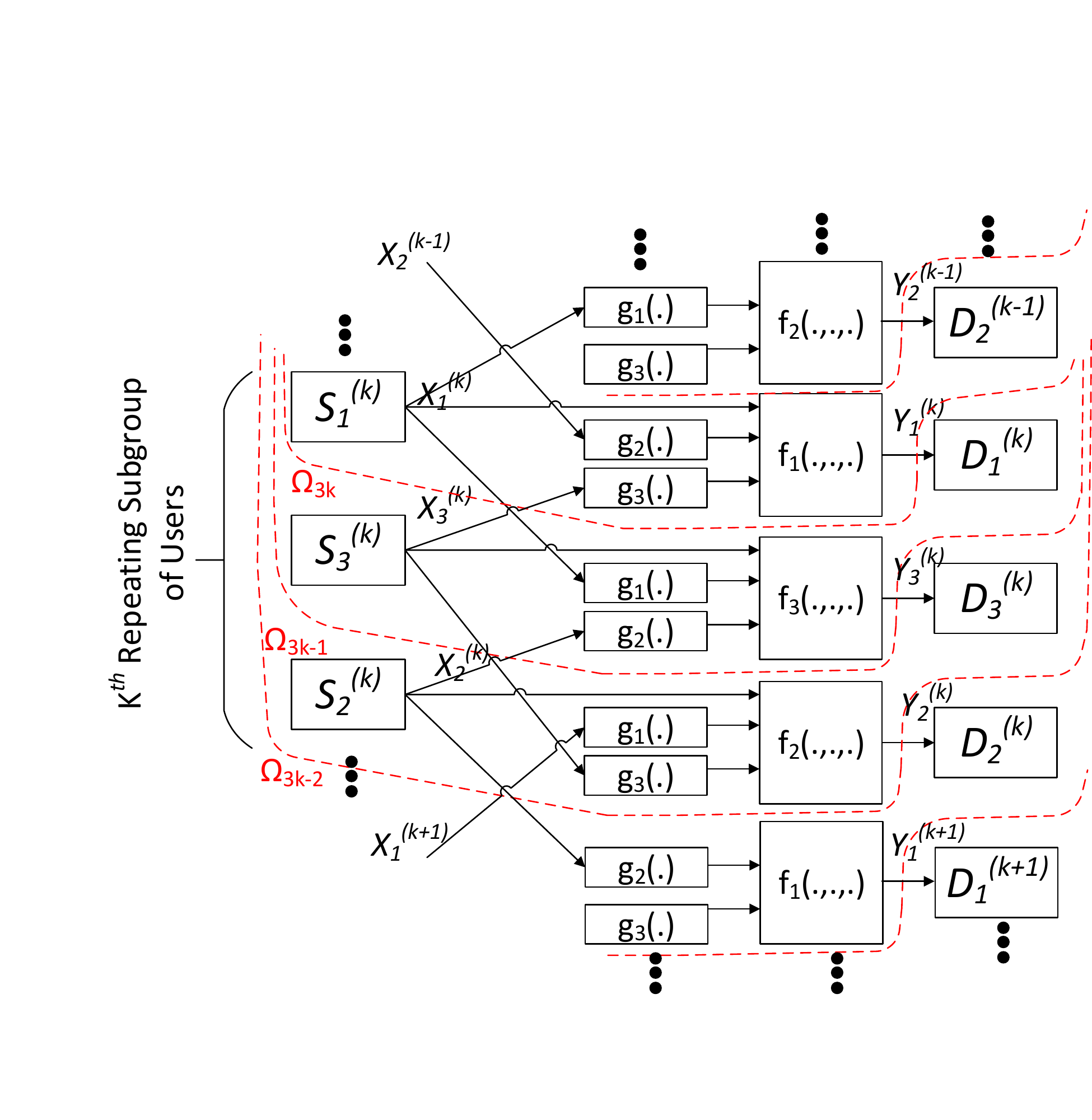}
\caption{Extended network and the cuts for deriving bound (\ref{ineq6}) for the three-user DIC}
\label{figb15}
\end{figure}
\\
\noindent \textbf{Derivation of bound (\ref{ineq7}).}
To derive this bound, we design the extended network considered in Fig. \ref{figb16}. By applying GCS bound and picking cuts depicted in Fig. \ref{figb16}, we have
\small
\begin{equation} \label{eq3434}
\begin{aligned}
nR_{\Sigma}&=n(R_1+R_2+R_3+kR_1+kR_2+kR_3)
\\& \overset{(a)}{\leq} H(Y_{{\Omega}_1^c}^n)+\sum_{l=2}^{3k+2}{H(Y_{{\Omega}_l^c\bigcap {\Omega}_{l-1}}^n|W_{\mathcal S\backslash{{\Omega}_l}},Y_{{\Omega}_{l-1}^c}^n)}+n\epsilon_n
\\& \overset{(b)}{\leq}\sum_{i=1}^{n}{kH(Y_{1}[i]|V_{1}[i]V_{2}[i]V_{3}[i])+kH(Y_{2}[i]|V_{1}[i])}\nonumber
\end{aligned}
\end{equation}
\begin{equation} \label{eq3434}
\begin{aligned}
&+\sum_{i=1}^{n}{kH(Y_{3}[i]|V_{2}[i]V_{3}[i])}+C+n\epsilon_n
\end{aligned}
\end{equation}
\normalsize
where $(a)$ follows from GCS \cite{gcs} for deterministic networks and step $(b)$ follows from the proof presented in Appendix B.

By letting $k$ goes to infinity, bound (\ref{ineq7}) will be recovered.
\begin{figure}[h10]
\centering
\includegraphics[trim = .5in 2.5in 2.5in 0in, clip,width=0.4\textwidth]{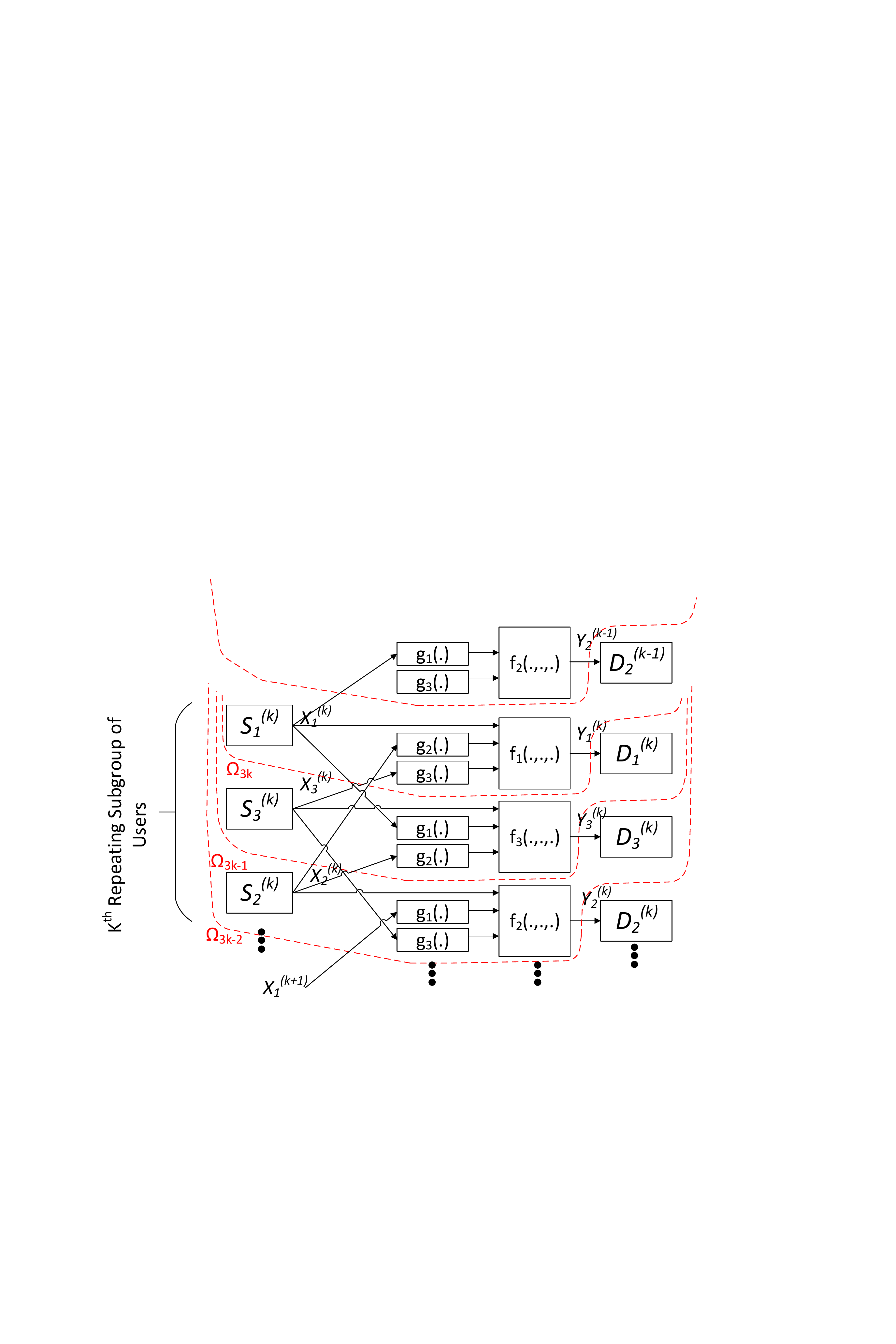}
\caption{Extended network and the cuts for deriving bound (\ref{ineq7}) for the three-user DIC}
\label{figb16}
\end{figure}
\\
\noindent \textbf{Derivation of bound (\ref{ineq8}).}
To derive this bound, we design the original network considered in Fig. \ref{figb17}. By applying GCS bound and picking cuts $\Omega_{1}, \Omega_{2}$, and $\Omega_{3}$ depicted in Fig. \ref{figb17}, we have
\small
\begin{equation} \label{eq3434}
\begin{aligned}
nR_{\Sigma}&=n(R_1+R_2+R_3)
\\& \overset{(a)}{\leq} H(Y_{{\Omega}_1^c}^n)+\sum_{l=2}^{3}{H(Y_{{\Omega}_l^c\bigcap {\Omega}_{l-1}}^n|W_{\mathcal S\backslash{{\Omega}_l}},Y_{{\Omega}_{l-1}^c}^n)}+n\epsilon_n
\\& \overset{(c)}{\leq}\sum_{i=1}^{n}{H(Y_{1}[i]|V_{1}[i]V_{2}[i]V_{3}[i])}+\sum_{i=1}^{n}{H(Y_{2}[i]|V_{1}[i]V_{2}[i]V_{3}[i])}\\&+\sum_{i=1}^{n}{H(Y_{3}[i])}+n\epsilon_n
\end{aligned}
\end{equation}
\normalsize
where $(a)$ follows from GCS \cite{gcs} for deterministic networks and step $(b)$ follows from the proof presented in Appendix B.
\begin{figure}[h10]
\centering
\includegraphics[trim = 1.5in 4.5in 2.5in 0in, clip,width=0.4\textwidth]{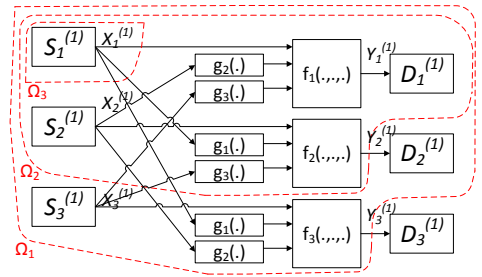}
\caption{Original network and the cuts for deriving bound (\ref{ineq8}) for the three-user DIC}
\label{figb17}
\end{figure}

\noindent \textbf{Derivation of bound (\ref{ineq9}).}
To derive this bound, we design the extended network considered in Fig. \ref{fig9}. Consider ${\{\mathcal C_n\}}$ as a sequence of coding scheme with block length $n$ that achieves sum rate $R_{\Sigma}$ on this extended network. Therefore, by applying GCS bound on this extended network and picking cuts $\Omega_{1}, \Omega_{2}$, $\Omega_{3}$, and $\Omega_{4}$ depicted in Fig. \ref{fig9}, we have
\small
\begin{equation} \label{eq3434}
\begin{aligned}
nR_{\Sigma}&=n(R_1+R_2+R_3+R_1)
\\& \overset{(a)}{\leq} H(Y_{{\Omega}_1^c}^n)+\sum_{l=2}^{4}{H(Y_{{\Omega}_l^c\bigcap {\Omega}_{l-1}}^n|W_{\mathcal S\backslash{{\Omega}_l}},Y_{{\Omega}_{l-1}^c}^n)}+n\epsilon_n
\\& \overset{(b)}{\leq}\sum_{i=1}^{n}{H(Y_{1}[i])}+\sum_{i=1}^{n}{H(Y_{1}[i]|V_{1}[i]V_{2}[i]V_{3}[i])}\\&+\sum_{i=1}^{n}{H(Y_{2}[i]|V_{2}[i]V_{3}[i])}
\\&+\sum_{i=1}^{n}{H(Y_{3}[i]|V_{1}[i]V_{2}[i]V_{3}[i])}+n\epsilon_n
\end{aligned}
\end{equation}
\normalsize
where $(a)$ follows from GCS \cite{gcs} for deterministic networks and step $(b)$ follows from the proof presented in Appendix B.

\begin{figure}[h17]
\centering
\includegraphics[trim = 2in 3.5in 1.2in 0in, clip,width=0.4\textwidth]{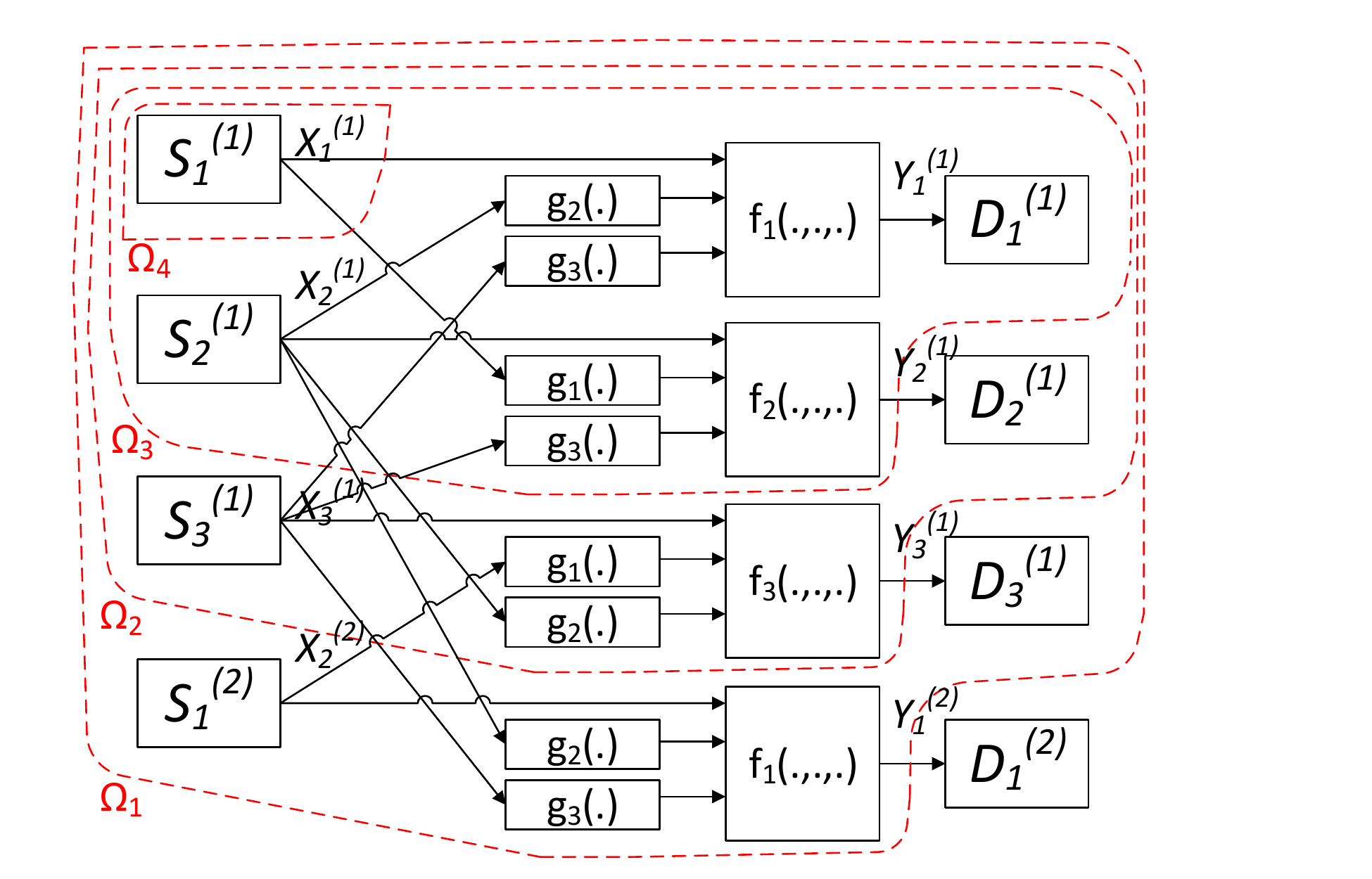}
\caption{Extended network and the cuts for deriving bound (\ref{ineq9}) for the three-user DIC}
\label{fig9}
\end{figure}

\noindent \textbf{Derivation of bound (\ref{ineq10}).}
To derive this bound, we design the extended network considered in Fig. \ref{figb19}. By applying GCS bound and picking cuts depicted in Fig. \ref{figb19}, we have
\small
\begin{equation} \label{eq3434}
\begin{aligned}
nR_{\Sigma}&=n(R_1+R_2+R_3+2kR_1+kR_2+kR_3)
\\& \overset{(a)}{\leq} H(Y_{{\Omega}_1^c}^n)+\sum_{l=2}^{4k+2}{H(Y_{{\Omega}_l^c\bigcap {\Omega}_{l-1}}^n|W_{\mathcal S\backslash{{\Omega}_l}},Y_{{\Omega}_{l-1}^c}^n)}+n\epsilon_n
\\& \overset{(b)}{\leq}\sum_{i=1}^{n}{kH(Y_{1}[i]|V_{1}[i])+kH(Y_{1}[i]|V_{1}[i]V_{2}[i]V_{3}[i])}\\&+\sum_{i=1}^{n}{kH(Y_{2}[i]|V_{2}[i]V_{3}[i])+kH(Y_{3}[i]|V_{2}[i]V_{3}[i])}
\\&+C+n\epsilon_n
\end{aligned}
\end{equation}
\normalsize
where $(a)$ follows from GCS \cite{gcs} for deterministic networks and step $(b)$ follows from the proof presented in Appendix B.

By letting $k$ goes to infinity, bound (\ref{ineq10}) will be recovered.
\begin{figure}[h10]
\centering
\includegraphics[trim = 0.5in 1.5in 2.5in 6in, clip,width=0.4\textwidth]{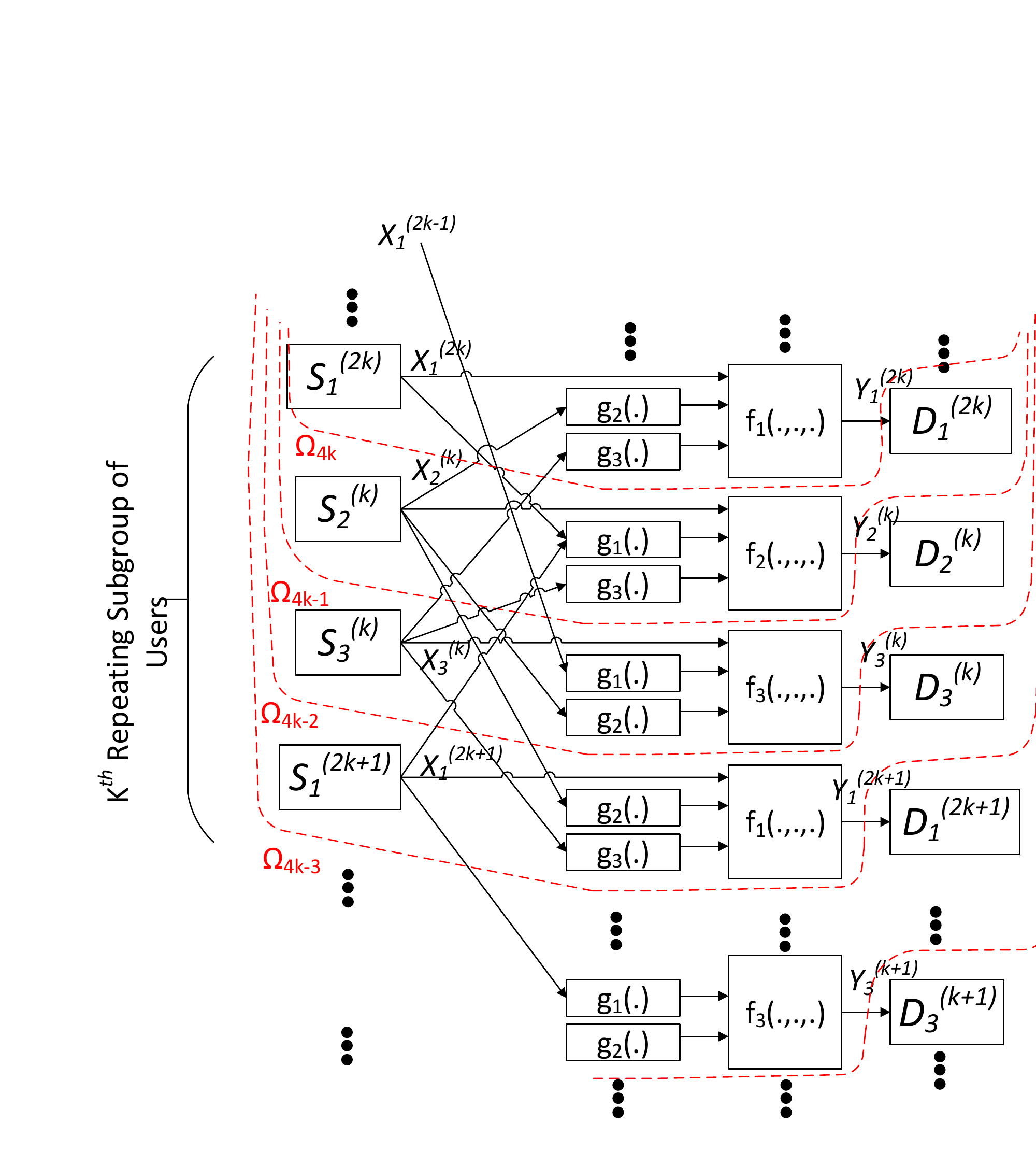}
\caption{Extended network and the cuts for deriving bound (\ref{ineq10}) for the three-user DIC}
\label{figb19}
\end{figure}

\noindent \textbf{Derivation of bound (\ref{ineq11}).}
To derive this bound, we design the extended network considered in Fig. \ref{figb20}. By applying GCS bound and picking cuts depicted in Fig. \ref{figb20}, we have
\small
\begin{equation} \label{eq3434}
\begin{aligned}
nR_{\Sigma}&=n(R_1+R_2+R_3+2kR_1+kR_2+kR_3)
\\& \overset{(a)}{\leq} H(Y_{{\Omega}_1^c}^n)+\sum_{l=2}^{4k+2}{H(Y_{{\Omega}_l^c\bigcap {\Omega}_{l-1}}^n|W_{\mathcal S\backslash{{\Omega}_l}},Y_{{\Omega}_{l-1}^c}^n)}+n\epsilon_n
\\& \overset{(b)}{\leq}\sum_{i=1}^{n}{kH(Y_{1}[i]|V_{1}[i]V_{3}[i])+kH(Y_{1}[i]|V_{2}[i])}\\&+\sum_{i=1}^{n}{kH(Y_{2}[i]|V_{1}[i]V_{2}[i]V_{3}[i])+kH(Y_{3}[i]|V_{2}[i]V_{3}[i])}
\\&+C+n\epsilon_n
\end{aligned}
\end{equation}
\normalsize
where $(a)$ follows from GCS \cite{gcs} for deterministic networks and step $(b)$ follows from the proof presented in Appendix B.

By letting $k$ goes to infinity, bound (\ref{ineq11}) will be recovered.
\begin{figure}[h10]
\centering
\includegraphics[trim = 0.5in 1.5in 2.5in 6in, clip,width=0.4\textwidth]{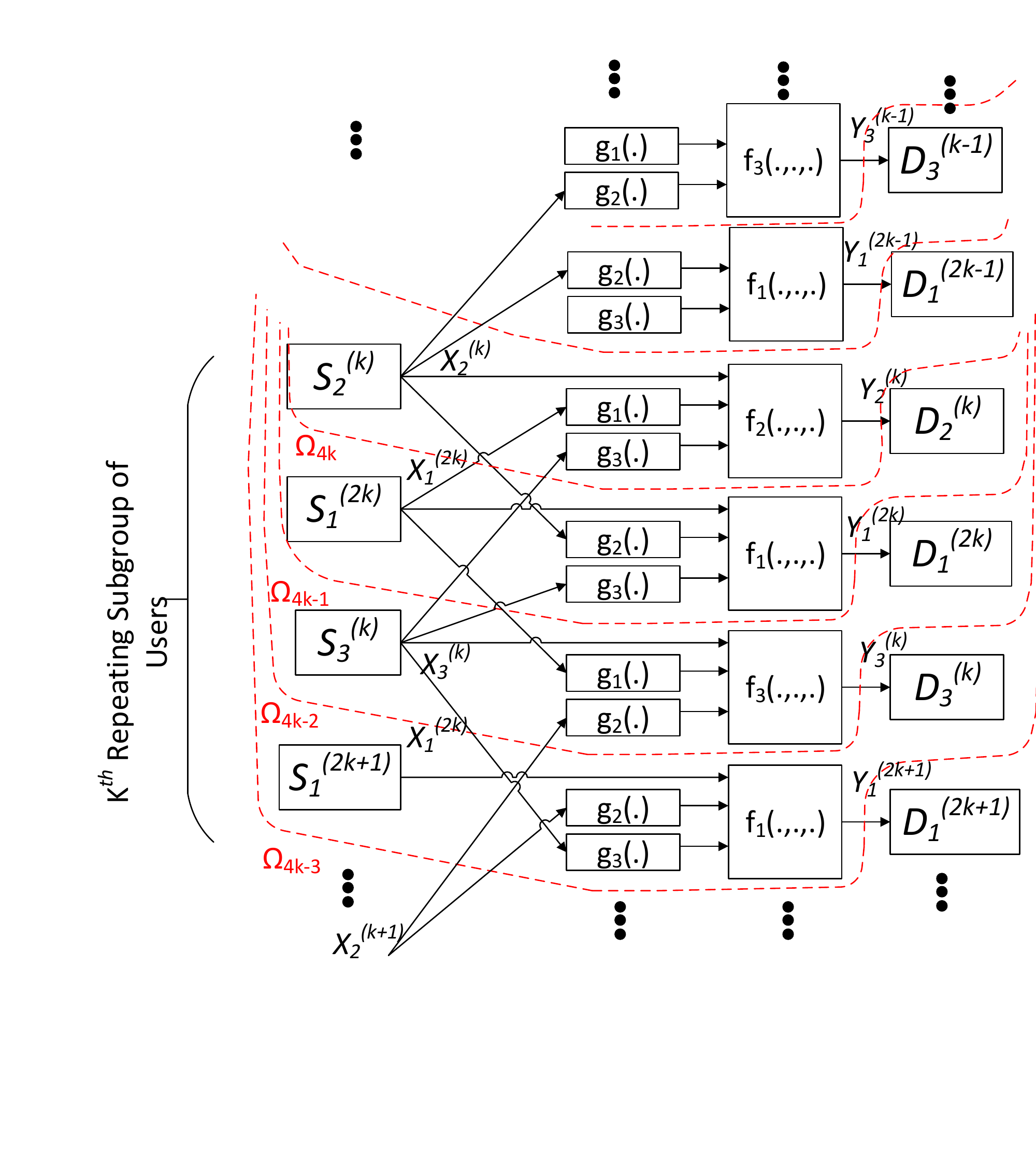}
\caption{Extended network and the cuts for deriving bound (\ref{ineq11}) for the three-user DIC}
\label{figb20}
\end{figure}

\noindent \textbf{Derivation of bound (\ref{ineq12}).}
To derive this bound, we design the extended network considered in Fig. \ref{figb21}. By applying GCS bound and picking cuts depicted in Fig. \ref{figb21}, we have
\small
\begin{equation} \label{eq3434}
\begin{aligned}
nR_{\Sigma}&=n(R_1+R_2+R_3+2kR_1+kR_2+kR_3)
\\& \overset{(a)}{\leq} H(Y_{{\Omega}_1^c}^n)+\sum_{l=2}^{4k+2}{H(Y_{{\Omega}_l^c\bigcap {\Omega}_{l-1}}^n|W_{\mathcal S\backslash{{\Omega}_l}},Y_{{\Omega}_{l-1}^c}^n)}+n\epsilon_n
\\& \overset{(c)}{\leq}\sum_{i=1}^{n}{kH(Y_{1}[i]|V_{1}[i]V_{2}[i]V_{3}[i])+kH(Y_{1}[i]|V_{3}[i])}\\&+\sum_{i=1}^{n}{kH(Y_{2}[i]|V_{1}[i]V_{2}[i])+kH(Y_{3}[i]|V_{2}[i]V_{3}[i])}
\\&+C+n\epsilon_n
\end{aligned}
\end{equation}
\normalsize
where $(a)$ follows from GCS \cite{gcs} for deterministic networks and step $(b)$ follows from the proof presented in Appendix B.

By letting $k$ goes to infinity, bound (\ref{ineq12}) will be recovered.
\begin{figure}[h10]
\centering
\includegraphics[trim = .5in 1.5in 2.5in 6in, clip,width=0.4\textwidth]{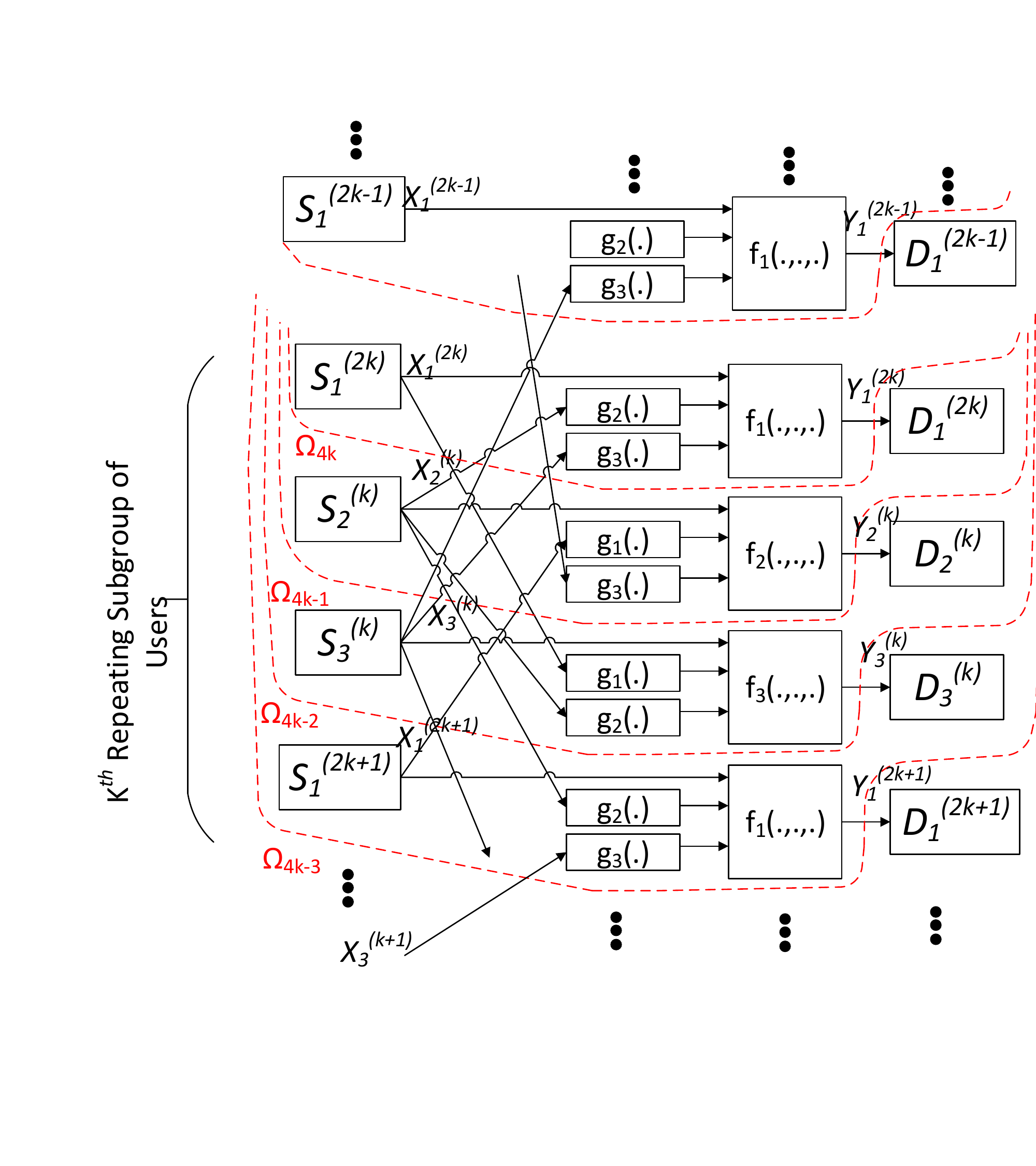}
\caption{Extended network and the cuts for deriving bound (\ref{ineq12}) for the three-user DIC}
\label{figb21}
\end{figure}

\noindent \textbf{Derivation of bound (\ref{ineq13}).}
To derive this bound, we design the extended network considered in Fig. \ref{figb22}. By applying GCS bound and picking cuts depicted in Fig. \ref{figb22}, we have
\small
\begin{equation} \label{eq3434}
\begin{aligned}
nR_{\Sigma}&=n(R_1+R_2+R_3+2kR_1+kR_2+kR_3)
\\& \overset{(a)}{\leq} H(Y_{{\Omega}_1^c}^n)+\sum_{l=2}^{4k+2}{H(Y_{{\Omega}_l^c\bigcap {\Omega}_{l-1}}^n|W_{\mathcal S\backslash{{\Omega}_l}},Y_{{\Omega}_{l-1}^c}^n)}+n\epsilon_n
\\& \overset{(b)}{\leq}\sum_{i=1}^{n}{kH(Y_{1}[i]|V_{1}[i]V_{2}[i]V_{3}[i])+kH(Y_{1}[i]|V_{3}[i])}\\&+\sum_{i=1}^{n}{kH(Y_{2}[i]|V_{2}[i])+kH(Y_{3}[i]|V_{1}[i]V_{2}[i]V_{3}[i])}
\\&+C+n\epsilon_n
\end{aligned}
\end{equation}
\normalsize
where $(a)$ follows from GCS \cite{gcs} for deterministic networks and step $(b)$ follows from the proof presented in Appendix B.

By letting $k$ goes to infinity, bound (\ref{ineq13}) will be recovered.
\begin{figure}[h10]
\centering
\includegraphics[trim = .5in 1.5in 2.5in 6in, clip,width=0.4\textwidth]{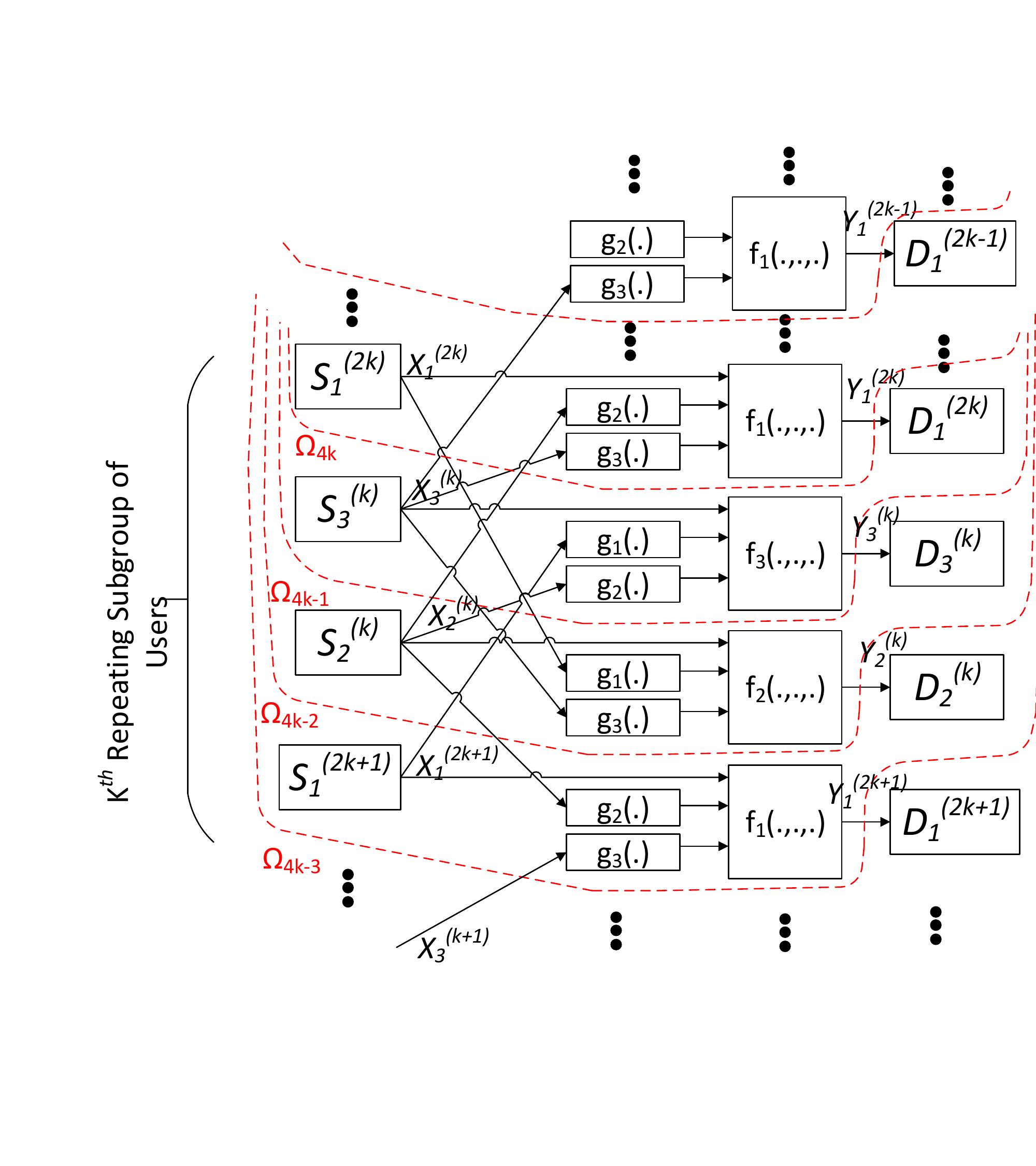}
\caption{Extended network and the cuts for deriving bound (\ref{ineq13}) for the three-user DIC}
\label{figb22}
\end{figure}

\noindent \textbf{Derivation of bound (\ref{ineq14}).}
To derive this bound, we design the extended network considered in Fig. \ref{figb23}. By applying GCS bound and picking cuts depicted in Fig. \ref{figb23}, we have
\small
\begin{equation} \label{eq3434}
\begin{aligned}
nR_{\Sigma}&=n(R_1+R_2+R_3+2kR_1+kR_2+kR_3)
\\& \overset{(a)}{\leq} H(Y_{{\Omega}_1^c}^n)+\sum_{l=2}^{4k+2}{H(Y_{{\Omega}_l^c\bigcap {\Omega}_{l-1}}^n|W_{\mathcal S\backslash{{\Omega}_l}},Y_{{\Omega}_{l-1}^c}^n)}+n\epsilon_n
\\& \overset{(b)}{\leq}\sum_{i=1}^{n}{kH(Y_{1}[i]|V_{1}[i]V_{2}[i]V_{3}[i])+kH(Y_{1}[i]|V_{1}[i]V_{2}[i])}\\&+\sum_{i=1}^{n}{kH(Y_{2}[i]|V_{2}[i]V_{3}[i])+kH(Y_{3}[i]|V_{3}[i])}
\\&+C+n\epsilon_n
\end{aligned}
\end{equation}
\normalsize
where $(a)$ follows from GCS \cite{gcs} for deterministic networks and step $(b)$ follows from the proof presented in Appendix B.

By letting $k$ goes to infinity, bound (\ref{ineq14}) will be recovered.
\begin{figure}[h10]
\centering
\includegraphics[trim =.5in 1.5in 2.5in 6in, clip,width=0.4\textwidth]{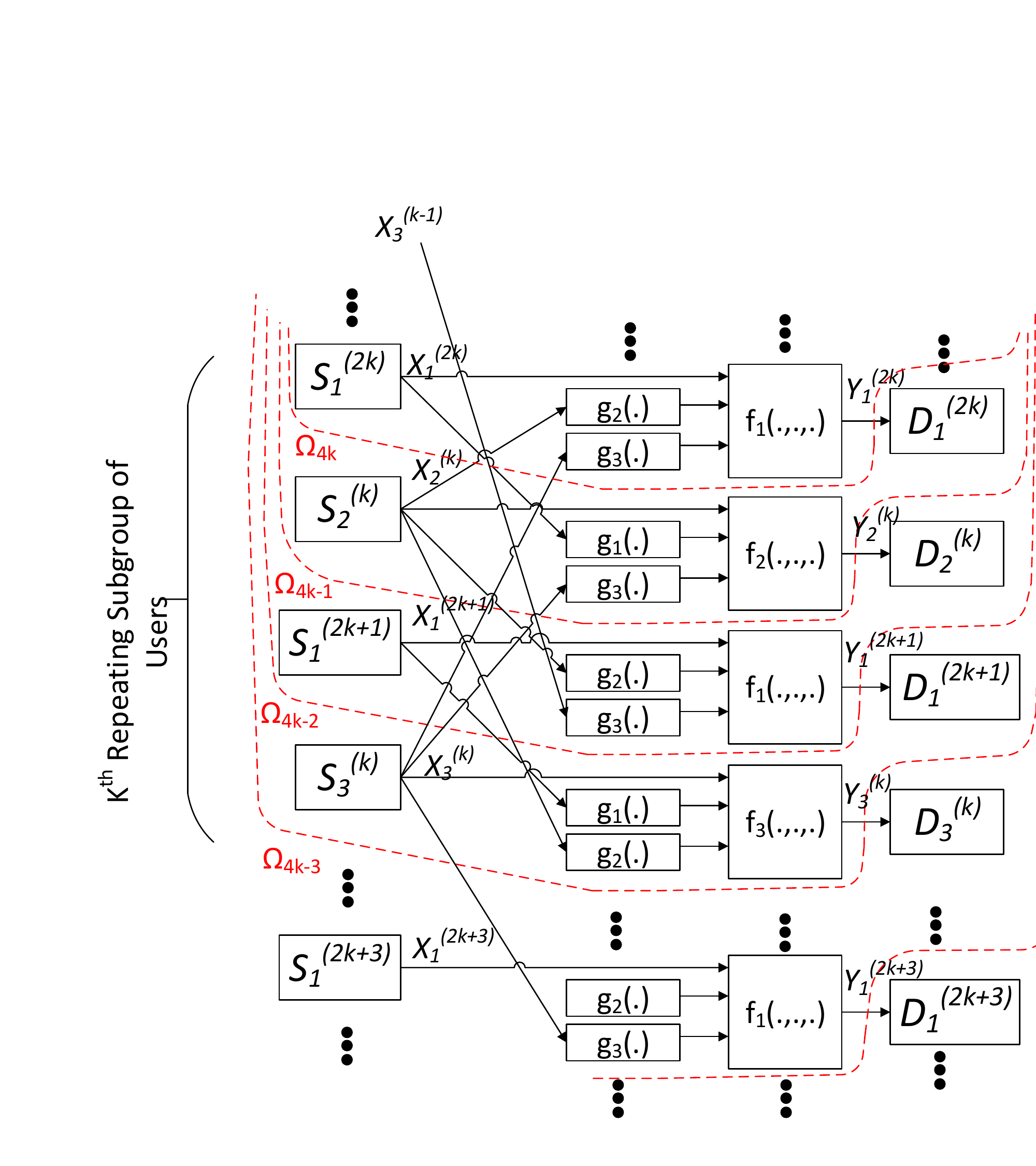}
\caption{Extended network and the cuts for deriving bound (\ref{ineq14}) for the three-user DIC}
\label{figb23}
\end{figure}

\noindent \textbf{Derivation of bound (\ref{ineq15}).}
To derive this bound, we design the extended network considered in Fig. \ref{figb24}. By applying GCS bound and picking cuts $\Omega_{1}, \Omega_{2}, \Omega_{3}$, and $\Omega_{4}$ depicted in Fig. \ref{figb24}, we have
\small
\begin{equation} \label{eq3434}
\begin{aligned}
nR_{\Sigma}&=n(R_1+R_2+R_3+R_1)
\\& \overset{(a)}{\leq} H(Y_{{\Omega}_1^c}^n)+\sum_{l=2}^{4}{H(Y_{{\Omega}_l^c\bigcap {\Omega}_{l-1}}^n|W_{\mathcal S\backslash{{\Omega}_l}},Y_{{\Omega}_{l-1}^c}^n)}+n\epsilon_n
\\& \overset{(b)}{\leq}\sum_{i=1}^{n}{2H(Y_{1}[i]|V_{1}[i]V_{2}[i]V_{3}[i])}+\sum_{i=1}^{n}{H(Y_{2}[i])}\\&+\sum_{i=1}^{n}{H(Y_{3}[i]|V_{2}[i]V_{3}[i])}+n\epsilon_n
\end{aligned}
\end{equation}
\normalsize
where $(a)$ follows from GCS \cite{gcs} for deterministic networks and step $(b)$ follows from the proof presented in Appendix B.
\begin{figure}[h10]
\centering
\includegraphics[trim = 1.5in 3.5in 2.5in 0in, clip,width=0.4\textwidth]{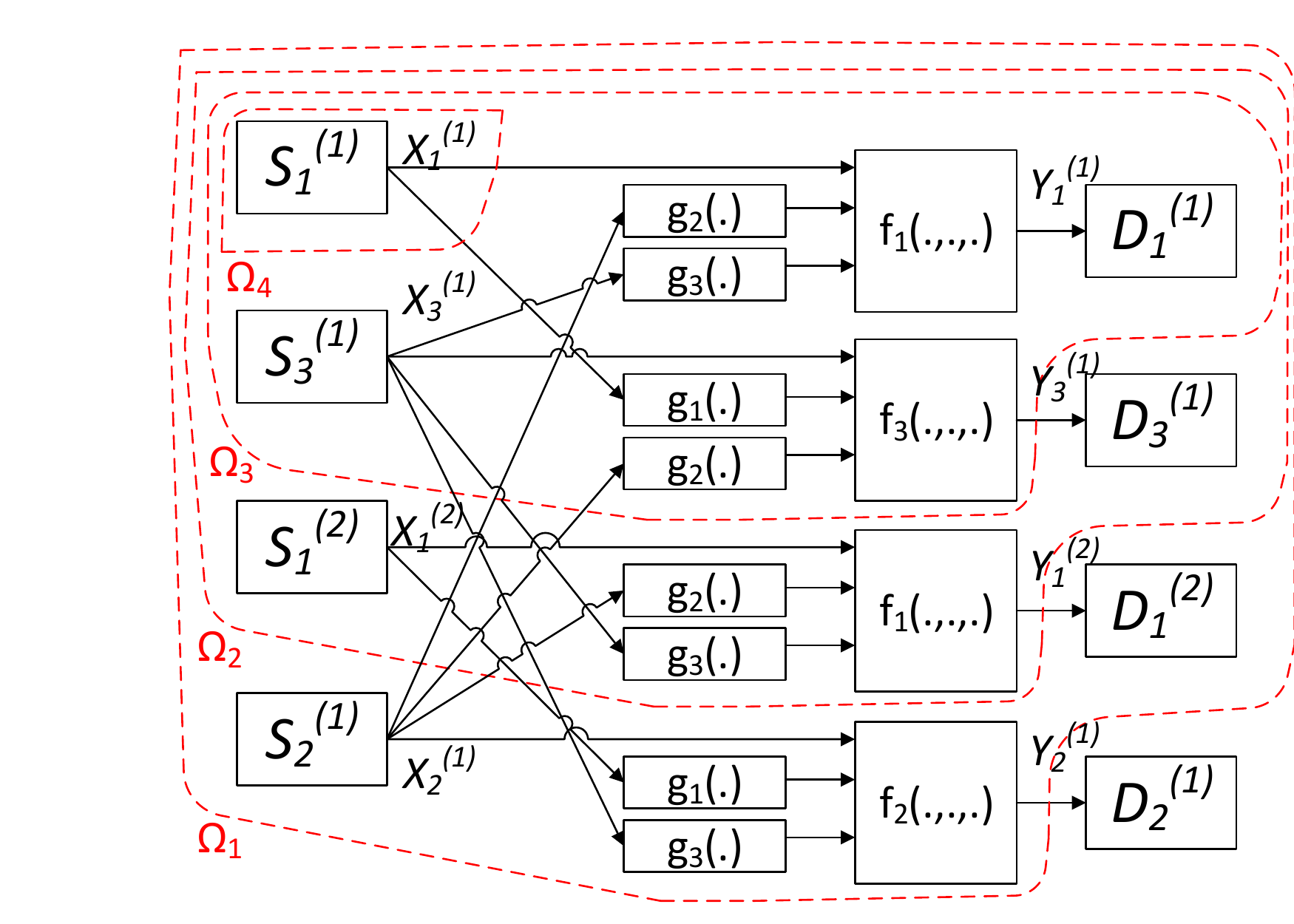}
\caption{Extended network and the cuts for deriving bound (\ref{ineq15}) for the three-user DIC}
\label{figb24}
\end{figure}

\noindent \textbf{Derivation of bound (\ref{ineq16}).}
To derive this bound, we design the extended network considered in Fig. \ref{figb25}. By applying GCS bound and picking cuts depicted in Fig. \ref{figb25}, we have
\small
\begin{equation} \label{eq3434}
\begin{aligned}
nR_{\Sigma}&=n(R_1+R_2+R_3+2kR_1+kR_2+kR_3)
\\& \overset{(a)}{\leq} H(Y_{{\Omega}_1^c}^n)+\sum_{l=2}^{4k+2}{H(Y_{{\Omega}_l^c\bigcap {\Omega}_{l-1}}^n|W_{\mathcal S\backslash{{\Omega}_l}},Y_{{\Omega}_{l-1}^c}^n)}+n\epsilon_n
\\& \overset{(b)}{\leq}\sum_{i=1}^{n}{2kH(Y_{1}[i]|V_{1}[i]V_{2}[i]V_{3}[i])+kH(Y_{2}[i]|V_{2}[i])}\\&+\sum_{i=1}^{n}{kH(Y_{3}[i]|V_{3}[i])}+C+n\epsilon_n
\end{aligned}
\end{equation}
\normalsize
where $(a)$ follows from GCS \cite{gcs} for deterministic networks and step $(b)$ follows from the proof presented in Appendix B.

By letting $k$ goes to infinity, bound (\ref{ineq16}) will be recovered.
\begin{figure}[h10]
\centering
\includegraphics[trim = .5in 1.5in 2.5in 6in, clip,width=0.4\textwidth]{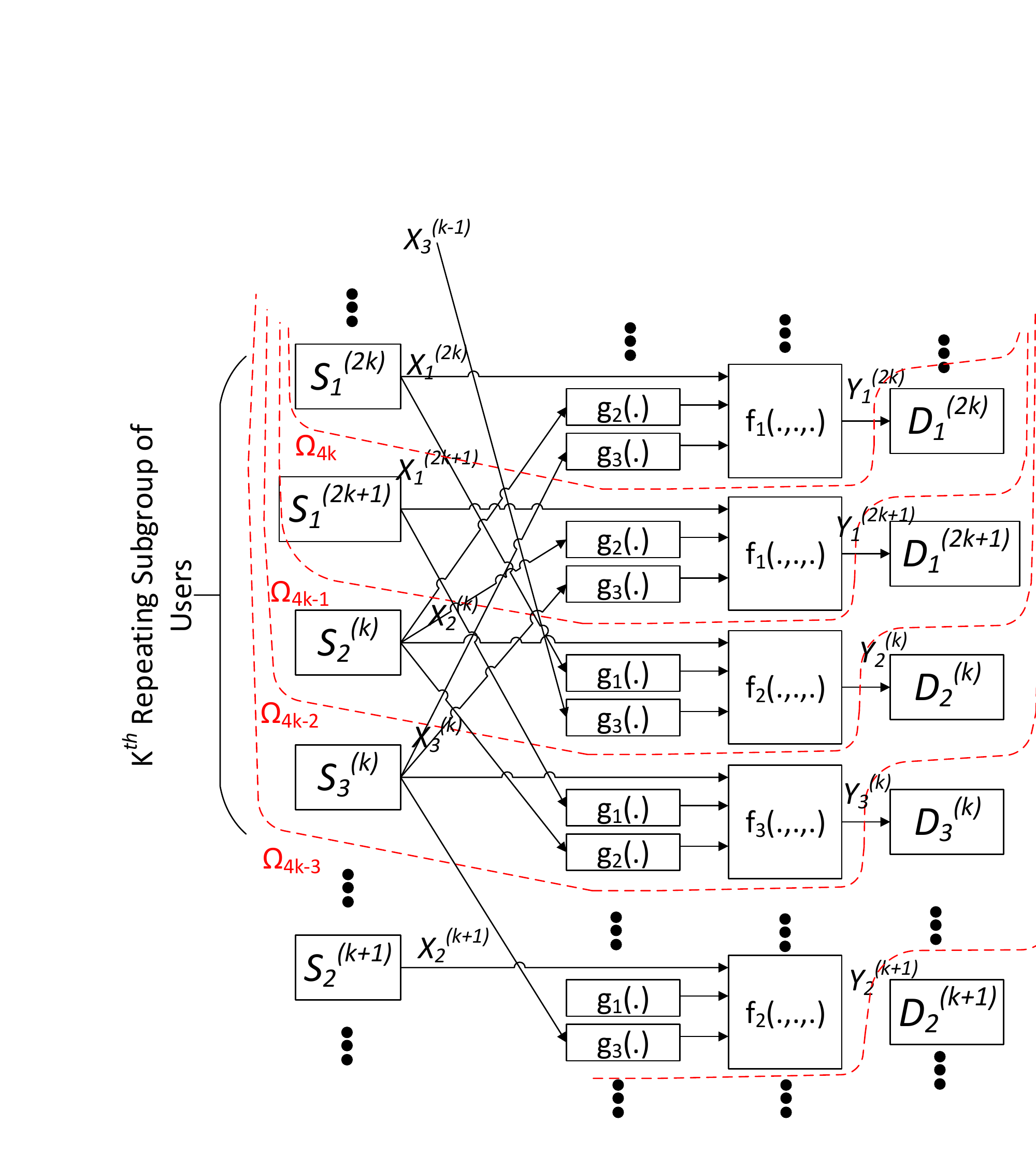}
\caption{Extended network and the cuts for deriving bound (\ref{ineq16}) for the three-user DIC}
\label{figb25}
\end{figure}

\noindent \textbf{Derivation of bound (\ref{ineq17}).}
To derive this bound, we design the extended network considered in Fig. \ref{figb26}. By applying GCS bound and picking cuts depicted in Fig. \ref{figb26}, we have
\small
\begin{equation} \label{eq3434}
\begin{aligned}
nR_{\Sigma}&=n(3R_1+R_2+R_3)
\\& \overset{(a)}{\leq} H(Y_{{\Omega}_1^c}^n)+\sum_{l=2}^{5}{H(Y_{{\Omega}_l^c\bigcap {\Omega}_{l-1}}^n|W_{\mathcal S\backslash{{\Omega}_l}},Y_{{\Omega}_{l-1}^c}^n)}+n\epsilon_n
\\& \overset{(b)}{\leq}\sum_{i=1}^{n}{2H(Y_{1}[i]|V_{1}[i]V_{2}[i]V_{3}[i])+H(Y_{1}[i])}\\&+\sum_{i=1}^{n}{H(Y_{2}[i]|V_{2}[i]V_{3}[i])+H(Y_{3}[i]|V_{2}[i]V_{3}[i])}+n\epsilon_n
\end{aligned}
\end{equation}
\normalsize
where $(a)$ follows from GCS \cite{gcs} for deterministic networks and step $(b)$ follows from the proof presented in Appendix B.

\begin{figure}[h11]
\centering
\includegraphics[trim = 1.5in 1.5in 1.5in 0in, clip,width=0.4\textwidth]{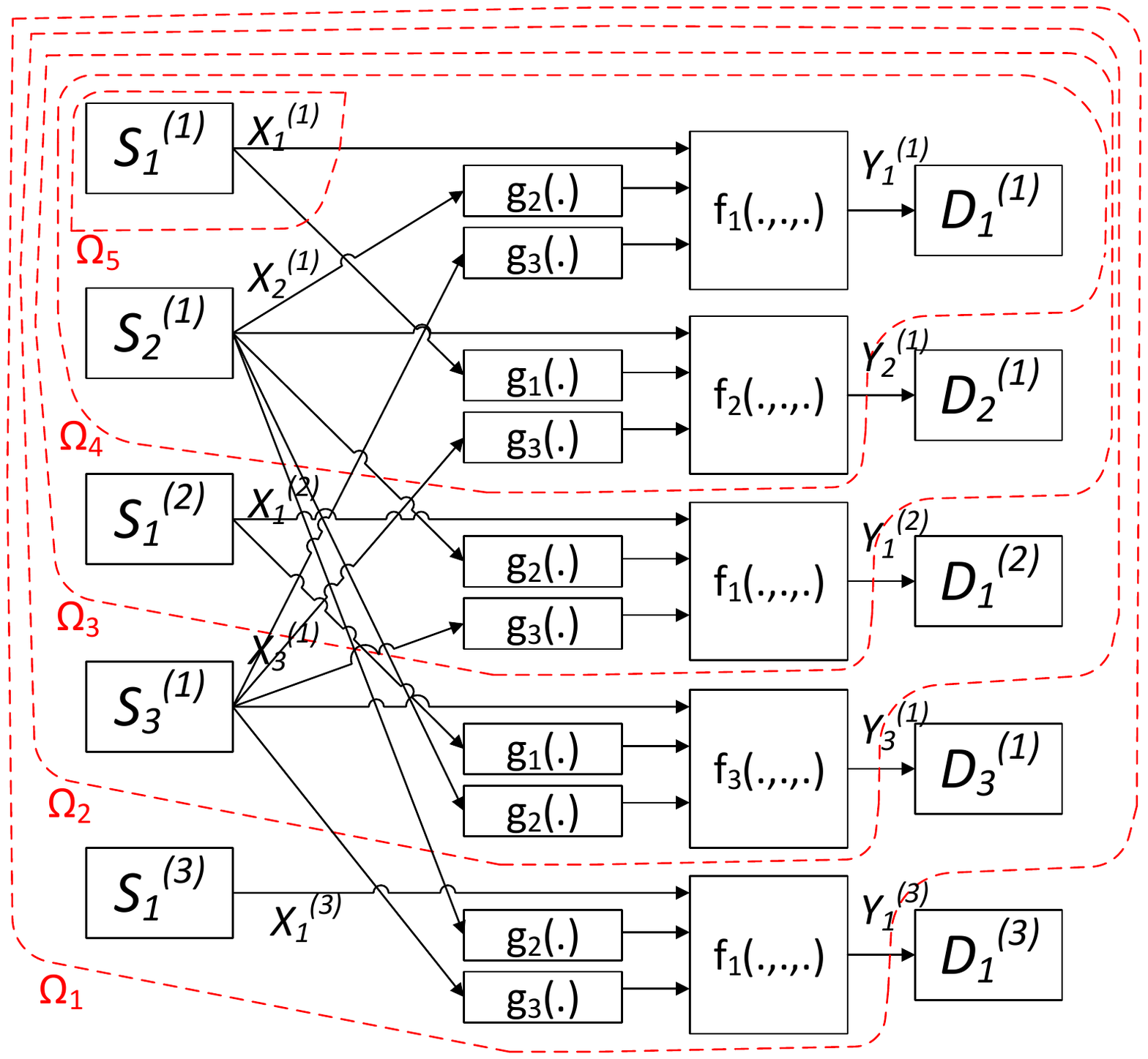}
\caption{Extended network and the cuts for deriving bound (\ref{ineq17}) for the three-user DIC}
\label{figb26}
\end{figure}

\noindent \textbf{Derivation of bound (\ref{ineq18}).}
To derive this bound, we design the extended network considered in Fig. \ref{figb27}. By applying GCS bound and picking cuts depicted in Fig. \ref{figb27}, we have
\small
\begin{equation} \label{eq3434}
\begin{aligned}
nR_{\Sigma}&=n(R_1+R_2+R_3+3kR_1+kR_2+kR_3)
\\& \overset{(a)}{\leq} H(Y_{{\Omega}_1^c}^n)+\sum_{l=2}^{5k+2}{H(Y_{{\Omega}_l^c\bigcap {\Omega}_{l-1}}^n|W_{\mathcal S\backslash{{\Omega}_l}},Y_{{\Omega}_{l-1}^c}^n)}+n\epsilon_n
\\& \overset{(b)}{\leq}\sum_{i=1}^{n}{2kH(Y_{1}[i]|V_{1}[i]V_{2}[i]V_{3}[i])+kH(Y_{1}[i]|V_{2}[i])}
\\&+\sum_{i=1}^{n}{kH(Y_{2}[i]|V_{2}[i]V_{3}[i])+kH(Y_{3}[i]|V_{3}[i])}
\\&+C+n\epsilon_n
\end{aligned}
\end{equation}
\normalsize
where $(a)$ follows from GCS \cite{gcs} for deterministic networks and step $(b)$ follows from the proof presented in Appendix B.

By letting $k$ goes to infinity, bound (\ref{ineq18}) will be recovered.
\begin{figure}[h11]
\centering
\includegraphics[trim = .5in 1in 2.5in 6in, clip,width=0.4\textwidth]{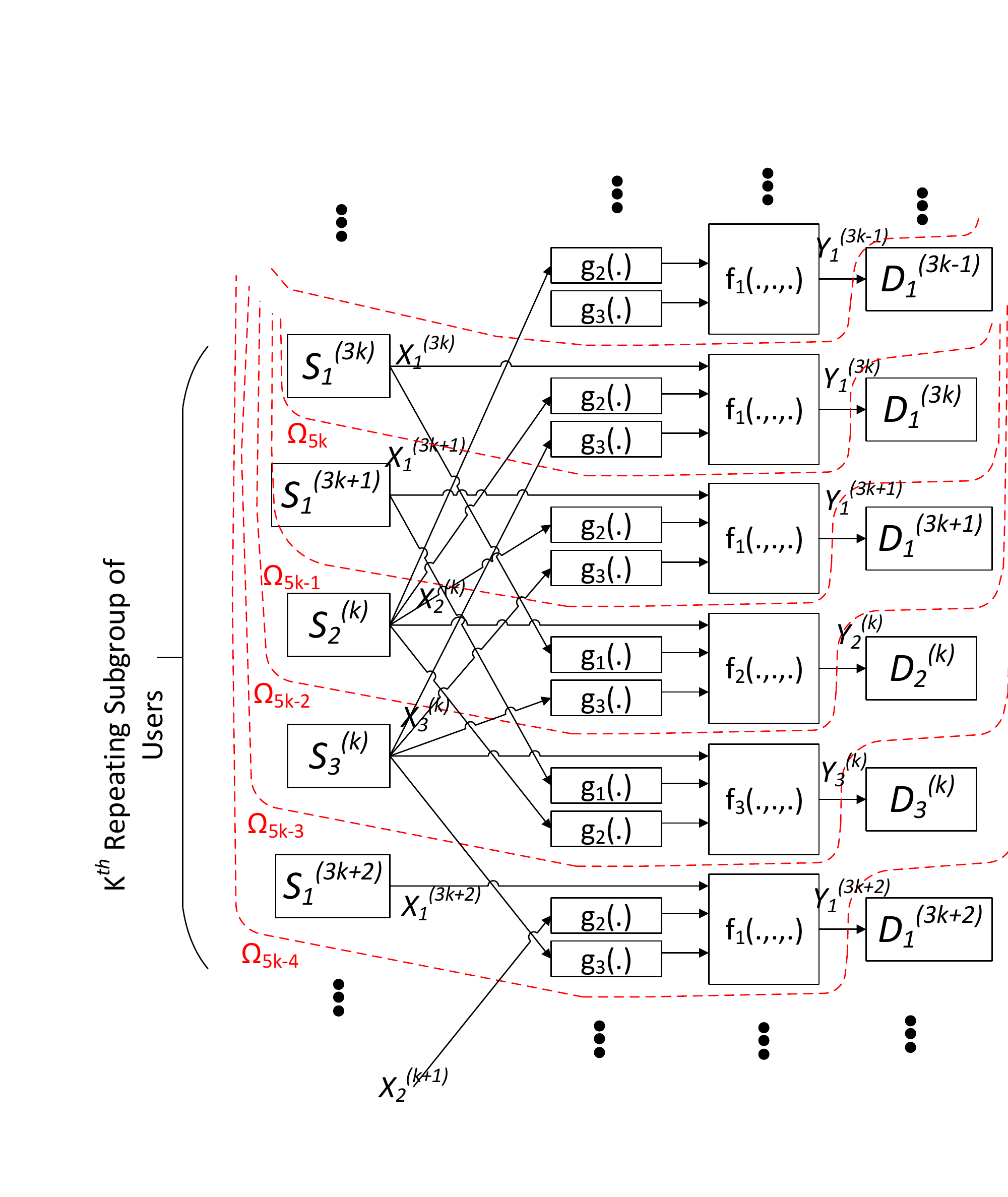}
\caption{Extended network and the cuts for deriving bound (\ref{ineq18}) for the three-user DIC}
\label{figb27}
\end{figure}

\noindent \textbf{Derivation of bound (\ref{ineq19}).}
To derive this bound, we design the extended network considered in Fig. \ref{figb29}.By applying GCS bound and picking cuts $\Omega_{1},...,\Omega_{5}$ depicted in Fig. \ref{figb10}, we have

\small
\begin{equation} \label{eq3434}
\begin{aligned}
nR_{\Sigma}&=n(2R_1+2R_2+R_3)
\\& \overset{(a)}{\leq} H(Y_{{\Omega}_1^c}^n)+\sum_{l=2}^{5}{H(Y_{{\Omega}_l^c\bigcap {\Omega}_{l-1}}^n|W_{\mathcal S\backslash{{\Omega}_l}},Y_{{\Omega}_{l-1}^c}^n)}+n\epsilon_n
\\&\overset{(b)}{\leq}\sum_{i=1}^{n}{Y_{{1^{(1)}}}[i])}+\sum_{i=1}^{n}{H(Y_{{3^{(1)}}}[i]|V_{{3^{(1)}}}[i]V_{{2^{(1)}}}[i])}\\&+\sum_{i=1}^{n}{H(Y_{1^{(1)}}[i]|V_{1^{(1)}}[i]V_{3^{(1)}}[i])}
\\&+2\sum_{i=1}^{n}{H(Y_{2^{(1)}}[i]|V_{2^{(1)}}[i]V_{1^{(1)}}[i]V_{3^{(1)}}[i])}+n\epsilon_n
\end{aligned}
\end{equation}
\normalsize
where $(a)$ follows from GCS \cite{gcs} for deterministic networks and step $(b)$ follows from the proof presented in Appendix B.
\begin{figure}[h28]
\centering
\includegraphics[trim = 2in 2.5in 1.2in 0in, clip,width=0.4\textwidth]{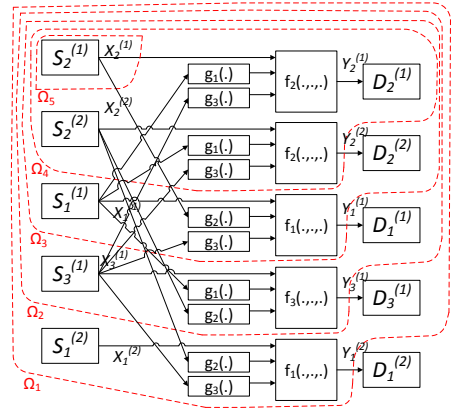}
\caption{Extended network and the cuts for deriving bound (\ref{ineq19}) for the three-user DIC}
\label{fig10}
\end{figure}

\noindent \textbf{Derivation of bound (\ref{ineq20}).}
To derive this bound, we design the extended network considered in Fig. \ref{figb29}. By applying GCS bound and picking cuts $\Omega_{1},..., \Omega_{5}$ depicted in Fig. \ref{figb29}, we have
\small
\begin{equation} \label{eq3434}
\begin{aligned}
nR_{\Sigma}&=n(2R_1+2R_2+R_3)
\\& \overset{(a)}{\leq} H(Y_{{\Omega}_1^c}^n)+\sum_{l=2}^{5}{H(Y_{{\Omega}_l^c\bigcap {\Omega}_{l-1}}^n|W_{\mathcal S\backslash{{\Omega}_l}},Y_{{\Omega}_{l-1}^c}^n)}+n\epsilon_n
\\& \overset{(b)}{\leq}\sum_{i=1}^{n}{H(Y_{1}[i])+H(Y_{1}[i]|V_{1}[i]V_{2}[i]V_{3}[i])}\\&+\sum_{i=1}^{n}{H(Y_{2}[i]|V_{1}[i]V_{2}[i]V_{3}[i])+H(Y_{2}[i]|V_{2}[i]V_{3}[i])}\\&+\sum_{i=1}^{n}{H(Y_{3}[i]|V_{1}[i]V_{3}[i])}+n\epsilon_n
\end{aligned}
\end{equation}
\normalsize
where $(a)$ follows from GCS \cite{gcs} for deterministic networks and step $(b)$ follows from the proof presented in Appendix B.
\begin{figure}[h10]
\centering
\includegraphics[trim = 1.5in 2.5in 2in 0in, clip,width=0.4\textwidth]{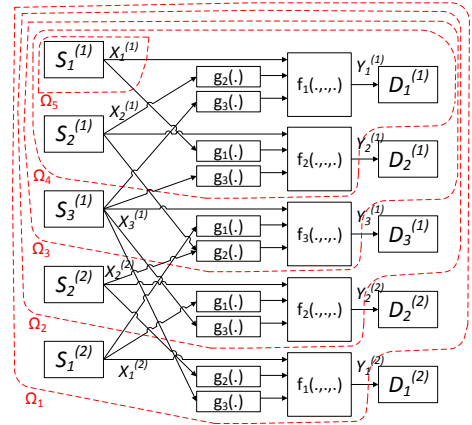}
\caption{Extended network and the cuts for deriving bound (\ref{ineq20}) for the three-user DIC}
\label{figb29}
\end{figure}

\noindent \textbf{Derivation of bound (\ref{ineq21}).}
To derive this bound, we design the extended network considered in Fig. \ref{figb30}. By applying GCS bound and picking cuts $\Omega_{1},..., \Omega_{5}$ depicted in Fig. \ref{figb30}, we have
\small
\begin{equation} \label{eq3434}
\begin{aligned}
nR_{\Sigma}&=n(2R_1+2R_2+R_3)
\\& \overset{(a)}{\leq} H(Y_{{\Omega}_1^c}^n)+\sum_{l=2}^{5}{H(Y_{{\Omega}_l^c\bigcap {\Omega}_{l-1}}^n|W_{\mathcal S\backslash{{\Omega}_l}},Y_{{\Omega}_{l-1}^c}^n)}+n\epsilon_n
\\& \overset{(b)}{\leq}\sum_{i=1}^{n}{H(Y_{1}[i])+H(Y_{1}[i]|V_{1}[i]V_{2}[i]V_{3}[i])}\\&+\sum_{i=1}^{n}{2H(Y_{2}[i]|V_{1}[i]V_{2}[i]V_{3}[i])+H(Y_{3}[i]|V_{3}[i])}+n\epsilon_n
\end{aligned}
\end{equation}
\normalsize
where $(a)$ follows from GCS \cite{gcs} for deterministic networks and step $(b)$ follows from the proof presented in Appendix B.
\begin{figure}[h10]
\centering
\includegraphics[trim = 1.5in 2.5in 2in 0in, clip,width=0.4\textwidth]{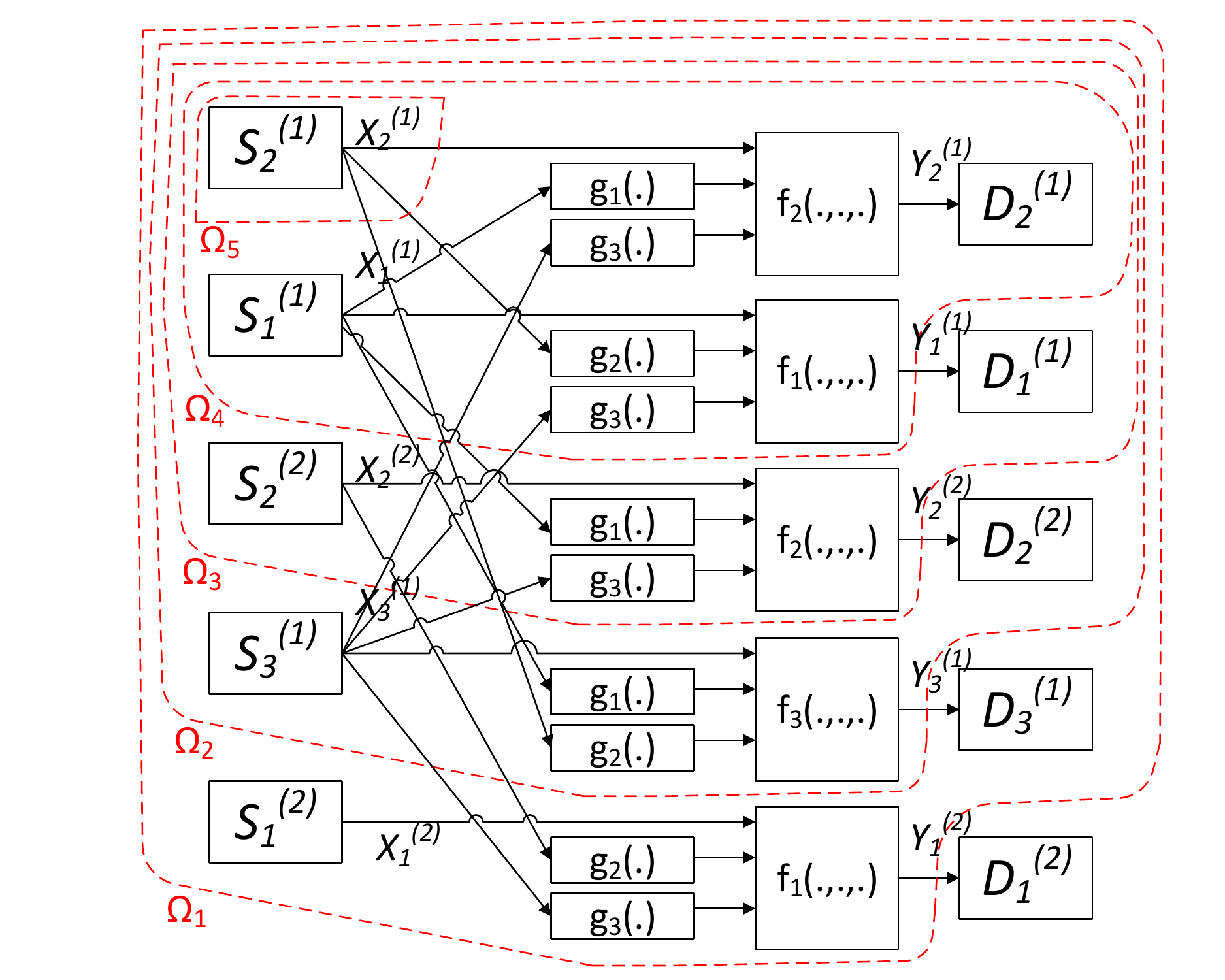}
\caption{Extended network and the cuts for deriving bound (\ref{ineq21}) for the three-user DIC}
\label{figb30}
\end{figure}

\noindent \textbf{Derivation of bound (\ref{ineq22}).}
To derive this bound, we design the extended network considered in Fig. \ref{figb31}. By applying GCS bound and picking cuts depicted in Fig. \ref{figb31}, we have
\small
\begin{equation} \label{eq3434}
\begin{aligned}
nR_{\Sigma}&=n(R_1+R_2+R_3+2kR_1+2kR_2+kR_3)
\\& \overset{(a)}{\leq} H(Y_{{\Omega}_1^c}^n)+\sum_{l=2}^{5k+2}{H(Y_{{\Omega}_l^c\bigcap {\Omega}_{l-1}}^n|W_{\mathcal S\backslash{{\Omega}_l}},Y_{{\Omega}_{l-1}^c}^n)}+n\epsilon_n
\\& \overset{(b)}{\leq}\sum_{i=1}^{n}{kH(Y_{1}[i]|V_{1}[i])+kH(Y_{1}[i]|V_{1}[i]V_{2}[i]V_{3}[i])}\\&+\sum_{i=1}^{n}{kH(Y_{2}[i]|V_{1}[i]V_{2}[i]V_{3}[i])}
\\&+\sum_{i=1}^{n}{kH(Y_{2}[i]|V_{2}[i]V_{3}[i])+kH(Y_{3}[i]|V_{3}[i])}\\&+C+n\epsilon_n
\end{aligned}
\end{equation}
\normalsize
where $(a)$ follows from GCS \cite{gcs} for deterministic networks and step $(b)$ follows from the proof presented in Appendix B.

By letting $k$ goes to infinity, bound (\ref{ineq22}) will be recovered.
\begin{figure}[h11]
\centering
\includegraphics[trim =0.5in 1in 2.5in 6in, clip,width=0.4\textwidth]{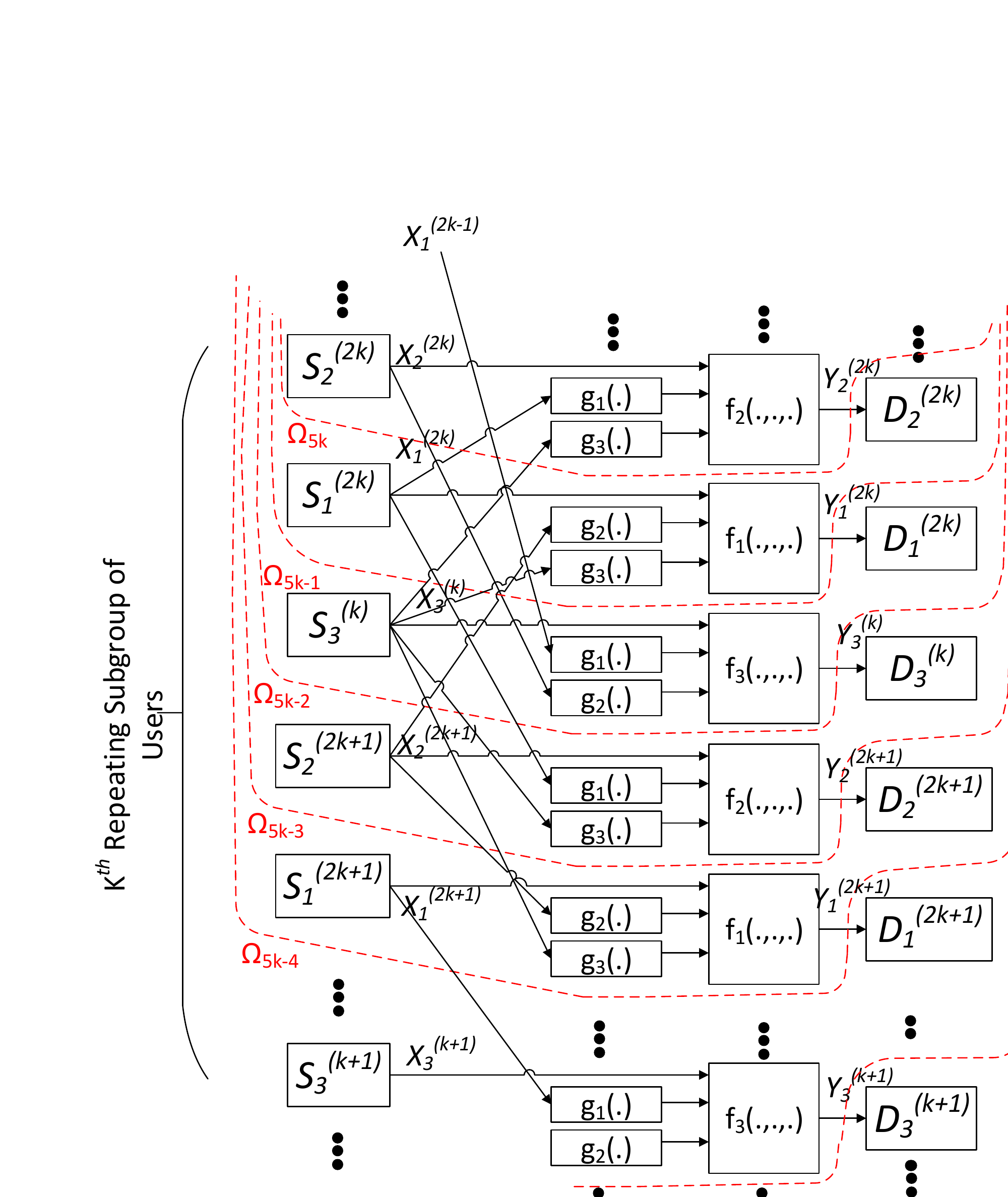}
\caption{Extended network and the cuts for deriving bound (\ref{ineq22}) for the three-user DIC}
\label{figb31}
\end{figure}

\noindent \textbf{Derivation of bound (\ref{ineq23}).}
To derive this bound, we design the extended network considered in Fig. \ref{figb32}. By applying GCS bound and picking cuts depicted in Fig. \ref{figb32}, we have
\small
\begin{equation} \label{eq3434}
\begin{aligned}
nR_{\Sigma}&=n(R_1+R_2+R_3+2kR_1+2kR_2+kR_3)
\\& \overset{(a)}{\leq} H(Y_{{\Omega}_1^c}^n)+\sum_{l=2}^{5k+2}{H(Y_{{\Omega}_l^c\bigcap {\Omega}_{l-1}}^n|W_{\mathcal S\backslash{{\Omega}_l}},Y_{{\Omega}_{l-1}^c}^n)}+n\epsilon_n
\\& \overset{(b)}{\leq}\sum_{i=1}^{n}{2kH(Y_{1}[i]|V_{1}[i]V_{3}[i])+kH(Y_{2}[i]|V_{1}[i]V_{2}[i]V_{3}[i])}\\&+\sum_{i=1}^{n}{2kH(Y_{2}[i]|V_{2}[i])+kH(Y_{3}[i]|V_{2}[i]V_{3}[i])}
\\&+C+n\epsilon_n
\end{aligned}
\end{equation}
\normalsize
where $(a)$ follows from GCS \cite{gcs} for deterministic networks and step $(b)$ follows from the proof presented in Appendix B.

By letting $k$ goes to infinity, bound (\ref{ineq23}) will be recovered.
\begin{figure}[h11]
\centering
\includegraphics[trim = 0.5in 1in 2.5in 6in, clip,width=0.4\textwidth]{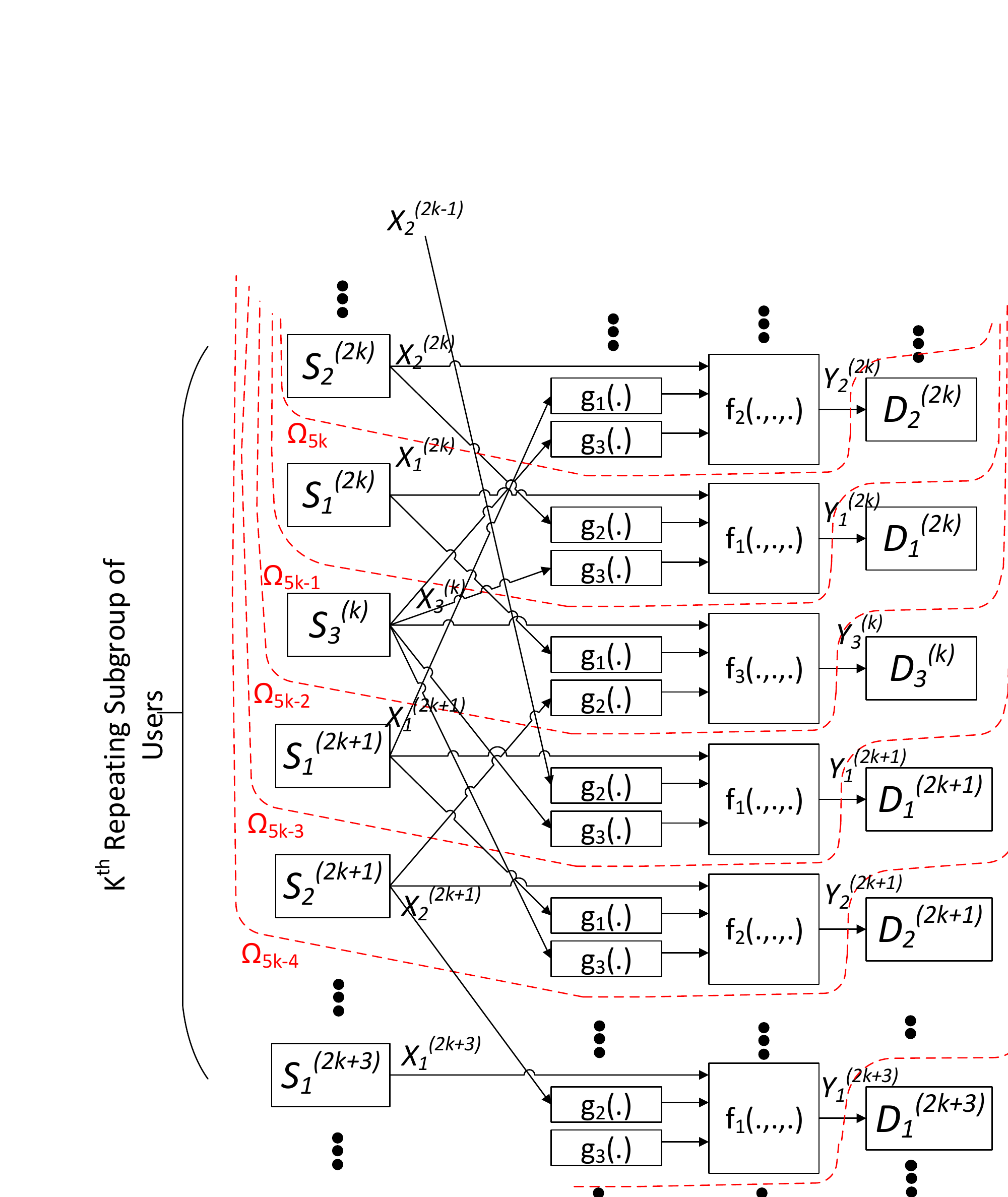}
\caption{Extended network and the cuts for deriving bound (\ref{ineq23}) for the three-user DIC}
\label{figb32}
\end{figure}

\noindent \textbf{Derivation of bound (\ref{ineq24}).}
To derive this bound, we design the extended network considered in Fig. \ref{figb33}. By applying GCS bound and picking cuts $\Omega_{1},..., \Omega_{6}$ depicted in Fig. \ref{figb33}, we have
\small
\begin{equation} \label{eq3434}
\begin{aligned}
nR_{\Sigma}&=n(3R_1+2R_2+R_3)
\\& \overset{(a)}{\leq} H(Y_{{\Omega}_1^c}^n)+\sum_{l=2}^{6}{H(Y_{{\Omega}_l^c\bigcap {\Omega}_{l-1}}^n|W_{\mathcal S\backslash{{\Omega}_l}},Y_{{\Omega}_{l-1}^c}^n)}+n\epsilon_n
\\& \overset{(b)}{\leq}\sum_{i=1}^{n}{2H(Y_{1}[i]|V_{1}[i]V_{2}[i]V_{3}[i])+H(Y_{1}[i])}\\&+\sum_{i=1}^{n}{H(Y_{2}[i]|V_{1}[i]V_{2}[i]V_{3}[i])+H(Y_{2}[i]|V_{2}[i]V_{3}[i])}\\&+\sum_{i=1}^{n}{H(Y_{3}[i]|V_{3}[i])}+n\epsilon_n
\end{aligned}
\end{equation}
\normalsize
where $(a)$ follows from GCS \cite{gcs} for deterministic networks and step $(b)$ follows from the proof presented in Appendix B.
\begin{figure}[h10]
\centering
\includegraphics[trim = 1.5in 1.5in 2in 0in, clip,width=0.4\textwidth]{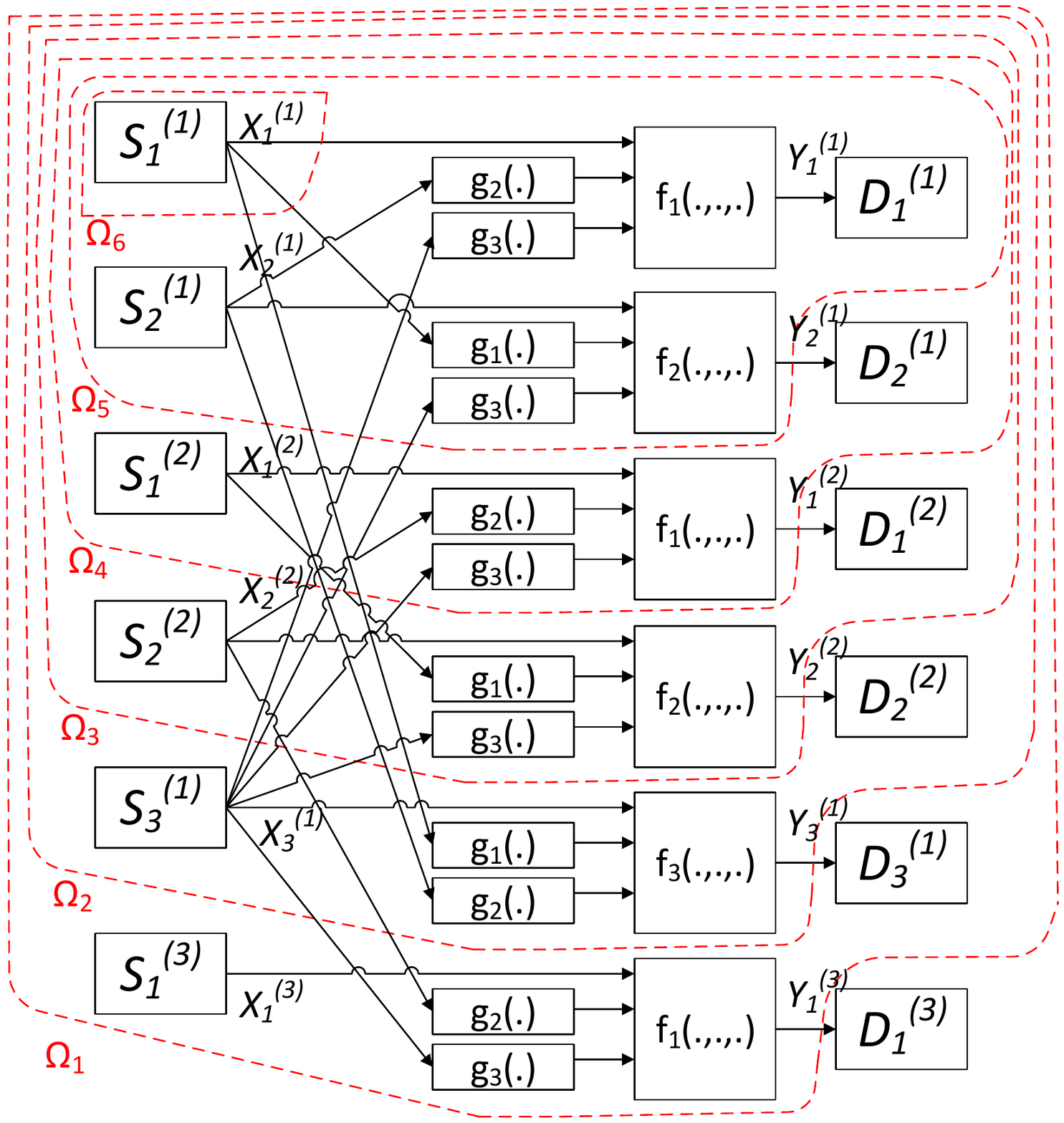}
\caption{Extended network and the cuts for deriving bound (\ref{ineq24}) for the three-user DIC}
\label{figb33}
\end{figure}

\noindent \textbf{Derivation of bound (\ref{ineq25}).}
To derive this bound, we design the extended network considered in Fig. \ref{figb34}. By applying GCS bound and picking cuts $\Omega_{1},..., \Omega_{6}$ depicted in Fig. \ref{figb34}, we have
\small
\begin{equation} \label{eq3434}
\begin{aligned}
nR_{\Sigma}&=n(3R_1+2R_2+R_3)
\\& \overset{(a)}{\leq} H(Y_{{\Omega}_1^c}^n)+\sum_{l=2}^{6}{H(Y_{{\Omega}_l^c\bigcap {\Omega}_{l-1}}^n|W_{\mathcal S\backslash{{\Omega}_l}},Y_{{\Omega}_{l-1}^c}^n)}+n\epsilon_n
\\& \overset{(b)}{\leq}\sum_{i=1}^{n}{2H(Y_{1}[i]|V_{1}[i]V_{2}[i]V_{3}[i])+H(Y_{1}[i])}\\&+\sum_{i=1}^{n}{2H(Y_{2}[i]|V_{2}[i]V_{3}[i])+H(Y_{3}[i]|V_{1}[i]V_{3}[i])}+n\epsilon_n
\end{aligned}
\end{equation}
\normalsize
where $(a)$ follows from GCS \cite{gcs} for deterministic networks and step $(b)$ follows from the proof presented in Appendix B.
\begin{figure}[h10]
\centering
\includegraphics[trim = 1.5in 1.5in 2in 0in, clip,width=0.4\textwidth]{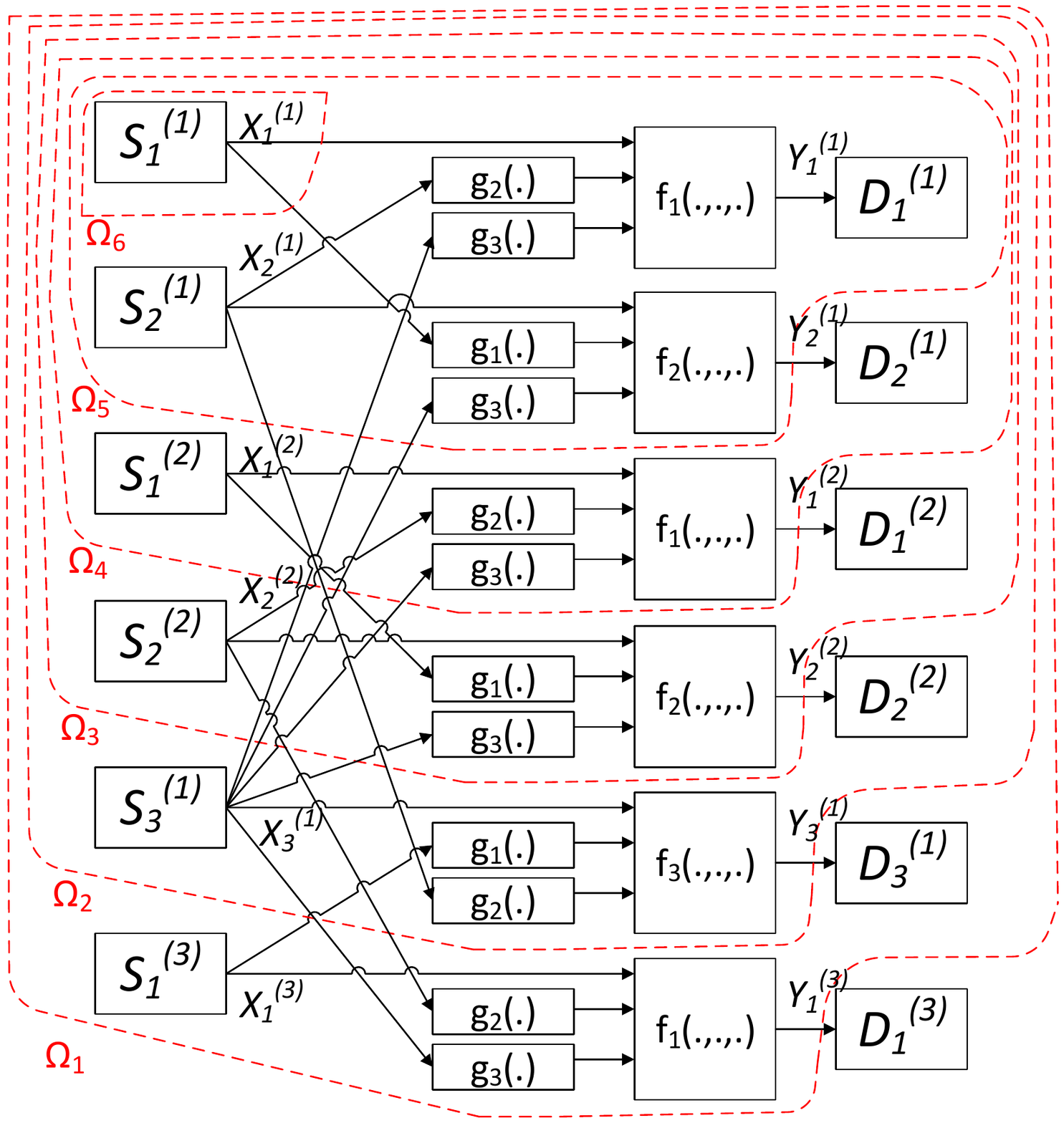}
\caption{Extended network and the cuts for deriving bound (\ref{ineq25}) for the three-user DIC}
\label{figb34}
\end{figure}

\noindent \textbf{Derivation of bound (\ref{ineq26}).}
To derive this bound, we design the extended network considered in Fig. \ref{figb35}. By applying GCS bound and picking cuts depicted in Fig. \ref{figb35}, we have
\small
\begin{equation} \label{eq3434}
\begin{aligned}
nR_{\Sigma}&=n(R_1+R_2+R_3+3kR_1+2kR_2+kR_3)
\\& \overset{(a)}{\leq} H(Y_{{\Omega}_1^c}^n)+\sum_{l=2}^{6k+2}{H(Y_{{\Omega}_l^c\bigcap {\Omega}_{l-1}}^n|W_{\mathcal S\backslash{{\Omega}_l}},Y_{{\Omega}_{l-1}^c}^n)}+n\epsilon_n
\\& \overset{(b)}{\leq}\sum_{i=1}^{n}{2kH(Y_{1}[i]|V_{1}[i]V_{2}[i]V_{3}[i])+kH(Y_{1}[i]|V_{1}[i])}\\&+\sum_{i=1}^{n}{2kH(Y_{2}[i]|V_{2}[i]V_{3}[i])+kH(Y_{3}[i]|V_{3}[i])}
\\&+C+n\epsilon_n
\end{aligned}
\end{equation}
\normalsize
where $(a)$ follows from GCS \cite{gcs} for deterministic networks and step $(b)$ follows from the proof presented in Appendix B.

By letting $k$ goes to infinity, bound (\ref{ineq26}) will be recovered.
\begin{figure}[h11]
\centering
\includegraphics[trim = .5in 7in 2.5in 6in, clip,width=0.4\textwidth]{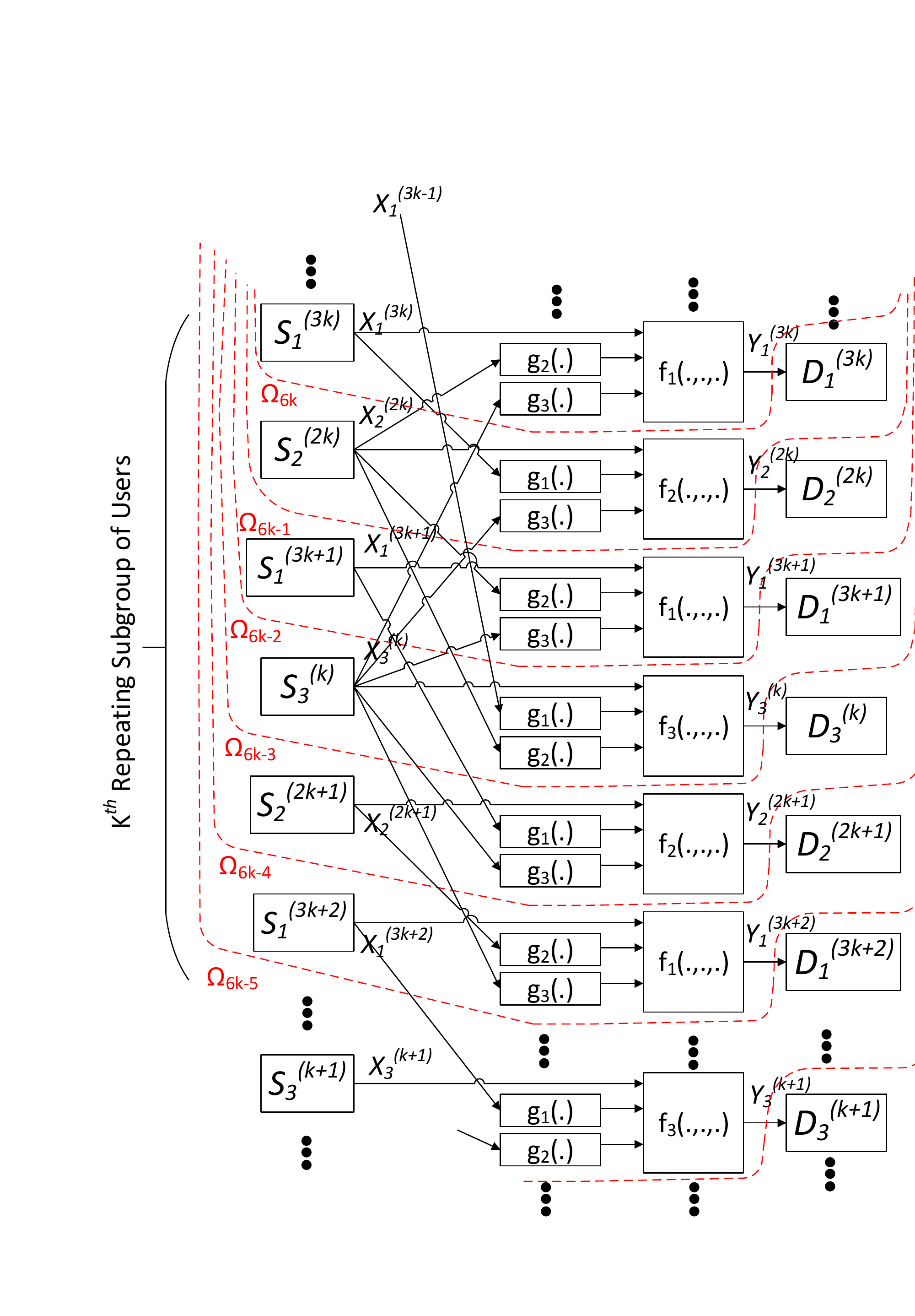}
\caption{Extended network and the cuts for deriving bound (\ref{ineq26}) for the three-user DIC}
\label{figb35}
\end{figure}

\noindent \textbf{Derivation of bound (\ref{ineq27}).}
To derive this bound, we design the extended network considered in Fig. \ref{figb36}. By applying GCS bound and picking cuts $\Omega_{1},..., \Omega_{6}$ depicted in Fig. \ref{figb36}, we have
\small
\begin{equation} \label{eq3434}
\begin{aligned}
nR_{\Sigma}&=n(3R_1+2R_2+R_3)
\\& \overset{(a)}{\leq} H(Y_{{\Omega}_1^c}^n)+\sum_{l=2}^{6}{H(Y_{{\Omega}_l^c\bigcap {\Omega}_{l-1}}^n|W_{\mathcal S\backslash{{\Omega}_l}},Y_{{\Omega}_{l-1}^c}^n)}+n\epsilon_n
\\& \overset{(b)}{\leq}\sum_{i=1}^{n}{3H(Y_{1}[i]|V_{1}[i]V_{2}[i]V_{3}[i])+H(Y_{2}[i]|V_{2}[i]V_{3}[i])}\\&+\sum_{i=1}^{n}{H(Y_{2}[i])+H(Y_{3}[i]|V_{3}[i])}+n\epsilon_n
\end{aligned}
\end{equation}
\normalsize
where $(a)$ follows from GCS \cite{gcs} for deterministic networks and step $(b)$ follows from the proof presented in Appendix B.
\begin{figure}[h10]
\centering
\includegraphics[trim = 1.5in 1.5in 2in 0in, clip,width=0.4\textwidth]{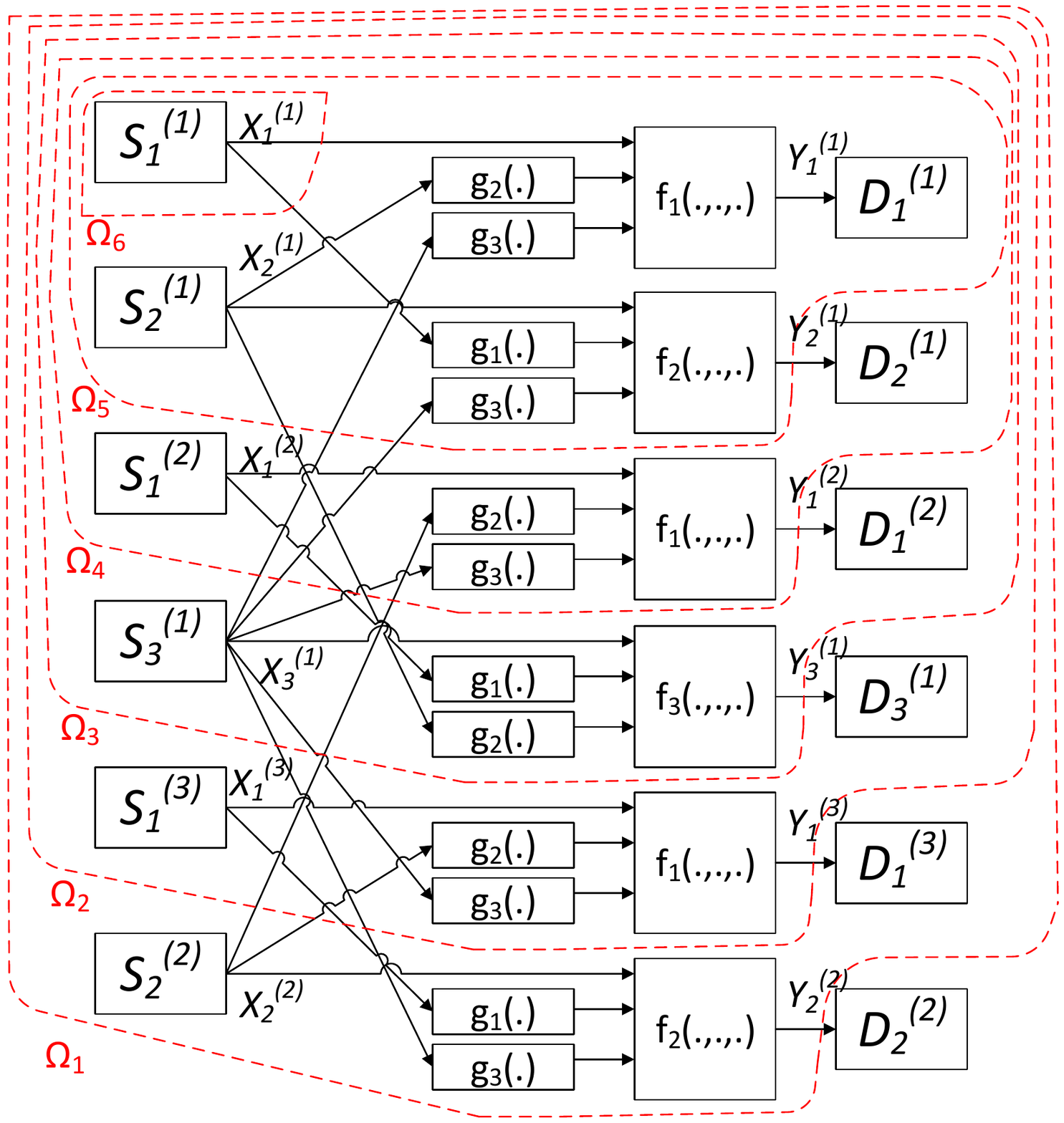}
\caption{Extended network and the cuts for deriving bound (\ref{ineq27}) for the three-user DIC}
\label{figb36}
\end{figure}

\noindent \textbf{Derivation of bound (\ref{ineq28}).}
To derive this bound, we design the extended network considered in Fig. \ref{figb37}. By applying GCS bound and picking cuts $\Omega_{1},..., \Omega_{7}$ depicted in Fig. \ref{figb37}, we have
\small
\begin{equation} \label{eq3434}
\begin{aligned}
nR_{\Sigma}&=n(4R_1+2R_2+R_3)
\\& \overset{(a)}{\leq} H(Y_{{\Omega}_1^c}^n)+\sum_{l=2}^{7}{H(Y_{{\Omega}_l^c\bigcap {\Omega}_{l-1}}^n|W_{\mathcal S\backslash{{\Omega}_l}},Y_{{\Omega}_{l-1}^c}^n)}+n\epsilon_n
\\& \overset{(b)}{\leq}\sum_{i=1}^{n}{3H(Y_{1}[i]|V_{1}[i]V_{2}[i]V_{3}[i])+H(Y_{1}[i])}\\&+\sum_{i=1}^{n}{2H(Y_{2}[i]|V_{2}[i]V_{3}[i])+H(Y_{3}[i]|V_{3}[i])}+n\epsilon_n
\end{aligned}
\end{equation}
\normalsize
where $(a)$ follows from GCS \cite{gcs} for deterministic networks and step $(b)$ follows from the proof presented in Appendix B.
\begin{figure}[h10]
\centering
\includegraphics[trim = 1.5in 8in 2in 0in, clip,width=0.4\textwidth]{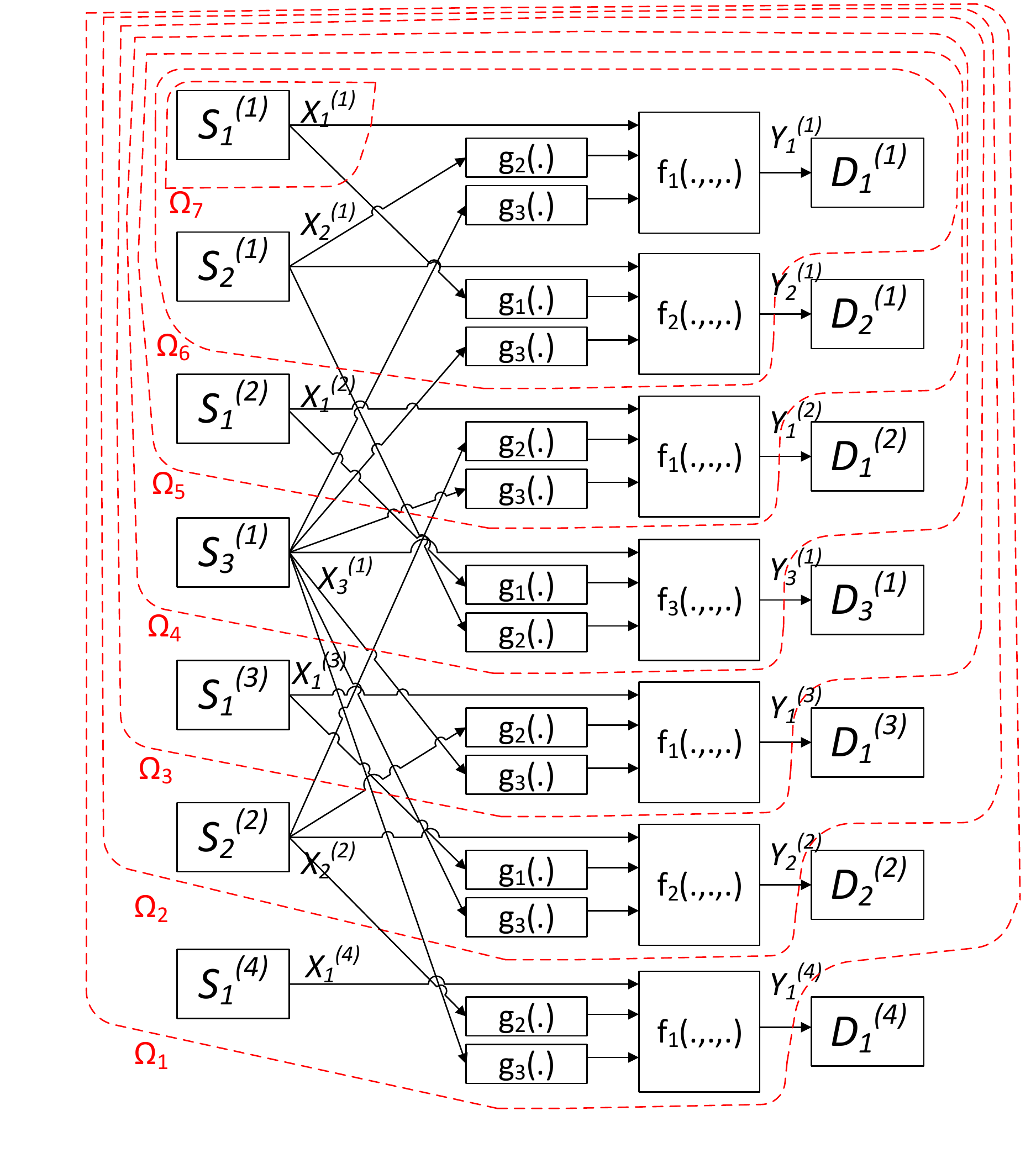}
\caption{Extended network and the cuts for deriving bound (\ref{ineq28}) for the three-user DIC}
\label{figb37}
\end{figure}

\section{Conclusion} \label{sec6}
We considered the two-user DIC, proposed the idea of extended networks, and demonstrated that by carefully designing extended networks and applying the GCS bound to them, we can derive a tight converse for the two-user DIC. Furthermore, we generalized our techniques to the three-user DIC, and demonstrated that the proposed approach also results in deriving a tight converse for the symmetric three-user DIC. An interesting future direction can be to investigate whether by applying GCS to extended networks one can establish the capacity region of general DIC. We will characterize the capacity region of symmetric $K$-user DIC in the longer version of this paper.

\appendix
\section{Title of Appendix A} \label{App:AppendixA}
\chapter{~~~~~~~~~~~~~~~~~~~~~~~~~Appendix A}
\\
\noindent \textbf{Proof of part (c) for bound (11).}
\\We first define $X_{i^{\mathcal S}}$ to be $\{X_{i^{(j)}}|j\in \mathcal S\}$ (Similarly for $W_{i^{\mathcal S}}$ and $Y_{i^{\mathcal S}}$).
\small
\begin{equation} \label{eq3434}
\begin{aligned}
&H(Y_{{\Omega}_1^c}^n)+\sum_{l=2}^{k+1}{H(Y_{{\Omega}_l^c\bigcap {\Omega}_{l-1}}^n|W_{\mathcal S\backslash{{\Omega}_l}},Y_{{\Omega}_{l-1}^c}^n)}+n\epsilon_n
\\&=H(Y_{1^{(k)}}^n)+\sum_{j=1}^{k-1}{H(Y_{1^{(j)}}^n|W_{1^{[j+1:k]}},Y_{1^{[j+1:k]}}^n)}
\\&~~~+H(Y_{2^{(1)}}^n|W_{1^{[1:k]}},Y_{1^{[1:k]}}^n)+n\epsilon_n
\\&=H(Y_{1^{(k)}}^n)
\\&~~~+\sum_{j=1}^{k-1}{H(Y_{1^{(j)}}^n|W_{1^{[j+1:k]}},Y_{1^{[j+1:k]}}^n,X_{1^{[j+1:k]}}^n)}
\\&~~~+H(Y_{2^{(1)}}^n|W_{1^{[1:k]}},Y_{1^{[1:k]}}^n,X_{1^{[1:k]}}^n)+n\epsilon_n
%
\\&\leq \sum_{i=1}^{n}{H(Y_{1^{(k)}}[i]|Y_{1^{(k)}}^{i-1})}
\\&~~~+\sum_{j=1}^{k-1}{\sum_{i=1}^{n}{H(Y_{1^{(j)}}[i]|Y_{1^{(j)}}^{i-1},Y_{1^{(j+1)}}[i],X_{1^{(j+1)}}[i])}}
\\&~~~+\sum_{i=1}^{n}{H(Y_{2^{(1)}}[i]|Y_{2^{(1)}}^{i-1},Y_{1^{(1)}}[i],X_{1^{(1)}}[i])}+n\epsilon_n
\\&\overset{(a)}{\leq} \sum_{i=1}^{n}{H(Y_{1^{(k)}}[i])}+\sum_{j=1}^{k-1}{\sum_{i=1}^{n}{H(Y_{1^{(j)}}[i]|V_{2^{(1)}}[i])}}
\\&~~~+\sum_{i=1}^{n}{H(Y_{2^{(1)}}[i]|V_{2^{(1)}}[i],V_{1^{(1)}}[i])}+n\epsilon_n
\\& \overset{(b)}{\leq}\sum_{i=1}^{n}{H(Y_{1^{(1)}}[i])}+(k-1)\sum_{i=1}^{n}{H(Y_{1^{(1)}}[i]|V_{2^{(1)}}[i])}
\\&~~~+\sum_{i=1}^{n}{H(Y_{2^{(1)}}[i]|V_{2^{(1)}}[i],V_{1^{(1)}}[i])}+n\epsilon_n
\end{aligned}
\end{equation}
\normalsize
where step $(a)$ follows from $V_{2^{(1)}}[i]=h_1(X_{1^{(j)}}[i],Y_{1^{(j)}}[i])$ and $V_{1^{(1)}}[i]=g_1(X_{1^{(1)}}[i])$. Step $(b)$ follows, since the coding schemes of $S_{1^{(j)}}$s are the same, $Y_{1^{(j)}}[i]$s have the same probability distribution as $Y_{1^{(1)}}[i]$.

\section{Title of Appendix A} \label{App:AppendixA}
\noindent \textbf{Proof of part (c) for bound (13).}
\small
\begin{equation} \label{eq3434}
\begin{aligned}
&H(Y_{{\Omega}_1^c}^n)+\sum_{l=2}^{2k}{H(Y_{{\Omega}_l^c\bigcap {\Omega}_{l-1}}^n|W_{\mathcal S\backslash{{\Omega}_l}},Y_{{\Omega}_{l-1}^c}^n)}
\\&=H(Y_{2^{(k)}}^n)+H(Y_{1^{(k)}}^n|W_{2^{(k)}},Y_{2^{(k)}}^n)
\\&+\sum_{j=1}^{k-1}{H(Y_{1^{(j)}}^n|W_{1^{[j+1:k]}},W_{2^{[j:k]}},Y_{1^{[j+1:k]}}^n,Y_{2^{[j:k]}}^n)}
\\&+\sum_{j=1}^{k-1}{H(Y_{2^{(j)}}^n|W_{1^{[j+1:k]}},W_{2^{[j+1:k]}},Y_{1^{[j+1:k]}}^n,Y_{2^{[j+1:k]}}^n)}
\\&\leq H(Y_{2^{(k)}}^n)+H(Y_{1^{(k)}}^n|X_{2^{(k)}}^n,Y_{2^{(k)}}^n)
\\&+\sum_{j=1}^{k-1}{H(Y_{1^{(j)}}^n|X_{1^{[j+1:k]}}^n,X_{2^{[j:k]}}^n,Y_{1^{[j+1:k]}}^n,Y_{2^{[j:k]}}^n)}
\\&+\sum_{j=1}^{k-1}{H(Y_{2^{(j)}}^n|X_{1^{[j+1:k]}}^n,X_{2^{[j+1:k]}}^n,Y_{1^{[j+1:k]}}^n,Y_{2^{[j+1:k]}}^n)}
\\&\overset{(a)}{=} H(Y_{2^{(k)}}^n)+\sum_{j=1}^{k}{H(Y_{1^{(j)}}^n|X_{2^{(j)}}^n,Y_{2^{(j)}}^n)}
\\&\sum_{j=1}^{k-1}{H(Y_{2^{(j)}}^n|X_{1^{(j+1)}}^n,Y_{1^{(j+1)}}^n)}
\\&{=} \sum_{i=1}^{n}{H(Y_{2^{(k)}}[i]|Y_{2^{(k)}}^{i-1})}
\\&~~~+\sum_{i=1}^{n}{\sum_{j=1}^{k}{H(Y_{1^{(j)}}[i]|Y_{1^{(j)}}^{(i-1)},X_{2^{(j)}}^n,Y_{2^{(j)}}^n)}}
\\&~~~+\sum_{i=1}^{n}{\sum_{j=1}^{k}{H(Y_{2^{(j)}}[i]|Y_{2^{(j)}}^{(i-1)},X_{1^{(j+1)}}^n,Y_{1^{(j+1)}}^n)}}
\\&\overset{(b)}{\leq} \sum_{i=1}^{n}{H(Y_{2^{(k)}}[i])}+\sum_{j=1}^{k}{\sum_{i=1}^{n}{{H(Y_{1^{(j)}}[i]|V_{1^{(j)}}[i])}}}
\\&~~~+\sum_{j=1}^{k-1}{\sum_{i=1}^{n}{{H(Y_{2^{(j)}}[i]|V_{2^{(j)}}[i])}}}
\\&\overset{(c)}{=} \sum_{i=1}^{n}{H(Y_{2^{(1)}}[i])}+k\sum_{i=1}^{n}{H(Y_{1^{(1)}}[i]|V_{1^{(1)}}[i])}
\\&~~~+(k-1)\sum_{i=1}^{n}{H(Y_{2^{(1)}}[i]|V_{2^{(1)}}[i])}
\end{aligned}
\end{equation}
\normalsize
where step $(a)$ follows from the structure of interference links. That is to say, each copied version of users only imposes interference on the user placed below it. Step $(b)$ follows from $V_{2^{(j)}}[i]=h_1(X_{1^{(j+1)}}[i],Y_{1^{(j+1)}}[i])$, $V_{1^{(j)}}[i]=h_2(X_{2^{(j)}}[i],Y_{2^{(j)}}[i])$. Finally, step $(c)$ follows, since the coding schemes of $S_{1^{(j)}}$s and $S_{2^{(j)}}$s are the same respectively, $Y_{1^{(j)}}[i]$s and $Y_{2^{(j)}}[i]$s have the same probability distribution respectively. Furthermore, since $V_{i^{(j)}}=g_i(X_{i^{(j)}})$, $V_{1^{(j)}}[i]$s and $V_{2^{(j)}}[i]$s have the same probability distribution respectively.

\section{Title of Appendix A} \label{App:AppendixA}
\noindent \textbf{Proof of part (c) for bound (15).}
\small
\begin{equation} \label{eq3434}
\begin{aligned}
&H(Y_{{\Omega}_1^c}^n)+H(Y_{{\Omega}_2^c\bigcap {\Omega}_1}^n|W_{\mathcal S\backslash{{\Omega}_2}},Y_{{\Omega}_1^c}^n)
\\&~~~+H(Y_{{\Omega}_3^c\bigcap {\Omega}_2}^n|W_{\mathcal S\backslash{{\Omega}_3}},Y_{{\Omega}_2^c}^n)
\\&=H(Y_{1^{(2)}}^n)+H(Y_{2^{(1)}}^n|W_{1^{(2)}},Y_{1^{(2)}}^n)
\\&~~~+H(Y_{1^{(1)}}^n|W_{2^{(1)}},W_{1^{(2)}},Y_{2^{(1)}}^n,Y_{1^{(2)}}^n)
\\&=H(Y_{1^{(2)}}^n)+H(Y_{2^{(1)}}^n|W_{1^{(2)}},Y_{1^{(2)}}^n,X_{1^{(2)}}^n)
\\&~~~+H(Y_{1^{(1)}}^n|W_{2^{(1)}},W_{1^{(2)}},Y_{2^{(1)}}^n,Y_{1^{(2)}}^n,X_{2^{(1)}}^n,X_{1^{(2)}}^n)
\\&\leq H(Y_{1^{(2)}}^n)+H(Y_{2^{(1)}}^n|Y_{1^{(2)}}^n,X_{1^{(2)}}^n)+H(Y_{1^{(1)}}^n|Y_{2^{(1)}}^n,X_{2^{(1)}}^n)
\\&=\sum_{i=1}^{n}{H(Y_{1^{(2)}}[i]|Y_{1^{(2)}}^{i-1})}
\\&~~~+\sum_{i=1}^{n}{H(Y_{2^{(1)}}[i]|Y_{2^{(1)}}^{i-1},Y_{1^{(2)}}^n,X_{1^{(2)}}^n)}
\\&~~~+\sum_{i=1}^{n}{H(Y_{1^{(1)}}[i]|Y_{1^{(1)}}^{i-1},Y_{2^{(1)}}^n,X_{2^{(1)}}^n)}
\\& \overset{(a)}{\leq} \sum_{i=1}^{n}{H(Y_{1^{(2)}}[i])}+\sum_{i=1}^{n}{H(Y_{2^{(1)}}[i]|V_{2^{(1)}}[i])}
\\&~~~+\sum_{i=1}^{n}{H(Y_{1^{(1)}}[i]|V_{2^{(1)}}[i],V_{1^{(1)}}[i])}
\\& \overset{(b)}{=}\sum_{i=1}^{n}{H(Y_{1^{(1)}}[i])}+\sum_{i=1}^{n}{H(Y_{2^{(1)}}[i]|V_{2^{(1)}}[i])}
\\&~~~+\sum_{i=1}^{n}{H(Y_{1^{(1)}}[i]|V_{2^{(1)}}[i],V_{1^{(1)}}[i])}
\end{aligned}
\end{equation}
\normalsize
where step $(a)$ follows from $V_{2^{(1)}}[i]=h_1(X_{1^{(2)}}[i],Y_{1^{(2)}}[i])$, $V_{1^{(1)}}[i]=h_2(X_{2^{(1)}}[i],Y_{2^{(1)}}[i])$, and $V_{2^{(1)}}[i]=g_2(X_{2^{(1)}}[i])$. Step $(b)$ follows, since the coding schemes of $S_{1^{(2)}}$ and $S_{1^{(1)}}$ are the same, $Y_{1^{(2)}}[i]$ has the same probability distribution as $Y_{1^{(1)}}[i]$.

\section{Title of Appendix A} \label{App:AppendixA}
\label{apb10}
\chapter{~~~~~~~~~~~~~~~~~~~~~~~~~Appendix B}
\\
\noindent \textbf{Proof of part (b) for bound (\ref{ineq1}).}
\small
\begin{equation} \label{eq3434}
\begin{aligned}
&H(Y_{{\Omega}_1^c}^n)+\sum_{l=2}^{k+2}{H(Y_{{\Omega}_l^c\bigcap {\Omega}_{l-1}}^n|W_{\mathcal S\backslash{{\Omega}_l}},Y_{{\Omega}_{l-1}^c}^n)}
\\& \leq \sum_{j=2}^{k}{H(Y_{1^{(j)}}^n|Y_{1^{(j+1)}}^nX_{1^{(j+1)}}^n)}+C
\\& \leq \sum_{j=2}^{k}{\sum_{i=1}^{n}{H(Y_{1^{(j)}}[i]|Y_{1^{(j+1)}}[i]X_{1^{(j+1)}}[i])}}+C
\\& \overset{(a)}{\leq} \sum_{j=2}^{k}{\sum_{i=1}^{n}{H(Y_{1^{(j)}}[i]|V_{2^{(1)}}[i]V_{3^{(1)}}[i])}}+C
\\& \overset{(b)}{\leq} k{\sum_{i=1}^{n}{H(Y_{1^{(1)}}[i]|V_{2^{(1)}}[i]V_{3^{(1)}}[i])}}+C
\end{aligned}
\end{equation}
\normalsize
where $C\triangleq H(Y_{1^{(1)}}^n)+H(Y_{2^{(1)}}^n)+H(Y_{3^{(1)}}^n)$ and $(a)$ follows from $(V_{2^{(1)}}[i],V_{3^{(1)}}[i])=h_1(X_{1^{(j)}}[i],Y_{1^{(j)}}[i])$. Step $(b)$ follows, since the coding schemes of $S_{1^{(j)}}$ are the same , $Y_{1^{(j)}}[i]$s  have the same probability distribution.

\section{Title of Appendix A} \label{App:AppendixA}
\label{apb11}
\noindent \textbf{Proof of part (b) for bound (\ref{ineq2}).}
\small
\begin{equation} \label{eq3434}
\begin{aligned}
&H(Y_{{\Omega}_1^c}^n)+\sum_{l=2}^{2k+2}{H(Y_{{\Omega}_l^c\bigcap {\Omega}_{l-1}}^n|W_{\mathcal S\backslash{{\Omega}_l}},Y_{{\Omega}_{l-1}^c}^n)}
\\& \leq \sum_{j=1}^{k}{H(Y_{2^{(j)}}^n|Y_{1^{(j+1)}}^nX_{1^{(j+1)}}^n)}
\\&+ \sum_{j=1}^{k}{H(Y_{1^{(j)}}^n|Y_{2^{(j)}}^nX_{2^{(j)}}^n)}+C
\\& \leq \sum_{j=1}^{k}{\sum_{i=1}^{n}{H(Y_{2^{(j)}}[i]|Y_{1^{(j+1)}}[i]X_{1^{(j+1)}}[i])}}
\\&+ \sum_{j=1}^{k}{\sum_{i=1}^{n}{H(Y_{1^{(j)}}[i]|Y_{2^{(j)}}[i]X_{2^{(j)}}[i])}}+C
\\& \overset{(a)}{\leq} \sum_{j=1}^{k}{\sum_{i=1}^{n}{H(Y_{2^{(j)}}[i]|V_{3^{(1)}}[i])+
H(Y_{1^{(j)}}[i]|V_{1^{(j)}}[i]V_{2^{(j)}}[i]V_{3^{(1)}}[i])}}\\&+C
\\& \overset{(b)}{\leq} \sum_{i=1}^{n}{kH(Y_{2^{(1)}}[i]|V_{3^{(1)}}[i])
+kH(Y_{1^{(1)}}[i]|V_{1^{(1)}}[i]V_{2^{(1)}}[i]V_{3^{(1)}}[i])}\\&+C
\end{aligned}
\end{equation}
\normalsize
where $(a)$ follows from $(V_{2^{(j)}}[i],V_{3^{(1)}}[i])=h_1(X_{1^{(j)}}[i],Y_{1^{(j)}}[i])$, $(V_{1^{(j)}}[i],V_{3^{(1)}}[i])=h_2(X_{2^{(j)}}[i],Y_{2^{(j)}}[i])$, and $V_{2^{(j)}}[i]=g_2(X_{2^{(j)}}[i])$. Step $(b)$ follows, since the coding schemes of $S_{1^{(j)}}$s and $S_{2^{(j)}}$s are the same respectively, $Y_{1^{(j)}}[i]$s and $Y_{2^{(j)}}[i]$s have the same probability distribution respectively.
\\
\section{Title of Appendix A} \label{App:AppendixA}
\label{apb12}
\noindent \textbf{Proof of part (b) for bound (\ref{ineq3}).}
\small
\begin{equation} \label{eq3434}
\begin{aligned}
&H(Y_{{\Omega}_1^c}^n)+\sum_{l=2}^{2k+2}{H(Y_{{\Omega}_l^c\bigcap {\Omega}_{l-1}}^n|W_{\mathcal S\backslash{{\Omega}_l}},Y_{{\Omega}_{l-1}^c}^n)}
\\& \leq \sum_{j=1}^{k}{H(Y_{2^{(j)}}^n|Y_{1^{(j+1)}}^nX_{1^{(j+1)}}^n)}
\\&+ \sum_{j=1}^{k}{H(Y_{1^{(j)}}^n|Y_{2^{(j)}}^nX_{2^{(j)}}^n)}+C
\\& \leq \sum_{j=1}^{k}{\sum_{i=1}^{n}{H(Y_{2^{(j)}}[i]|Y_{1^{(j+1)}}[i]X_{1^{(j+1)}}[i])}}
\\&+\sum_{j=1}^{k}{\sum_{i=1}^{n}{H(Y_{1^{(j)}}[i]|Y_{2^{(j)}}[i]X_{2^{(j)}}[i])}}+C
\\& \overset{(a)}{\leq} \sum_{j=1}^{k}{\sum_{i=1}^{n}{H(Y_{2^{(j)}}[i]|V_{2^{(j)}}[i]V_{3^{(1)}}[i])+
H(Y_{1^{(j)}}[i]|V_{1^{(j)}}[i]V_{3^{(1)}}[i])}}\\&+C
\\& \overset{(b)}{\leq} \sum_{i=1}^{n}{kH(Y_{2^{(1)}}[i]|Y_{2^{(1)}}[i]V_{3^{(1)}}[i])
+kH(Y_{1^{(1)}}[i]|V_{1^{(1)}}[i]V_{3^{(1)}}[i])}\\&+C
\end{aligned}
\end{equation}
\normalsize
where step $(a)$ follows from $(V_{2^{(j)}}[i],V_{3^{(1)}}[i])=h_1(X_{1^{(j+1)}}[i],Y_{1^{(j+1)}}[i])$ and $(V_{1^{(j)}}[i],V_{3^{(1)}}[i])=h_2(X_{2^{(j)}}[i],Y_{2^{(j)}}[i])$. Step $(b)$ follows, since the coding schemes of $S_{1^{(j)}}$s and $S_{2^{(j)}}$s are the same respectively, $Y_{1^{(j)}}[i]$s and $Y_{2^{(j)}}[i]$s have the same probability distribution respectively.
\\
\section{Title of Appendix A} \label{App:AppendixA}
\label{apb13}
\noindent \textbf{Proof of part (b) for bound (\ref{ineq4}).}
\small
\begin{equation} \label{eq3434}
\begin{aligned}
&H(Y_{{\Omega}_1^c}^n)+\sum_{l=2}^{3k+2}{H(Y_{{\Omega}_l^c\bigcap {\Omega}_{l-1}}^n|W_{\mathcal S\backslash{{\Omega}_l}},Y_{{\Omega}_{l-1}^c}^n)}
\\& \leq \sum_{j}{H(Y_{1^{(2j+1)}}^n|Y_{1^{(2j+2)}}^nX_{1^{(2j+2)}}^n)}
\\&+ \sum_{j}{H(Y_{2^{(j)}}^n|Y_{1^{(2j+1)}}^nX_{1^{(2j+1)}}^n)}
\\&+ \sum_{j}{H(Y_{1^{(2j)}}^n|Y_{2^{(j)}}^nX_{2^{(j)}}^n)}+C
\\& \leq \sum_{j}{\sum_{i=1}^{n}{H(Y_{1^{(2j+1)}}[i]|Y_{1^{(2j+2)}}[i]X_{1^{(2j+2)}}[i])}}
\\&+\sum_{j}{\sum_{i=1}^{n}{H(Y_{2^{(j)}}[i]|Y_{1^{(2j+1)}}[i]X_{1^{(2j+1)}}[i])}}
\\&+ \sum_{j=1}{\sum_{i=1}^{n}{H(Y_{1^{(2j)}}[i]|Y_{2^{(j)}}[i]X_{2^{(j)}}[i])}}+C
\\& \overset{(a)}{\leq} \sum_{j}{\sum_{i=1}^{n}{H(Y_{1^{(2j+1)}}[i]|V_{3^{(1)}}[i])}}
\\&+ \sum_{j}{\sum_{i=1}^{n}{H(Y_{2^{(j)}}[i]|V_{2^{(j)}}[i]V_{3^{(1)}}[i])}}
\\&+\sum_{j}{\sum_{i=1}^{n}{H(Y_{1^{(2j)}}[i]|V_{1^{(2j)}}[i]V_{2^{(j)}}[i]V_{3^{(1)}}[i])}}+C
\\& \overset{(b)}{\leq} \sum_{i=1}^{n}{kH(Y_{1^{(1)}}[i]|V_{3^{(1)}}[i])}
\\&+ \sum_{i=1}^{n}{kH(Y_{2^{(1)}}[i]|V_{2^{(1)}}[i]V_{3^{(1)}}[i])}
\\&+\sum_{i=1}^{n}{kH(Y_{1^{(1)}}[i]|V_{1^{(1)}}[i]V_{2^{(1)}}[i]V_{3^{(1)}}[i])}+C
\end{aligned}
\end{equation}
\normalsize
where step $(a)$ follows from $(V_{2^{(j+1)}}[i],V_{3^{(1)}}[i])=h_1(X_{1^{(2j+2)}}[i],Y_{1^{(2j+2)}}[i])$, $(V_{2^{(j)}}[i],V_{3^{(1)}}[i])=h_1(X_{1^{(2j+1)}}[i],Y_{1^{(2j+1)}}[i])$, $(V_{1^{(2j)}}[i],V_{3^{(1)}}[i])=h_2(X_{2^{(j)}}[i],Y_{2^{(j)}}[i])$, and
$V_{2^{(j)}}[i]=g_2(X_{2^{(j)}}[i])$. Step $(b)$ follows, since the coding schemes of $S_{1^{(j)}}$s and $S_{2^{(j)}}$s are the same respectively, $Y_{1^{(j)}}[i]$s and $Y_{2^{(j)}}[i]$s have the same probability distribution respectively.

\section{Title of Appendix A} \label{App:AppendixA}
\label{apb14}
\noindent \textbf{Proof of part (b) for bound (\ref{ineq5}).}
\small
\begin{equation} \label{eq3434}
\begin{aligned}
&H(Y_{{\Omega}_1^c}^n)+\sum_{l=2}^{3k+2}{H(Y_{{\Omega}_l^c\bigcap {\Omega}_{l-1}}^n|W_{\mathcal S\backslash{{\Omega}_l}},Y_{{\Omega}_{l-1}^c}^n)}
\\& \leq \sum_{j}{H(Y_{3^{(j)}}^n|Y_{1^{(j+1)}}^nX_{1^{(j+1)}}^nX_{2^{(j+1)}}^n)}
\\&+ \sum_{j}{H(Y_{2^{(j)}}^n|Y_{1^{(j+1)}}^nX_{1^{(j+1)}}^nX_{3^{(j)}}^n)}
\\&+ \sum_{j}{H(Y_{1^{(j)}}^n|Y_{2^{(j)}}^nX_{2^{(j)}}^n)}
\\& \leq \sum_{j=1}{\sum_{i=1}^{n}{H(Y_{3^{(j)}}[i]|Y_{1^{(j+1)}}[i]X_{1^{(j+1)}}[i]X_{2^{(j+1)}}[i])}}\nonumber
\end{aligned}
\end{equation}
\begin{equation} \label{eq343413}
\begin{aligned}
\\&+ \sum_{j=1}{\sum_{i=1}^{n}{H(Y_{2^{(j)}}[i]|Y_{1^{(j+1)}}[i]X_{1^{(j+1)}}[i]X_{3^{(j)}}[i])}}
\\&+ \sum_{j=1}{\sum_{i=1}^{n}{H(Y_{1^{(j)}}[i]|Y_{2^{(j)}}[i]X_{2^{(j)}}[i])}}
\\& \overset{(a)}{\leq} \sum_{j}{\sum_{i=1}^{n}{H(Y_{3^{(j)}}[i]|V_{1^{(j+1)}}[i]V_{2^{(j+1)}}[i]V_{3^{(j)}}[i])}}
\\&+ \sum_{j}{\sum_{i=1}^{n}{H(Y_{2^{(j)}}[i]|V_{2^{(j)}}[i]V_{3^{(j)}}[i])+
H(Y_{1^{(j)}}[i]|V_{1^{(j)}}[i])}}+C
\\& \overset{(b)}{\leq} \sum_{i=1}^{n}{kH(Y_{3^{(1)}}[i]|V_{1^{(1)}}[i],V_{2^{(1)}}[i],V_{3^{(1)}}[i])}
\\&+ \sum_{i=1}^{n}{kH(Y_{2^{(1)}}[i]|V_{2^{(1)}}[i],V_{3^{(1)}}[i])
+kH(Y_{1^{(1)}}[i]|V_{1^{(1)}}[i])}+C
\end{aligned}
\end{equation}
\normalsize
where step $(a)$ follows from $(V_{2^{(j)}}[i],V_{3^{(j)}}[i])=h_1(X_{1^{(j+1)}}[i],Y_{1^{(j+1)}}[i])$, $(V_{1^{(j)}}[i],V_{3^{(j)}}[i])=h_2(X_{2^{(j)}}[i],Y_{2^{(j)}}[i])$, $V_{3^{(j)}}[i]=g_3(X_{3^{(j)}}[i])$, $V_{1^{(j+1)}}[i]=g_1(X_{1^{(j+1)}}[i])$, and $V_{2^{(j+1)}}[i]=g_2(X_{2^{(j+1)}}[i])$. Step $(b)$ follows, since the coding schemes of $S_{1^{(j)}}$s, $S_{2^{(j)}}$s, and $S_{3^{(j)}}$s are the same respectively, $Y_{1^{(j)}}[i]$s, $Y_{2^{(j)}}[i]$s, and $Y_{3^{(j)}}[i]$s have the same probability distribution respectively.

\section{Title of Appendix A} \label{App:AppendixA}
\label{apb15}
\noindent \textbf{Proof of part (b) for bound (\ref{ineq6}).}
\small
\begin{equation} \label{eq3434}
\begin{aligned}
&H(Y_{{\Omega}_1^c}^n)+\sum_{l=2}^{3k+2}{H(Y_{{\Omega}_l^c\bigcap {\Omega}_{l-1}}^n|W_{\mathcal S\backslash{{\Omega}_l}},Y_{{\Omega}_{l-1}^c}^n)}
\\& \leq \sum_{j=1}{H(Y_{2^{(j)}}^n|Y_{1^{(j+1)}}^n,X_{1^{(j+1)}}^n)+H(Y_{3^{(j)}}^n|Y_{2^{(j)}}^n,X_{2^{(j)}}^n)}
\\&+ \sum_{j=1}{H(Y_{1^{(j)}}^n|Y_{3^{(j)}}^n,X_{3^{(j)}}^n)}+C
\\& \leq \sum_{j}{\sum_{i=1}^{n}{H(Y_{2^{(j)}}[i]|Y_{1^{(j+1)}}[i],X_{1^{(j+1)}}[i])}}
\\&+ \sum_{j}{\sum_{i=1}^{n}{H(Y_{3^{(j)}}[i]|Y_{2^{(j)}}[i],X_{2^{(j)}}[i])}}
\\&+ \sum_{j}{\sum_{i=1}^{n}{H(Y_{1^{(j)}}[i]|Y_{3^{(j)}}[i],X_{3^{(j)}}[i])}}+C
\\& \overset{(a)}{\leq} \sum_{j}{\sum_{i=1}^{n}{H(Y_{2^{(j)}}[i]|V_{2^{(j)}}[i],V_{1^{(j+1)}}[i])}}
\\&+ \sum_{j}{\sum_{i=1}^{n}{H(Y_{3^{(j)}}[i]|V_{3^{(j)}}[i],V_{2^{(j)}}[i])
+H(Y_{1^{(j)}}[i]|V_{1^{(j)}}[i],V_{3^{(j)}}[i])}}\\&+C
\\& \overset{(b)}{\leq} \sum_{i=1}^{n}{kH(Y_{3^{(1)}}[i]|V_{2^{(1)}}[i],V_{3^{(1)}}[i])+kH(Y_{2^{(1)}}[i]|V_{2^{(1)}}[i],V_{1^{(1)}}[i])}
\\&+ \sum_{i=1}^{n}{kH(Y_{1^{(1)}}[i]|V_{1^{(1)}}[i],V_{3^{(1)}}[i])}+C
\end{aligned}
\end{equation}
\normalsize
where step $(a)$ follows from $(V_{2^{(j)}}[i],V_{3^{(j+1)}}[i])=h_1(X_{1^{(j+1)}}[i],Y_{1^{(j+1)}}[i])$, $(V_{1^{(j+1)}}[i],V_{3^{(j)}}[i])=h_2(X_{2^{(j)}}[i],Y_{2^{(j)}}[i])$, $(V_{1^{(j)}}[i],V_{2^{(j)}}[i])=h_3(X_{3^{(j)}}[i],Y_{3^{(j)}}[i])$, $V_{3^{(j)}}[i]=g_3(X_{3^{(j)}}[i])$, $V_{1^{(j+1)}}[i]=g_1(X_{1^{(j+1)}}[i])$, and $V_{2^{(j)}}[i]=g_2(X_{2^{(j)}}[i])$. Step $(b)$ follows, since the coding schemes of $S_{1^{(j)}}$s, $S_{2^{(j)}}$s, and $S_{3^{(j)}}$s are the same respectively, $Y_{1^{(j)}}[i]$s, $Y_{2^{(j)}}[i]$s, and $Y_{3^{(j)}}[i]$s have the same probability distribution respectively.
\\
\section{Title of Appendix A} \label{App:AppendixA}
\label{apb16}
\noindent \textbf{Proof of part (b) for bound (\ref{ineq7}).}
\small
\begin{equation} \label{eq3434}
\begin{aligned}
&H(Y_{{\Omega}_1^c}^n)+\sum_{l=2}^{3k+2}{H(Y_{{\Omega}_l^c\bigcap {\Omega}_{l-1}}^n|W_{\mathcal S\backslash{{\Omega}_l}},Y_{{\Omega}_{l-1}^c}^n)}
\\& \leq \sum_{j}{H(Y_{2^{(j)}}^n|Y_{1^{(j+1)}}^nX_{1^{(j+1)}}^n)+H(Y_{3^{(j)}}^n|Y_{2^{(j)}}^nX_{2^{(j)}}^n)}
\\&+ \sum_{j}{H(Y_{1^{(j)}}^n|Y_{3^{(j)}}^nX_{3^{(j)}}^nX_{2^{(j)}}^n)}
\\& \leq \sum_{j}{\sum_{i=1}^{n}{H(Y_{2^{(j)}}[i]|Y_{1^{(j+1)}}[i]X_{1^{(j+1)}}[i])}}
\\&+\sum_{j}{\sum_{i=1}^{n}{H(Y_{3^{(j)}}[i]|Y_{2^{(j)}}[i]X_{2^{(j)}}[i])}}
\\&+ \sum_{j=1}{\sum_{i=1}^{n}{H(Y_{1^{(j)}}[i]|Y_{3^{(j)}}[i]X_{3^{(j)}}[i]X_{2^{(j)}}[i])}}
\\& \overset{(a)}{\leq} \sum_{j}{\sum_{i=1}^{n}{H(Y_{2^{(j)}}[i]|V_{1^{(j+1)}}[i])+H(Y_{3^{(j)}}[i]|V_{2^{(j)}}[i]V_{3^{(j)}}[i])}}
\\&+ \sum_{j}{\sum_{i=1}^{n}{H(Y_{1^{(j)}}[i]|V_{1^{(j)}}[i]V_{2^{(j)}}[i]V_{3^{(j)}}[i])}}+C
\\& \overset{(b)}{\leq} \sum_{i=1}^{n}{kH(Y_{2^{(1)}}[i]|V_{1^{(1)}}[i])+kH(Y_{3^{(1)}}[i]|V_{2^{(1)}}[i]V_{3^{(1)}}[i])}
\\&+ \sum_{i=1}^{n}{kH(Y_{1^{(1)}}[i]|V_{1^{(1)}}[i]V_{2^{(1)}}[i]V_{3^{(1)}}[i])}+C
\end{aligned}
\end{equation}
\normalsize
where step $(a)$ follows from $(V_{1^{(j+1)}}[i],V_{3^{(j)}}[i])=h_2(X_{2^{(j)}}[i],Y_{2^{(j)}}[i])$, $(V_{1^{(j)}}[i],V_{2^{(j)}}[i])=h_3(X_{3^{(j)}}[i],Y_{3^{(j)}}[i])$, $V_{3^{(j)}}[i]=g_3(X_{3^{(j)}}[i])$, $V_{1^{(j+1)}}[i]=g_1(X_{1^{(j+1)}}[i])$, and $V_{2^{(j)}}[i]=g_2(X_{2^{(j)}}[i])$. Step $(b)$ follows, since the coding schemes of $S_{1^{(j)}}$s, $S_{2^{(j)}}$s, and $S_{3^{(j)}}$s are the same respectively, $Y_{1^{(j)}}[i]$s, $Y_{2^{(j)}}[i]$s, and $Y_{3^{(j)}}[i]$s have the same probability distribution respectively.

\section{Title of Appendix A} \label{App:AppendixA}
\label{apb17}
\noindent \textbf{Proof of part (b) for bound (\ref{ineq8}).}
\small
\begin{equation} \label{eq3434}
\begin{aligned}
&H(Y_{{\Omega}_1^c}^n)+\sum_{l=2}^{3}{H(Y_{{\Omega}_l^c\bigcap {\Omega}_{l-1}}^n|W_{\mathcal S\backslash{{\Omega}_l}},Y_{{\Omega}_{l-1}^c}^n)}
\\& \leq \sum_{j}{H(Y_{3^{(1)}}^n)+H(Y_{2^{(1)}}^n|Y_{3^{(1)}}^nX_{3^{(1)}}^n)}
\\&+ \sum_{j}{H(Y_{1^{(1)}}^n|Y_{2^{(1)}}^nX_{2^{(1)}}^nX_{3^{(1)}}^n)}
\\& \leq \sum_{j}{\sum_{i=1}^{n}{H(Y_{3^{(1)}}[i])}}
\\&+\sum_{j}{\sum_{i=1}^{n}{H(Y_{2^{(1)}}[i]|Y_{3^{(1)}}[i]X_{3^{(1)}}[i])}}
\\&+ \sum_{j=1}{\sum_{i=1}^{n}{H(Y_{1^{(1)}}[i]|Y_{2^{(1)}}[i]X_{2^{(1)}}[i]X_{3^{(1)}}[i])}}
\\& \overset{(a)}{\leq} \sum_{j}{\sum_{i=1}^{n}{H(Y_{3^{(1)}}[i])+H(Y_{2^{(1)}}[i]|V_{1^{(1)}}[i]V_{2^{(1)}}[i]V_{3^{(1)}}[i])}}
\\&+ \sum_{j}{\sum_{i=1}^{n}{H(Y_{1^{(1)}}[i]|V_{1^{(1)}}[i]V_{2^{(1)}}[i]V_{3^{(1)}}[i])}}
\end{aligned}
\end{equation}
\normalsize
where step $(a)$ follows from $(V_{1^{(1)}}[i],V_{3^{(1)}}[i])=h_2(X_{2^{(1)}}[i],Y_{2^{(1)}}[i])$, $(V_{1^{(1)}}[i],V_{2^{(1)}}[i])=h_3(X_{3^{(1)}}[i],Y_{3^{(1)}}[i])$, $V_{3^{(1)}}[i]=g_3(X_{3^{(1)}}[i])$, and $V_{2^{(1)}}[i]=g_2(X_{2^{(1)}}[i])$.

\section{Title of Appendix A} \label{App:AppendixA}
\label{apb17}
\noindent \textbf{Proof of part (b) for bound (\ref{ineq9}).}
\small
\begin{equation} \label{eq3434}
\begin{aligned}
&H(Y_{{\Omega}_1^c}^n)+H(Y_{{\Omega}_2^c\bigcap {\Omega}_1}^n|W_{\mathcal S\backslash{{\Omega}_2}},Y_{{\Omega}_1^c}^n)\\&+H(Y_{{\Omega}_3^c\bigcap {\Omega}_2}^n|W_{\mathcal S\backslash{{\Omega}_3}},Y_{{\Omega}_2^c}^n)+H(Y_{{\Omega}_4^c\bigcap {\Omega}_3}^n|W_{\mathcal S\backslash{{\Omega}_4}},Y_{{\Omega}_3^c}^n)
\\&=H(Y_{1^{(2)}}^n)+H(Y_{3^{(1)}}^n|W_{1^{(2)}},Y_{1^{(2)}}^n)\\&+H(Y_{2^{(1)}}^n|W_{3^{(1)}},W_{1^{(2)}},Y_{3^{(1)}}^n,Y_{1^{(2)}}^n)
\\&+H(Y_{1^{(1)}}^n|W_{2^{(1)}},W_{3^{(1)}},W_{1^{(2)}},Y_{2^{(1)}}^n,Y_{3^{(1)}}^n,Y_{1^{(2)}}^n)
\\&=H(Y_{1^{(2)}}^n)+H(Y_{3^{(1)}}^n|W_{1^{(2)}},Y_{1^{(2)}}^n,X_{1^{(2)}}^n)\\&+H(Y_{2^{(1)}}^n|W_{3^{(1)}},W_{1^{(2)}},Y_{3^{(1)}}^n,Y_{1^{(2)}}^n,X_{3^{(1)}}^n,X_{1^{(2)}}^n)
\\&+H(Y_{1^{(1)}}^n|W_{2^{(1)}},W_{3^{(1)}},W_{1^{(2)}},Y_{2^{(1)}}^n,Y_{3^{(1)}}^n,Y_{1^{(2)}}^n,X_{2^{(1)}}^n,X_{3^{(1)}}^n,X_{1^{(2)}}^n)
\\&\leq H(Y_{1^{(2)}}^n)+H(Y_{3^{(1)}}^n|Y_{1^{(2)}}^n,X_{1^{(2)}}^n)\\&+H(Y_{2^{(1)}}^n|Y_{3^{(1)}}^n,X_{3^{(1)}}^n,Y_{1^{(2)}}^n,X_{1^{(2)}}^n)
\\&+H(Y_{1^{(1)}}^n|Y_{2^{(1)}}^n,X_{2^{(1)}}^n,Y_{3^{(1)}}^n,X_{3^{(1)}}^n,Y_{1^{(2)}}^n,X_{1^{(2)}}^n)
\\&= \sum_{i}{H(Y_{1^{(2)}}[i]|Y_{1^{(2)}}^{i-1})}+\sum_{i}{H(Y_{3^{(1)}}[i]|Y_{3^{(1)}}^{i-1},Y_{1^{(2)}}^n,X_{1^{(2)}}^n)}\\&~~~+\sum_{i}{H(Y_{2^{(1)}}[i]|Y_{2^{(1)}}^{i-1},Y_{3^{(1)}}^n,X_{3^{(1)}}^n,Y_{1^{(2)}}^n,X_{1^{(2)}}^n)}
\\&~~~+\sum_{i}{H(Y_{1^{(1)}}[i]|Y_{1^{(1)}}^{i-1},Y_{2^{(1)}}^n,X_{2^{(1)}}^n,Y_{3^{(1)}}^n,X_{3^{(1)}}^n,Y_{1^{(2)}}^n,X_{1^{(2)}}^n)}
\\& \overset{(a)}{\leq}\sum_{i}{H(Y_{1^{(2)}}[i])}+\sum_{i}{H(Y_{3^{(1)}}[i]|V_{1^{(2)}}[i],V_{2^{(1)}}[i],V_{3^{(1)}}[i])}\\&+\sum_{i}{H(Y_{2^{(1)}}[i]|V_{2^{(1)}}[i],V_{3^{(1)}}[i])}
\\&+\sum_{i}{H(Y_{1^{(1)}}[i]|V_{1^{(1)}}[i],V_{2^{(1)}}[i],V_{3^{(1)}}[i])}
\\&\overset{(b)}{\leq}\sum_{i}{H(Y_{1^{(1)}}[i])}+\sum_{i}{H(Y_{3^{(1)}}[i]|V_{1^{(1)}}[i],V_{2^{(1)}}[i],V_{3^{(1)}}[i])}\\&+\sum_{i}{H(Y_{2^{(1)}}[i]|V_{2^{(1)}}[i],V_{3^{(1)}}[i])}
\\&+\sum_{i}{H(Y_{1^{(1)}}[i]|V_{1^{(1)}}[i],V_{2^{(1)}}[i],V_{3^{(1)}}[i])}
\end{aligned}
\end{equation}
\normalsize
where $(a)$ follows from $V_{1^{(2)}}[i]=g_1(X_{1^{(2)}}[i])$, $(V_{2^{(1)}}[i],V_{3^{(1)}}[i])=h_1(X_{1^{(2)}}[i],Y_{1^{(2)}}[i])$, $V_{3^{(1)}}[i]=g_3(X_{3^{(1)}}[i])$, $V_{2^{(1)}}[i]=g_2(X_{2^{(1)}}[i])$, and $(V_{1^{(1)}}[i],V_{3^{(1)}}[i])=h_2(X_{2^{(1)}}[i],Y_{2^{(1)}}[i])$. Step $(b)$ follows, since the coding schemes of $S_{1^{(2)}}$ is the same as $S_{1^{(1)}}$, $Y_{1^{(2)}}[i]$ and $Y_{1^{(1)}}[i]$ have the same probability distribution. It should be noted that $V_{1^{(2)}}[i]$ and $V_{1^{(1)}}[i]$ have the same probability distribution as well.

\section{Title of Appendix A} \label{App:AppendixA}
\label{apb19}
\noindent \textbf{Proof of part (b) for bound (\ref{ineq10}).}
\small
\begin{equation} \label{eq3434}
\begin{aligned}
&H(Y_{{\Omega}_1^c}^n)+\sum_{l=2}^{4k+2}{H(Y_{{\Omega}_l^c\bigcap {\Omega}_{l-1}}^n|W_{\mathcal S\backslash{{\Omega}_l}},Y_{{\Omega}_{l-1}^c}^n)}
\\& \leq \sum_{j}{H(Y_{1^{(2j+1)}}^n|Y_{3^{(j+1)}}^n,X_{3^{(j+1)}}^n)+H(Y_{3^{(j)}}^n|Y_{1^{(2j+1)}}^n,X_{1^{(2j+1)}}^n)}
\\&+ \sum_{j}{H(Y_{2^{(j)}}^n|Y_{3^{(j)}}^n,X_{3^{(j)}}^n)+H(Y_{1^{(2j)}}^n|Y_{2^{(j)}}^n,X_{2^{(j)}}^n,X_{3^{(j)}}^n)}+C\nonumber
\end{aligned}
\end{equation}
\begin{equation} \label{eq343413}
\begin{aligned}
\\& \leq \sum_{j}{\sum_{i=1}^{n}{H(Y_{1^{(2j+1)}}[i]|Y_{3^{(j+1)}}[i],X_{3^{(j+1)}}[i])}}
\\&+ \sum_{j}{\sum_{i=1}^{n}{H(Y_{3^{(j)}}[i]|Y_{1^{(2j+1)}}[i],X_{1^{(2j+1)}}[i])}}
\\&+ \sum_{j}{\sum_{i=1}^{n}{H(Y_{2^{(j)}}[i]|Y_{3^{(j)}}[i],X_{3^{(j)}}[i])}}
\\&+ \sum_{j}{\sum_{i=1}^{n}{H(Y_{1^{(2j)}}[i]|Y_{2^{(j)}}[i],X_{2^{(j)}}[i],X_{3^{(j)}}[i])}}+C
\\& \overset{(a)}{\leq} \sum_{j=1}{\sum_{i=1}^{n}{H(Y_{1^{(2j+1)}}[i]|V_{1^{(2j+1)}}[i])+H(Y_{3^{(j)}}[i]|V_{2^{(j)}}[i],V_{3^{(j)}}[i])}}
\\&+\sum_{j=1}{\sum_{i=1}^{n}{H(Y_{2^{(j)}}[i]|V_{2^{(j)}}[i],V_{3^{(j)}}[i])}}
\\&+\sum_{j=1}{\sum_{i=1}^{n}{H(Y_{1^{(2j)}}[i]|V_{1^{(2j)}}[i],V_{2^{(j)}}[i],V_{3^{(j)}}[i])}}+C
\\& \overset{(b)}{\leq} \sum_{i=1}^{n}{kH(Y_{1^{(1)}}[i]|V_{1^{(1)}}[i])+kH(Y_{3^{(1)}}[i]|V_{2^{(1)}}[i],V_{3^{(1)}}[i])}
\\&+ \sum_{i=1}^{n}{kH(Y_{2^{(1)}}[i]|V_{2^{(1)}}[i],V_{3^{(1)}}[i])}
\\&+ \sum_{i=1}^{n}{kH(Y_{1^{(1)}}[i]|V_{1^{(1)}}[i],V_{2^{(1)}}[i],V_{3^{(1)}}[i])}+C
\end{aligned}
\end{equation}
\normalsize
where step $(a)$ follows from $(V_{2^{(j)}}[i],V_{3^{(j)}}[i])=h_1(X_{1^{(2j+1)}}[i],Y_{1^{(2j+1)}}[i])$, $(V_{1^{(2j)}}[i],V_{3^{(j)}}[i])=h_2(X_{2^{(j)}}[i],Y_{2^{(j)}}[i])$, $(V_{1^{(2j+1)}}[i],V_{2^{(j+1)}}[i])=h_3(X_{3^{(j+1)}}[i],Y_{3^{(j+1)}}[i])$, $V_{3^{(j)}}[i]=g_3(X_{3^{(j)}}[i])$, $V_{1^{(2j+1)}}[i]=g_1(X_{1^{(2j+1)}}[i])$, and $V_{2^{(j)}}[i]=g_2(X_{2^{(j)}}[i])$. Step $(b)$ follows, since the coding schemes of $S_{1^{(j)}}$s, $S_{2^{(j)}}$s, and $S_{3^{(j)}}$s are the same respectively, $Y_{1^{(j)}}[i]$s, $Y_{2^{(j)}}[i]$s, and $Y_{3^{(j)}}[i]$s have the same probability distribution respectively.

\section{Title of Appendix A} \label{App:AppendixA}
\label{apb20}
\noindent \textbf{Proof of part (b) for bound (\ref{ineq11}).}
\small
\begin{equation} \label{eq3434}
\begin{aligned}
&H(Y_{{\Omega}_1^c}^n)+\sum_{l=2}^{4k+2}{H(Y_{{\Omega}_l^c\bigcap {\Omega}_{l-1}}^n|W_{\mathcal S\backslash{{\Omega}_l}},Y_{{\Omega}_{l-1}^c}^n)}
\\& \leq \sum_{j}{H(Y_{1^{(2j+1)}}^n|X_{2^{(j+1)}}^n)+H(Y_{3^{(j)}}^n|Y_{1^{(2j+1)}}^n,X_{1^{(2j+1)}}^n,X_{2^{(j+1)}}^n)}
\\&+ \sum_{j}{H(Y_{1^{(2j)}}^n|Y_{3^{(j)}}^n,X_{3^{(j)}}^n)+H(Y_{2^{(j)}}^n|Y_{1^{(2j)}}^n,X_{1^{(2j)}}^n,X_{3^{(j)}}^n)}
\\& \leq \sum_{j}{\sum_{i=1}^{n}{H(Y_{1^{(2j+1)}}[i]|X_{2^{(j+1)}}[i])}}
\\&+\sum_{j}{\sum_{i=1}^{n}{H(Y_{3^{(j)}}[i]|Y_{1^{(2j+1)}}[i],X_{1^{(2j+1)}}[i],X_{2^{(j+1)}}[i])}}
\\&+ \sum_{j}{\sum_{i=1}^{n}{H(Y_{1^{(2j)}}[i]|Y_{3^{(j)}}[i],X_{3^{(j)}}[i])}}
\\&+ \sum_{j}{\sum_{i=1}^{n}{H(Y_{2^{(j)}}[i]|Y_{1^{(2j)}}[i],X_{1^{(2j)}}[i],X_{3^{(j)}}[i])}}\nonumber
\end{aligned}
\end{equation}
\begin{equation} \label{eq343413}
\begin{aligned}
\\& \overset{(a)}{\leq} \sum_{j}{\sum_{i=1}^{n}{H(Y_{1^{(2j+1)}}[i]|V_{2^{(j+1)}}[i])+H(Y_{3^{(j)}}[i]|V_{2^{(j+1)}}[i],V_{3^{(j)}}[i])}}
\\&+\sum_{j}{\sum_{i=1}^{n}{H(Y_{1^{(2j)}}[i]|V_{1^{(2j)}}[i],V_{3^{(j)}}[i])}}
\\&+\sum_{j}{\sum_{i=1}^{n}{H(Y_{2^{(j)}}[i]|V_{1^{(2j)}}[i],V_{2^{(j)}}[i],V_{3^{(j)}}[i])}}+C
\\& \overset{(b)}{\leq} \sum_{i=1}^{n}{kH(Y_{1^{(1)}}[i]|V_{1^{(1)}}[i],V_{3^{(1)}}[i])+kH(Y_{3^{(1)}}[i]|V_{2^{(1)}}[i],V_{3^{(1)}}[i])}
\\&+ \sum_{i=1}^{n}{kH(Y_{2^{(1)}}[i]|V_{1^{(1)}}[i],V_{2^{(1)}}[i],V_{3^{(1)}}[i])
+kH(Y_{1^{(1)}}[i]|V_{2^{(1)}}[i])}\\&+C
\end{aligned}
\end{equation}
\normalsize
where step $(a)$ follows from $(V_{2^{(j+1)}}[i],V_{3^{(j)}}[i])=h_1(X_{1^{(2j+1)}}[i],Y_{1^{(2j+1)}}[i])$, $(V_{2^{(j)}}[i],V_{3^{(j)}}[i])=h_1(X_{1^{(2j)}}[i],Y_{1^{(2j)}}[i])$, $(V_{1^{(2j)}}[i],V_{2^{(j+1)}}[i])=h_3(X_{3^{(j)}}[i],Y_{3^{(j)}}[i])$, $V_{3^{(j)}}[i]=g_3(X_{3^{(j)}}[i])$, $V_{1^{(2j)}}[i]=g_1(X_{1^{(2j)}}[i])$, and $V_{2^{(j+1)}}[i]=g_2(X_{2^{(j+1)}}[i])$. Step $(b)$ follows, since the coding schemes of $S_{1^{(j)}}$s, $S_{2^{(j)}}$s, and $S_{3^{(j)}}$s are the same respectively, $Y_{1^{(j)}}[i]$s, $Y_{2^{(j)}}[i]$s, and $Y_{3^{(j)}}[i]$s have the same probability distribution respectively.

\section{Title of Appendix A} \label{App:AppendixA}
\label{apb21}
\noindent \textbf{Proof of part (b) for bound (\ref{ineq12}).}
\small
\begin{equation} \label{eq3434}
\begin{aligned}
&H(Y_{{\Omega}_1^c}^n)+\sum_{l=2}^{4k+2}{H(Y_{{\Omega}_l^c\bigcap {\Omega}_{l-1}}^n|W_{\mathcal S\backslash{{\Omega}_l}},Y_{{\Omega}_{l-1}^c}^n)+n\epsilon_n}
\\& \leq \sum_{j}{H(Y_{1^{(2j+1)}}^n|X_{3^{(j+1)}}^n)}
\\&+\sum_{j}{H(Y_{3^{(j)}}^n|Y_{1^{(2j+1)}}^n,X_{1^{(2j+1)}}^n,Y_{2^{(j+1)}}^n,X_{2^{(j+1)}}^n)}
\\&+ \sum_{j}{H(Y_{2^{(j)}}^n|Y_{3^{(j)}}^n,X_{3^{(j)}}^n,X_{1^{(2j+1)}}^n)}
\\&+ \sum_{j}{H(Y_{1^{(2j)}}^n|Y_{3^{(j)}}^n,X_{3^{(j)}}^n,X_{2^{(j)}}^n)}
\\& \leq \sum_{j}{\sum_{i=1}^{n}{H(Y_{1^{(2j+1)}}[i]|X_{3^{(j+1)}}[i])}}
\\&+ \sum_{j}{\sum_{i=1}^{n}{H(Y_{3^{(j)}}[i]|Y_{1^{(2j+1)}}[i],X_{1^{(2j+1)}}[i],Y_{2^{(j+1)}}[i],X_{2^{(j+1)}}[i])}}
\\&+ \sum_{j}{\sum_{i=1}^{n}{H(Y_{2^{(j)}}[i]|Y_{3^{(j)}}[i],X_{3^{(j)}}[i],X_{1^{(2j+1)}}[i])}}
\\&+ \sum_{j}{\sum_{i=1}^{n}{H(Y_{1^{(2j)}}[i]|Y_{3^{(j)}}[i],X_{3^{(j)}}[i],X_{2^{(j)}}[i])}}
\\& \overset{(a)}{\leq} \sum_{j}{\sum_{i=1}^{n}{H(Y_{1^{(2j+1)}}[i]|V_{3^{(j+1)}}[i])+H(Y_{3^{(j)}}[i]|V_{2^{(j)}}[i],V_{3^{(j)}}[i])}}
\\&+ \sum_{j}{\sum_{i=1}^{n}{H(Y_{2^{(j)}}[i]|V_{1^{(2j+1)}}[i],V_{2^{(j)}}[i])}}
\\&+ \sum_{j}{\sum_{i=1}^{n}{H(Y_{1^{(2j)}}^n|V_{1^{(2j)}}[i],V_{2^{(j)}}[i],V_{3^{(j)}}[i])}}\nonumber
\end{aligned}
\end{equation}
\begin{equation} \label{eq343413}
\begin{aligned}
\\& \overset{(b)}{\leq} \sum_{i=1}^{n}{kH(Y_{1^{(1)}}[i]|V_{3^{(1)}}[i])+kH(Y_{3^{(1)}}[i]|V_{2^{(1)}}[i],V_{3^{(1)}}[i])}
\\&+ \sum_{i=1}^{n}{kH(Y_{2^{(1)}}[i]|V_{1^{(1)}}[i],V_{2^{(1)}}[i])}
\\&+ \sum_{i=1}^{n}{kH(Y_{1^{(1)}}[i]|V_{1^{(1)}}[i],V_{2^{(1)}}[i],V_{3^{(1)}}[i])}
\end{aligned}
\end{equation}
\normalsize
where step $(a)$ follows from $(V_{2^{(j)}}[i],V_{3^{(j+1)}}[i])=h_1(X_{1^{(2j+1)}}[i],Y_{1^{(2j+1)}}[i])$, $(V_{1^{(2j+1)}}[i],V_{3^{(j-1)}}[i])=h_2(X_{2^{(j)}}[i],Y_{2^{(j)}}[i])$, $(V_{1^{(2j)}}[i],V_{2^{(j)}}[i])=h_3(X_{3^{(j)}}[i],Y_{3^{(j)}}[i])$, $V_{3^{(j+1)}}[i]=g_3(X_{3^{(j+1)}}[i])$, $V_{1^{(2j+1)}}[i]=g_1(X_{1^{(2j+1)}}[i])$, and $V_{2^{(j)}}[i]=g_2(X_{2^{(j)}}[i])$. Step $(b)$ follows, since the coding schemes of $S_{1^{(j)}}$s, $S_{2^{(j)}}$s, and $S_{3^{(j)}}$s are the same respectively, $Y_{1^{(j)}}[i]$s, $Y_{2^{(j)}}[i]$s, and $Y_{3^{(j)}}[i]$s have the same probability distribution respectively.

\section{Title of Appendix A} \label{App:AppendixA}
\label{apb22}
\noindent \textbf{Proof of part (b) for bound (\ref{ineq13}).}
\small
\begin{equation} \label{eq3434}
\begin{aligned}
&H(Y_{{\Omega}_1^c}^n)+\sum_{l=2}^{4k+2}{H(Y_{{\Omega}_l^c\bigcap {\Omega}_{l-1}}^n|W_{\mathcal S\backslash{{\Omega}_l}},Y_{{\Omega}_{l-1}^c}^n)}
\\& \leq \sum_{j}{H(Y_{1^{(2j+1)}}^n|X_{3^{(j+1)}}^n)+H(Y_{2^{(j)}}^n|Y_{1^{(2j+1)}}^n,X_{1^{(2j+1)}}^n)}
\\&+ \sum_{j}{H(Y_{3^{(j)}}^n|Y_{2^{(j)}}^n,X_{2^{(j)}}^n,X_{1^{(2j+1)}}^n)}
\\&+ \sum_{j}{H(Y_{1^{(2j)}}^n|Y_{2^{(j)}}^n,X_{2^{(j)}}^n,X_{3^{(j)}}^n)}
\\& \leq \sum_{j}{\sum_{i=1}^{n}{H(Y_{1^{(2j+1)}}[i]|X_{3^{(j+1)}}[i])}}
\\&+ \sum_{j}{\sum_{i=1}^{n}{H(Y_{2^{(j)}}[i]|Y_{1^{(2j+1)}}[i],X_{1^{(2j+1)}}[i])}}
\\&+ \sum_{j}{\sum_{i=1}^{n}{H(Y_{3^{(j)}}[i]|Y_{2^{(j)}}[i],X_{2^{(j)}}[i],X_{1^{(2j+1)}}[i])}}
\\&+ \sum_{j}{\sum_{i=1}^{n}{H(Y_{1^{(2j)}}[i]|Y_{2^{(j)}}[i],X_{2^{(j)}}[i],X_{3^{(j)}}[i])}}
\\& \overset{(a)}{\leq} \sum_{j}{\sum_{i=1}^{n}{H(Y_{1^{(2j+1)}}[i]|V_{3^{(j+1)}}[i])+H(Y_{2^{(j)}}[i]|V_{2^{(j)}}[i])}}
\\&+ \sum_{j}{\sum_{i=1}^{n}{H(Y_{3^{(j)}}[i]|V_{1^{(2j+1)}}[i],V_{2^{(j)}}[i],V_{3^{(j)}}[i])}}
\\&+ \sum_{j}{\sum_{i=1}^{n}{H(Y_{1^{(2j)}}[i]|V_{1^{(2j)}}[i],V_{2^{(j)}}[i],V_{3^{(j)}}[i])}}
\\& \overset{(b)}{\leq} \sum_{i=1}^{n}{kH(Y_{1^{(1)}}[i]|V_{3^{(1)}}[i])+kH(Y_{2^{(1)}}[i]|V_{2^{(1)}}[i])}
\\&+ \sum_{i=1}^{n}{kH(Y_{3^{(1)}}[i]|V_{1^{(1)}}[i],V_{2^{(1)}}[i],V_{3^{(1)}}[i])}
\\&+ \sum_{i=1}^{n}{kH(Y_{1^{(1)}}[i]|V_{1^{(1)}}[i],V_{2^{(1)}}[i],V_{3^{(1)}}[i])}
\end{aligned}
\end{equation}
\normalsize
where step $(a)$ follows from $(V_{2^{(j)}}[i],V_{3^{(j+1)}}[i])=h_1(X_{1^{(2j+1)}}[i],Y_{1^{(2j+1)}}[i])$, $(V_{1^{(2j)}}[i],V_{3^{(j)}}[i])=h_2(X_{2^{(j)}}[i],Y_{2^{(j)}}[i])$, $(V_{1^{(2j+1)}}[i],V_{2^{(j)}}[i])=h_3(X_{3^{(j)}}[i],Y_{3^{(j)}}[i])$, $V_{3^{(j+1)}}[i]=g_3(X_{3^{(j+1)}}[i])$, $V_{1^{(2j+1)}}[i]=g_1(X_{1^{(2j+1)}}[i])$, and $V_{2^{(j)}}[i]=g_2(X_{2^{(j)}}[i])$. Step $(b)$ follows, since the coding schemes of $S_{1^{(j)}}$s, $S_{2^{(j)}}$s, and $S_{3^{(j)}}$s are the same respectively, $Y_{1^{(j)}}[i]$s, $Y_{2^{(j)}}[i]$s, and $Y_{3^{(j)}}[i]$s have the same probability distribution respectively.

\section{Title of Appendix A} \label{App:AppendixA}
\label{apb23}
\noindent \textbf{Proof of part (b) for bound (\ref{ineq14}).}
\small
\begin{equation} \label{eq3434}
\begin{aligned}
&H(Y_{{\Omega}_1^c}^n)+\sum_{l=2}^{4k+2}{H(Y_{{\Omega}_l^c\bigcap {\Omega}_{l-1}}^n|W_{\mathcal S\backslash{{\Omega}_l}},Y_{{\Omega}_{l-1}^c}^n)}
\\& \leq \sum_{j}{H(Y_{3^{(j)}}^n|Y_{1^{(2j+3)}}^n,X_{1^{(2j+3)}}^n)+H(Y_{1^{(2j+1)}}^n|Y_{3^{(j)}}^n,X_{3^{(j)}}^n)}
\\&+ \sum_{j}{H(Y_{2^{(j)}}^n|Y_{3^{(j)}}^n,X_{3^{(j)}}^n)+H(Y_{1^{(2j)}}^n|Y_{2^{(j)}}^n,X_{2^{(j)}}^n,X_{3^{(j)}}^n)}
\\& \leq \sum_{j=1}{\sum_{i=1}^{n}{H(Y_{3^{(j)}}[i]|Y_{1^{(2j+3)}}[i],X_{1^{(2j+3)}}[i])}}
\\&+ \sum_{j=1}{\sum_{i=1}^{n}{H(Y_{1^{(2j+1)}}[i]|Y_{3^{(j)}}[i],X_{3^{(j)}}[i])}}
\\&+ \sum_{j=1}{\sum_{i=1}^{n}{H(Y_{2^{(j)}}[i]|Y_{3^{(j)}}[i],X_{3^{(j)}}[i])}}
\\&+ \sum_{j=1}{\sum_{i=1}^{n}{H(Y_{1^{(2j)}}[i]|Y_{2^{(j)}}^n,X_{2^{(j)}}^n,X_{3^{(j)}}^n)}}
\\& \overset{(a)}{\leq} \sum_{j=1}{\sum_{i=1}^{n}{H(Y_{3^{(j)}}[i]|V_{3^{(j)}}[i])+H(Y_{1^{(2j+1)}}[i]|V_{1^{(2j+1)}}[i],V_{2^{(j)}}[i])}}
\\&+ \sum_{j=1}{\sum_{i=1}^{n}{H(Y_{2^{(j)}}[i]|V_{2^{(j)}}[i],V_{3^{(j)}}[i])}}
\\&+ \sum_{j=1}{\sum_{i=1}^{n}{H(Y_{1^{(2j)}}[i]|V_{1^{(2j)}}[i],V_{2^{(j)}}[i],V_{3^{(j)}}[i])}}
\\& \overset{(b)}{\leq} \sum_{i=1}^{n}{kH(Y_{3^{(1)}}[i]|V_{3^{(1)}}[i])+kH(Y_{1^{(1)}}[i]|V_{1^{(1)}}[i],V_{2^{(1)}}[i])}
\\&+ \sum_{i=1}^{n}{kH(Y_{2^{(1)}}[i]|V_{2^{(1)}}[i],V_{3^{(1)}}[i])}
\\&+ \sum_{i=1}^{n}{kH(Y_{1^{(1)}}[i]|V_{1^{(1)}}[i],V_{2^{(1)}}[i],V_{3^{(1)}}[i])}
\end{aligned}
\end{equation}
\normalsize
where step $(a)$ follows from $(V_{2^{(j+1)}}[i],V_{3^{(j)}}[i])=h_1(X_{1^{(2j+3)}}[i],Y_{1^{(2j+3)}}[i])$, $(V_{2^{(j)}}[i],V_{3^{(j-1)}}[i])=h_1(X_{1^{(2j+1)}}[i],Y_{1^{(2j+1)}}[i])$, $(V_{1^{(2j)}}[i],V_{3^{(j)}}[i])=h_2(X_{2^{(j)}}[i],Y_{2^{(j)}}[i])$, $(V_{1^{(2j+1)}}[i],V_{2^{(j)}}[i])=h_3(X_{3^{(j)}}[i],Y_{3^{(j)}}[i])$, $V_{3^{(j+1)}}[i]=g_3(X_{3^{(j+1)}}[i])$, $V_{1^{(2j+1)}}[i]=g_1(X_{1^{(2j+1)}}[i])$, and $V_{2^{(j)}}[i]=g_2(X_{2^{(j)}}[i])$. Step $(b)$ follows, since the coding schemes of $S_{1^{(j)}}$s, $S_{2^{(j)}}$s, and $S_{3^{(j)}}$s are the same respectively, $Y_{1^{(j)}}[i]$s, $Y_{2^{(j)}}[i]$s, and $Y_{3^{(j)}}[i]$s have the same probability distribution respectively.

\section{Title of Appendix A} \label{App:AppendixA}
\label{apb24}
\noindent \textbf{Proof of part (b) for bound (\ref{ineq15}).}
\small
\begin{equation} \label{eq3434}
\begin{aligned}
&H(Y_{{\Omega}_1^c}^n)+\sum_{l=2}^{4}{H(Y_{{\Omega}_l^c\bigcap {\Omega}_{l-1}}^n|W_{\mathcal S\backslash{{\Omega}_l}},Y_{{\Omega}_{l-1}^c}^n)}
\\&\leq H(Y_{2^{(1)}}^n)+H(Y_{1^{(2)}}^n|Y_{2^{(1)}}^n,X_{2^{(1)}}^n)+H(Y_{3^{(1)}}^n|Y_{2^{(1)}}^n,X_{2^{(1)}}^n)
\\&+H(Y_{1^{(1)}}^n|Y_{3^{1}}^n,X_{3^{1}}^n,X_{2^{1}}^n)\nonumber
\end{aligned}
\end{equation}
\begin{equation} \label{eq343413}
\begin{aligned}
\\&\overset{(a)}{\leq} \sum_{i=1}^{n}{H(Y_{2^{(1)}}[i])+H(Y_{1^{(2)}}[i]|V_{1^{(2)}}[i],V_{2^{(1)}}[i],V_{3^{(1)}}[i])}
\\&+\sum_{i=1}^{n}{H(Y_{3^{(1)}}[i]|V_{2^{(1)}}[i],V_{3^{(1)}}[i])}
\\&+\sum_{i=1}^{n}{+H(Y_{1^{(1)}}[i]|V_{1^{(1)}}[i],V_{2^{(1)}}[i],V_{3^{(1)}}[i])}
\\&\overset{(b)}{\leq} \sum_{i=1}^{n}{H(Y_{2^{(1)}}[i])+H(Y_{1^{(1)}}[i]|V_{1^{(1)}}[i],V_{2^{(1)}}[i],V_{3^{(1)}}[i])}
\\&+\sum_{i=1}^{n}{H(Y_{3^{(1)}}[i]|V_{2^{(1)}}[i],V_{3^{(1)}}[i])}
\\&+\sum_{i=1}^{n}{H(Y_{1^{(1)}}[i]|V_{1^{(1)}}[i],V_{2^{(1)}}[i],V_{3^{(1)}}[i])}
\end{aligned}
\end{equation}
\normalsize
where $(a)$ follows from $V_{1^{(2)}}[i]=g_1(X_{1^{(2)}}[i])$, $(V_{2^{(1)}}[i],V_{3^{(1)}}[i])=h_1(X_{1^{(2)}}[i],Y_{1^{(2)}}[i])$, $V_{3^{(1)}}[i]=g_3(X_{3^{(1)}}[i])$, $V_{2^{(1)}}[i]=g_2(X_{2^{(1)}}[i])$, $(V_{1^{(1)}}[i],V_{2^{(1)}}[i])=h_3(X_{3^{(1)}}[i],Y_{3^{(1)}}[i])$, and $(V_{1^{(1)}}[i],V_{3^{(1)}}[i])=h_2(X_{2^{(1)}}[i],Y_{2^{(1)}}[i])$. Step $(b)$ follows, since the coding schemes of $S_{1^{(j)}}$ and $S_{2^{(j)}}$ are the same respectively, $Y_{1^{(j)}}[i]$s and $Y_{2^{(j)}}[i]$s have the same probability distribution respectively.

\section{Title of Appendix A} \label{App:AppendixA}
\label{apb25}
\noindent \textbf{Proof of part (b) for bound (\ref{ineq16}).}
\small
\begin{equation} \label{eq3434}
\begin{aligned}
&H(Y_{{\Omega}_1^c}^n)+\sum_{l=2}^{4k+2}{H(Y_{{\Omega}_l^c\bigcap {\Omega}_{l-1}}^n|W_{\mathcal S\backslash{{\Omega}_l}},Y_{{\Omega}_{l-1}^c}^n)}
\\& \leq \sum_{j}{H(Y_{3^{(j)}}^n|Y_{2^{(j+1)}}^n,X_{2^{(j+1)}}^n)+H(Y_{2^{(j)}}^n|Y_{3^{(j)}}^n,X_{3^{(j)}}^n)}
\\&+ \sum_{j}{H(Y_{1^{(2j+1)}}^n|Y_{3^{(j)}}^n,X_{3^{(j)}}^n,X_{2^{(j)}}^n)}
\\&+ \sum_{j}{H(Y_{1^{(2j)}}^n|Y_{2^{(j)}}^n,X_{2^{(j)}}^n,X_{3^{(j)}}^n)}
\\& \overset{(a)}{\leq} \sum_{j}{\sum_{i=1}^{n}{H(Y_{3^{(j)}}[i]|V_{3^{(j)}}[i])+H(Y_{2^{(j)}}[i]|V_{2^{(j)}}[i])}}
\\&+ \sum_{j}{\sum_{i=1}^{n}{H(Y_{1^{(2j+1)}}[i]|V_{1^{(2j+1)}}[i],V_{2^{(j)}}[i],V_{3^{(j)}}[i])}}
\\&+ \sum_{j}{\sum_{i=1}^{n}{H(Y_{1^{(2j)}}[i]|V_{1^{(2j)}}[i],V_{2^{(j)}}[i],V_{3^{(j)}}[i])}}
\\& \overset{(b)}{\leq} \sum_{i=1}^{n}{kH(Y_{3^{(1)}}[i]|V_{3^{(1)}}[i])+kH(Y_{2^{(1)}}[i]|V_{2^{(1)}}[i])}
\\&+ \sum_{i=1}^{n}{kH(Y_{1^{(1)}}[i]|V_{1^{(1)}}[i],V_{2^{(1)}}[i],V_{3^{(1)}}[i])}
\\&+ \sum_{i=1}^{n}{kH(Y_{1^{(1)}}[i]|V_{1^{(1)}}[i],V_{2^{(1)}}[i],V_{3^{(1)}}[i])}
\end{aligned}
\end{equation}
\normalsize
where step $(a)$ follows from $(V_{1^{(2j)}}[i],V_{3^{(j-1)}}[i])=h_2(X_{2^{(j)}}[i],Y_{2^{(j)}}[i])$,  $(V_{1^{(2j)}}[i],V_{3^{(j)}}[i])=h_2(X_{2^{(j+1)}}[i],Y_{2^{(j+1)}}[i])$, $(V_{1^{(2j+1)}}[i],V_{2^{(j)}}[i])=h_3(X_{3^{(j)}}[i],Y_{3^{(j)}}[i])$, $V_{3^{(j)}}[i]=g_3(X_{3^{(j)}}[i])$, and $V_{2^{(j)}}[i]=g_2(X_{2^{(j)}}[i])$. Step $(b)$ follows, since the coding schemes of $S_{1^{(j)}}$s, $S_{2^{(j)}}$s, and $S_{3^{(j)}}$s are the same respectively, $Y_{1^{(j)}}[i]$s, $Y_{2^{(j)}}[i]$s, and $Y_{3^{(j)}}[i]$s have the same probability distribution respectively.

\section{Title of Appendix A} \label{App:AppendixA}
\label{apb26}
\noindent \textbf{Proof of part (b) for bound (\ref{ineq17}).}
\small
\begin{equation} \label{eq3434}
\begin{aligned}
&H(Y_{{\Omega}_1^c}^n)+\sum_{l=2}^{5}{H(Y_{{\Omega}_l^c\bigcap {\Omega}_{l-1}}^n|W_{\mathcal S\backslash{{\Omega}_l}},Y_{{\Omega}_{l-1}^c}^n)}
\\&\leq H(Y_{1^{(3)}}^n)+H(Y_{3^{(1)}}^n|Y_{1^{(3)}}^n,X_{1^{(3)}}^n)+H(Y_{1^{(2)}}^n|Y_{3^{(1)}}^n,X_{3^{(1)}}^n)
\\&+H(Y_{2^{(1)}}^n|Y_{3^{1}}^n,X_{3^{1}}^n)+H(Y_{1^{(1)}}^n|Y_{2^{(1)}}^n,X_{2^{(1)}}^n,X_{3^{(1)}}^n)
\\&\overset{(a)}{\leq} \sum_{i=1}^{n}{H(Y_{1^{(3)}}[i])+H(Y_{3^{(1)}}[i]|V_{2^{(1)}}[i],V_{3^{(1)}}[i])}
\\&+ \sum_{i=1}^{n}{H(Y_{1^{(2)}}[i]|V_{1^{(2)}}[i],V_{2^{(1)}}[i],V_{3^{(1)}}[i])}
\\&+ \sum_{i=1}^{n}{H(Y_{2^{(1)}}[i]|V_{2^{(1)}}[i],V_{3^{(1)}}[i])}
\\&+\sum_{i=1}^{n}{H(Y_{1^{(1)}}[i]|V_{1^{(1)}}[i],V_{2^{(1)}}[i],V_{3^{(1)}}[i])}
\\&\overset{(b)}{\leq} \sum_{i=1}^{n}{H(Y_{1^{(1)}}[i])+H(Y_{3^{(1)}}[i]|V_{2^{(1)}}[i],V_{3^{(1)}}[i])}
\\&+ \sum_{i=1}^{n}{H(Y_{1^{(1)}}[i]|V_{1^{(1)}}[i],V_{2^{(1)}}[i],V_{3^{(1)}}[i])}
\\&+ \sum_{i=1}^{n}{H(Y_{2^{(1)}}[i]|V_{2^{(1)}}[i],V_{3^{(1)}}[i])}
\\&+\sum_{i=1}^{n}{H(Y_{1^{(1)}}[i]|V_{1^{(1)}}[i],V_{2^{(1)}}[i],V_{3^{(1)}}[i])}
\end{aligned}
\end{equation}
\normalsize
where $(a)$ follows from $V_{1^{(2)}}[i]=g_1(X_{1^{(2)}}[i])$, $(V_{2^{(1)}}[i],V_{3^{(1)}}[i])=h_1(X_{1^{(3)}}[i],Y_{1^{(3)}}[i])$, $V_{3^{(1)}}[i]=g_3(X_{3^{(1)}}[i])$, $V_{2^{(1)}}[i]=g_2(X_{2^{(1)}}[i])$, $(V_{1^{(2)}}[i],V_{2^{(1)}}[i])=h_3(X_{3^{(1)}}[i],Y_{3^{(1)}}[i])$, and $(V_{1^{(1)}}[i],V_{3^{(1)}}[i])=h_2(X_{2^{(1)}}[i],Y_{2^{(1)}}[i])$. Step $(b)$ follows, since the coding schemes of $S_{1^{(j)}}$ and $S_{2^{(j)}}$ are the same respectively, $Y_{1^{(j)}}[i]$s and $Y_{2^{(j)}}[i]$s have the same probability distribution respectively.

\section{Title of Appendix A} \label{App:AppendixA}
\label{apb27}
\noindent \textbf{Proof of part (b) for bound (\ref{ineq18}).}
\small
\begin{equation} \label{eq3434}
\begin{aligned}
&H(Y_{{\Omega}_1^c}^n)+\sum_{l=2}^{5k+2}{H(Y_{{\Omega}_l^c\bigcap {\Omega}_{l-1}}^n|W_{\mathcal S\backslash{{\Omega}_l}},Y_{{\Omega}_{l-1}^c}^n)}
\\& \leq \sum_{j}{H(Y_{1^{(3j+2)}}^n|X_{2^{(j+1)}}^n)+H(Y_{3^{(j)}}^n|Y_{1^{(3j+2)}}^n,X_{1^{(3j+2)}}^n)}
\\&+ \sum_{j}{H(Y_{2^{(j)}}^n|Y_{3^{(j)}}^n,X_{3^{(j)}}^n)+H(Y_{1^{(3j+1)}}^n|Y_{3^{(j)}}^n,X_{3^{(j)}}^n,X_{2^{(j)}}^n)}
\\&+ \sum_{j}{H(Y_{1^{(3j)}}^n|Y_{2^{(j)}}^n,X_{2^{(j)}}^n,X_{3^{(j)}}^n)}
\\& \overset{(a)}{\leq} \sum_{j}{\sum_{i=1}^{n}{H(Y_{1^{(3j+2)}}[i]|V_{2^{(j+1)}}[i])+H(Y_{3^{(j)}}[i]|V_{3^{(j)}}[i])}}
\\&+ \sum_{j}{\sum_{i=1}^{n}{H(Y_{2^{(j)}}[i]|V_{2^{(j)}}[i],V_{3^{(j)}}[i])}}\nonumber
\end{aligned}
\end{equation}
\begin{equation} \label{eq343413}
\begin{aligned}
\\&+ \sum_{j}{\sum_{i=1}^{n}{H(Y_{1^{(3j+1)}}[i]|V_{1^{(3j+1)}}[i],V_{2^{(j)}}[i],V_{3^{(j)}}[i])}}
\\&+\sum_{j}{\sum_{i=1}^{n}{H(Y_{1^{(3j)}}[i]|V_{1^{(3j)}}[i],V_{2^{(j)}}[i],V_{3^{(j)}}[i])}}
\\& \overset{(b)}{\leq} \sum_{i=1}^{n}{kH(Y_{1^{(1)}}[i]|V_{2^{(1)}}[i])+2kH(Y_{1^{(1)}}[i]|V_{1^{(1)}}[i],V_{2^{(1)}}[i],V_{3^{(1)}}[i])}
\\&+ \sum_{i=1}^{n}{kH(Y_{2^{(1)}}[i]|V_{2^{(1)}}[i],V_{3^{(1)}}[i])
+kH(Y_{3^{(1)}}[i]|V_{3^{(1)}}[i])}
\end{aligned}
\end{equation}
\normalsize
where step $(a)$ follows from $(V_{1^{(3j)}}[i],V_{3^{(j)}}[i])=h_2(X_{2^{(j)}}[i],Y_{2^{(j)}}[i])$,  $(V_{2^{(j)}}[i],V_{3^{(j)}}[i])=h_1(X_{1^{(3j+1)}}[i],Y_{1^{(3j+1)}}[i])$, $(V_{1^{(3j+1)}}[i],V_{2^{(j)}}[i])=h_3(X_{3^{(j)}}[i],Y_{3^{(j)}}[i])$, $V_{3^{(j)}}[i]=g_3(X_{3^{(j)}}[i])$, and $V_{2^{(j)}}[i]=g_2(X_{2^{(j)}}[i])$. Step $(b)$ follows, since the coding schemes of $S_{1^{(j)}}$s, $S_{2^{(j)}}$s, and $S_{3^{(j)}}$s are the same respectively, $Y_{1^{(j)}}[i]$s, $Y_{2^{(j)}}[i]$s, and $Y_{3^{(j)}}[i]$s have the same probability distribution respectively.

\section{Title of Appendix A} \label{App:AppendixA}
\noindent \textbf{Proof of part (b) for bound (\ref{ineq19}).}
\small
\begin{equation} \label{eq3434}
\begin{aligned}
&H(Y_{{\Omega}_1^c}^n)+\sum_{l=2}^{5}{H(Y_{{\Omega}_l^c\bigcap {\Omega}_{l-1}}^n|W_{\mathcal S\backslash{{\Omega}_l}},Y_{{\Omega}_{l-1}^c}^n)}
\\&=H(Y_{{1^{(2)}}}^n)+H(Y_{{3^{(1)}}}^n|W_{{1^{(2)}}},Y_{{1^{(2)}}}^n)\\&+H(Y_{{1^{(1)}}}^n|W_{[{3^{(1)}},{1^{(2)}}]},Y_{[{3^{(1)}},{1^{(2)}}]}^n)
\\&+H(Y_{{2^{(2)}}}^n|W_{[{1^{(1)}},{3^{(1)}},{1^{(2)}}]},Y_{[{1^{(1)}},{3^{(1)}},{1^{(2)}}]}^n)
\\&+H(Y_{{2^{(1)}}}^n|W_{[{2^{(2)}},{1^{(1)}},{3^{(1)}},{1^{(2)}}]},Y_{[{2^{(2)}},{1^{(1)}},{3^{(1)}},{1^{(2)}}]}^n)
\\&=H(Y_{{1^{(2)}}}^n)+H(Y_{{3^{(1)}}}^n|W_{{1^{(2)}}},Y_{{1^{(2)}}}^n,X_{{1^{(2)}}}^n)\\&+H(Y_{{1^{(1)}}}^n|W_{[{3^{(1)}},{1^{(2)}}]},Y_{[{3^{(1)}},{1^{(2)}}]}^n,X_{[{3^{(1)}},{1^{(2)}}]}^n)
\\&+H(Y_{{2^{(2)}}}^n|W_{[{1^{(1)}},{3^{(1)}},{1^{(2)}}]},Y_{[{1^{(1)}},{3^{(1)}},{1^{(2)}}]}^n,X_{[{1^{(1)}},{3^{(1)}},{1^{(2)}}]}^n)
\\&+H(Y_{{2^{(1)}}}^n|W_{[{2^{(2)}},{1^{(1)}},{3^{(1)}},{1^{(2)}}]},Y_{[{2^{(2)}},{1^{(1)}},{3^{(1)}},{1^{(2)}}]}^n,X_{[{2^{(2)}},{1^{(1)}},{3^{(1)}},{1^{(2)}}]}^n)
\\&\leq H(Y_{{1^{(2)}}}^n)+H(Y_{{3^{(1)}}}^n|Y_{{1^{(2)}}}^n,X_{{1^{(2)}}}^n)
\\&+H(Y_{{1^{(1)}}}^n|Y_{[{3^{(1)}},{1^{(2)}}]}^n,X_{[{3^{(1)}},{1^{(2)}}]}^n)
\\&+H(Y_{{2^{(2)}}}^n|Y_{[{1^{(1)}},{3^{(1)}},{1^{(2)}}]}^n,X_{[{1^{(1)}},{3^{(1)}},{1^{(2)}}]}^n)
\\&+H(Y_{{2^{(1)}}}^n|Y_{[{2^{(2)}},{1^{(1)}},{3^{(1)}},{1^{(2)}}]}^n,X_{[{2^{(2)}},{1^{(1)}},{3^{(1)}},{1^{(2)}}]}^n)
\\&= \sum_{i}{H(Y_{{1^{(2)}}}[i]|Y_{{1^{(2)}}}^{i-1})}
\\&~~~+\sum_{i}{H(Y_{{3^{(1)}}}[i]|Y_{{3^{(1)}}}^{i-1},Y_{{1^{(2)}}}^n,X_{{1^{(2)}}}^n)}\\&~~~+\sum_{i}{H(Y_{{1^{(1)}}}[i]|Y_{{1^{(1)}}}^{i-1},Y_{[{3^{(1)}},{1^{(2)}}]}^n,X_{[{3^{(1)}},{1^{(2)}}]}^n)}
\\&~~~+\sum_{i}{H(Y_{{2^{(2)}}}[i]|Y_{{2^{(2)}}}^{i-1},Y_{[{1^{(1)}},{3^{(1)}},{1^{(2)}}]}^n,X_{[{1^{(1)}},{3^{(1)}},{1^{(2)}}]}^n)}
\\&~~~+\sum_{i}{H(Y_{{2^{(1)}}}[i]|Y_{{2^{(1)}}}^{i-1},Y_{[{2^{(2)}},{1^{(1)}},{3^{(1)}},{1^{(2)}}]}^n,X_{[{2^{(2)}},{1^{(1)}},{3^{(1)}},{1^{(2)}}]}^n)}
\\&\overset{(a)}{\leq}\sum_{i}{H(Y_{{1^{(2)}}}[i])}+\sum_{i}{H(Y_{{3^{(1)}}}|V_{{3^{(1)}}}[i],V_{{2^{(2)}}}[i])}\\&+\sum_{i}{H(Y_{1^{(1)}}[i]|V_{1^{(1)}}[i],V_{3^{(1)}}[i])}
\\&+\sum_{i}{H(Y_{2^{(2)}}[i]|V_{2^{(2)}}[i],V_{1^{(1)}}[i],V_{3^{(1)}}[i])}\\&+\sum_{i}{H(Y_{2^{(1)}}[i]|V_{2^{(1)}}[i],V_{1^{(1)}}[i],V_{3^{(1)}}[i])}\nonumber
\end{aligned}
\end{equation}
\begin{equation} \label{eq343413}
\begin{aligned}
\\&\overset{(b)}{\leq}\sum_{i}{Y_{{1^{(1)}}}[i])}+\sum_{i}{H(Y_{{3^{(1)}}}|V_{{3^{(1)}}}[i],V_{{2^{(1)}}}[i])}\\&+\sum_{i}{H(Y_{1^{(1)}}[i]|V_{1^{(1)}}[i],V_{3^{(1)}}[i])}
\\&+2\sum_{i}{H(Y_{2^{(1)}}[i]|V_{2^{(1)}}[i],V_{1^{(1)}}[i],V_{3^{(1)}}[i])}
\end{aligned}
\end{equation}
\normalsize
where step $(a)$ follows from $(V_{{2^{(2)}}}[i],V_{{3^{(1)}}}[i])=h_1(X_{{1^{(2)}}}[i],Y_{{1^{(2)}}}[i])$, $(V_{{1^{(1)}}}[i],V_{{2^{(2)}}}[i])=h_3(X_{{3^{(1)}}}[i],Y_{{3^{(1)}}}[i])$, $V_{{3^{(1)}}}[i]=g_3(X_{{3^{(1)}}}[i])$, $V_{{1^{(1)}}}[i]=g_1(X_{{1^{(1)}}}[i])$, and $(V_{{2^{(1)}}}[i],V_{{3^{(1)}}}[i])=h_1(X_{{1^{(1)}}}[i],Y_{{1^{(1)}}}[i])$ . Step $(b)$ follows, since the coding schemes of $S_{1^{(j)}}$s and $S_{2^{(j)}}$s are the same respectively, $Y_{1^{(j)}}[i]$s and $Y_{2^{(j)}}[i]$s have the same probability distribution respectively. Note that $V_{2^{(2)}}[i]$ and $V_{2^{(1)}}[i]$ also have the same probability distribution.

\section{Title of Appendix A} \label{App:AppendixA}
\label{apb29}
\noindent \textbf{Proof of part (b) for bound (\ref{ineq20}).}
\small
\begin{equation} \label{eq3434}
\begin{aligned}
&H(Y_{{\Omega}_1^c}^n)+\sum_{l=2}^{5}{H(Y_{{\Omega}_l^c\bigcap {\Omega}_{l-1}}^n|W_{\mathcal S\backslash{{\Omega}_l}},Y_{{\Omega}_{l-1}^c}^n)}
\\ &\leq H(Y_{1^{(2)}}^n)+H(Y_{2^{(2)}}^n|Y_{1^{(2)}}^n,X_{1^{(2)}}^n)+H(Y_{3^{(1)}}^n|Y_{1^{(2)}}^n,X_{1^{(2)}}^n)
\\&+H(Y_{2^{(1)}}^n|Y_{3^{1}}^n,X_{3^{1}}^n)+H(Y_{1^{(1)}}^n|Y_{2^{(1)}}^n,X_{2^{(1)}}^n,X_{3^{(1)}}^n)
\\&\overset{(a)}{\leq} \sum_{i=1}^{n}{H(Y_{1^{(2)}}[i])+H(Y_{2^{(2)}}[i]|V_{2^{(2)}}[i],V_{3^{(1)}}[i])}
\\&+ \sum_{i=1}^{n}{H(Y_{3^{(1)}}[i]|V_{1^{(2)}}[i]V_{3^{(1)}}[i])+H(Y_{2^{(1)}}[i]|V_{2^{(1)}}[i]V_{3^{(1)}}[i])}
\\&+\sum_{i=1}^{n}{H(Y_{1^{(1)}}[i]|V_{1^{(1)}}[i],V_{2^{(1)}}[i],V_{3^{(1)}}[i])}
\\&\overset{(b)}{\leq} \sum_{i=1}^{n}{H(Y_{1^{(1)}}[i])+H(Y_{2^{(1)}}[i]|V_{1^{(1)}}[i],V_{2^{(1)}}[i],V_{3^{(1)}}[i])}
\\&+ \sum_{i=1}^{n}{H(Y_{3^{(1)}}[i]|V_{1^{(1)}}[i],V_{3^{(1)}}[i])+H(Y_{2^{(1)}}[i]|V_{2^{(1)}}[i],V_{3^{(1)}}[i])}
\\&+\sum_{i=1}^{n}{H(Y_{1^{(1)}}[i]|V_{1^{(1)}}[i],V_{2^{(1)}}[i],V_{3^{(1)}}[i])}
\end{aligned}
\end{equation}
\normalsize
where $(a)$ follows from $(V_{1^{(1)}}[i],V_{3^{(1)}}[i])=h_2(X_{2^{(1)}}[i],Y_{2^{(1)}}[i])$, $(V_{1^{(2)}}[i],V_{3^{(1)}}[i])=h_2(X_{2^{(2)}}[i],Y_{2^{(2)}}[i])$, $V_{3^{(1)}}[i]=g_3(X_{3^{(1)}}[i])$, $V_{2^{(1)}}[i]=g_2(X_{2^{(1)}}[i])$, $V_{2^{(2)}}[i]=g_2(X_{2^{(2)}}[i])$, $V_{1^{(2)}}[i]=g_2(X_{1^{(2)}}[i])$, $(V_{2^{(2)}}[i],V_{3^{(1)}}[i])=h_1(X_{1^{(2)}}[i],Y_{1^{(2)}}[i])$, and $(V_{1^{(2)}}[i],V_{2^{(2)}}[i])=h_3(X_{3^{(1)}}[i],Y_{3^{(1)}}[i])$. Step $(b)$ follows, since the coding schemes of $S_{1^{(j)}}$ and $S_{2^{(j)}}$ are the same respectively, $Y_{1^{(j)}}[i]$s and $Y_{2^{(j)}}[i]$s have the same probability distribution respectively.

\section{Title of Appendix A} \label{App:AppendixA}
\label{apb30}
\noindent \textbf{Proof of part (b) for bound (\ref{ineq21}).}
\small
\begin{equation} \label{eq3434}
\begin{aligned}
&H(Y_{{\Omega}_1^c}^n)+\sum_{l=2}^{5}{H(Y_{{\Omega}_l^c\bigcap {\Omega}_{l-1}}^n|W_{\mathcal S\backslash{{\Omega}_l}},Y_{{\Omega}_{l-1}^c}^n)}
\\ &\leq H(Y_{1^{(2)}}^n)+H(Y_{3^{(1)}}^n|Y_{1^{(2)}}^n,X_{1^{(2)}}^n)
\\&+H(Y_{2^{(2)}}^n|Y_{1^{(2)}}^n,X_{1^{(2)}}^n,Y_{3^{(1)}}^n,X_{3^{(1)}}^n)+H(Y_{1^{(1)}}^n|Y_{3^{1}}^n,X_{3^{1}}^n)
\\&+H(Y_{2^{(1)}}^n|Y_{1^{(1)}}^n,X_{1^{(1)}}^n,X_{3^{(1)}}^n)
\\&\overset{(a)}{\leq} \sum_{i=1}^{n}{H(Y_{1^{(2)}}[i])+H(Y_{3^{(1)}}[i]|V_{3^{(1)}}[i])}\nonumber
\end{aligned}
\end{equation}
\begin{equation} \label{eq343413}
\begin{aligned}
\\&+ \sum_{i=1}^{n}{H(Y_{2^{(2)}}[i]|V_{1^{(1)}}[i]V_{2^{(2)}}[i]V_{3^{(1)}}[i])}
\\&+\sum_{i=1}^{n}{H(Y_{1^{(1)}}[i]|V_{1^{(1)}}[i]V_{2^{(1)}}[i]V_{3^{(1)}}[i])}
\\&+ \sum_{i=1}^{n}{H(Y_{2^{(1)}}[i]|V_{1^{(1)}}[i]V_{2^{(1)}}[i]V_{3^{(1)}}[i])}
\\&\overset{(b)}{\leq} \sum_{i=1}^{n}{H(Y_{1^{(1)}}[i])+H(Y_{3^{(1)}}[i]|V_{3^{(1)}}[i])}
\\&+ \sum_{i=1}^{n}{H(Y_{2^{(1)}}[i]|V_{1^{(1)}}[i]V_{2^{(1)}}[i]V_{3^{(1)}}[i])}
\\&+\sum_{i=1}^{n}{H(Y_{1^{(1)}}[i]|V_{1^{(1)}}[i]V_{2^{(1)}}[i]V_{3^{(1)}}[i])}
\\&+ \sum_{i=1}^{n}{H(Y_{2^{(1)}}[i]|V_{1^{(1)}}[i]V_{2^{(1)}}[i]V_{3^{(1)}}[i])}
\end{aligned}
\end{equation}
\normalsize
where $(a)$ follows from $(V_{2^{(2)}}[i],V_{3^{(1)}}[i])=h_1(X_{1^{(2)}}[i],Y_{1^{(2)}}[i])$, $V_{3^{(1)}}[i]=g_3(X_{3^{(1)}}[i])$, $V_{1^{(1)}}[i]=g_1(X_{1^{(1)}}[i])$, $(V_{1^{(1)}}[i],V_{2^{(1)}}[i])=h_3(X_{3^{(1)}}[i],Y_{3^{(1)}}[i])$, $(V_{2^{(1)}}[i],V_{3^{(1)}}[i])=h_1(X_{1^{(1)}}[i],Y_{1^{(1)}}[i])$, and $(V_{1^{(1)}}[i],V_{3^{(1)}}[i])=h_2(X_{2^{(2)}}[i],Y_{2^{(2)}}[i])$. Step $(b)$ follows, since the coding schemes of $S_{1^{(j)}}$ and $S_{2^{(j)}}$ are the same respectively, $Y_{1^{(j)}}[i]$s and $Y_{2^{(j)}}[i]$s have the same probability distribution respectively.

\section{Title of Appendix A} \label{App:AppendixA}
\label{apb31}
\noindent \textbf{Proof of part (b) for bound (\ref{ineq22}).}
\small
\begin{equation} \label{eq3434}
\begin{aligned}
&H(Y_{{\Omega}_1^c}^n)+\sum_{l=2}^{5k+2}{H(Y_{{\Omega}_l^c\bigcap {\Omega}_{l-1}}^n|W_{\mathcal S\backslash{{\Omega}_l}},Y_{{\Omega}_{l-1}^c}^n)}
\\& \leq \sum_{j}{H(Y_{1^{(2j+1)}}^n|Y_{3^{(j+1)}}^nX_{3^{(j+1)}}^n)}
\\&+ \sum_{j}{H(Y_{2^{(2j+1)}}^n|Y_{1^{(2j+1)}}^nX_{1^{(2j+1)}}^n)}
\\&+ \sum_{j}{H(Y_{3^{(j)}}^n|Y_{1^{(2j+1)}}^nX_{1^{(2j+1)}}^n)}
\\&+ \sum_{j}{H(Y_{1^{(2j)}}^n|Y_{2^{(2j+1)}}^nX_{2^{(2j+1)}}^nX_{3^{(j)}}^n)}
\\&+ \sum_{j}{H(Y_{2^{(2j)}}^n|Y_{3^{(j)}}^nX_{3^{(j)}}^nX_{1^{(2j)}}^n)}
\\& \overset{(a)}{\leq} \sum_{j}{\sum_{i=1}^{n}{H(Y_{1^{(2j+1)}}[i]|V_{1^{(2j+1)}}[i])}}
\\&+ \sum_{j}{\sum_{i=1}^{n}{H(Y_{2^{(2j+1)}}[i]|V_{2^{(2j+1)}}[i]V_{3^{(j)}}[i])}}
\\&+ \sum_{j}{\sum_{i=1}^{n}{H(Y_{3^{(j)}}[i]|V_{3^{(j)}}[i])}}
\\&+ \sum_{j}{\sum_{i=1}^{n}{H(Y_{1^{(2j)}}[i]|V_{1^{(2j)}}[i],V_{2^{(2j+1)}}[i],V_{3^{(j)}}[i])}}
\\&+\sum_{j}{\sum_{i=1}^{n}{H(Y_{2^{(2j)}}[i]|V_{1^{(2j)}}[i],V_{2^{(2j)}}[i],V_{3^{(j)}}[i])}}\nonumber
\end{aligned}
\end{equation}
\begin{equation} \label{eq343413}
\begin{aligned}
\\& \overset{(b)}{\leq} \sum_{i=1}^{n}{kH(Y_{1^{(1)}}[i]|V_{1^{(1)}}[i])+kH(Y_{2^{(1)}}[i]|V_{2^{(1)}}[i]V_{3^{(1)}}[i])}
\\&+ \sum_{i=1}^{n}{kH(Y_{3^{(1)}}[i]|V_{3^{(1)}}[i])}
\\&+ \sum_{i=1}^{n}{H(Y_{1^{(1)}}[i]|V_{1^{(1)}}[i],V_{2^{(1)}}[i],V_{3^{(1)}}[i])}
\\&+ \sum_{i=1}^{n}{kH(Y_{2^{(1)}}[i]|V_{1^{(1)}}[i],V_{2^{(1)}}[i],V_{3^{(1)}}[i])}
\end{aligned}
\end{equation}
\normalsize
where step $(a)$ follows from $(V_{2^{(2j+1)}}[i],V_{3^{(j)}}[i])=h_1(X_{1^{(2j)}}[i],Y_{1^{(2j)}}[i])$,  $(V_{1^{(2j-1)}}[i],V_{2^{(2j)}}[i])=h_3(X_{3^{(j)}}[i],Y_{3^{(j)}}[i])$, $(V_{1^{(2j)}}[i],V_{3^{(j)}}[i])=h_2(X_{2^{(2j+1)}}[i],Y_{2^{(2j+1)}}[i])$, $(V_{2^{(2j+1)}}[i],V_{3^{(j)}}[i])=h_1(X_{1^{(2j+1)}}[i],Y_{1^{(2j+1)}}[i])$, $V_{3^{(j)}}[i]=g_3(X_{3^{(j)}}[i])$, $V_{1^{(2j)}}[i]=g_1(X_{1^{(2j)}}[i])$, and $V_{2^{(2j+1)}}[i]=g_2(X_{2^{(2j+1)}}[i])$. Step $(b)$ follows, since the coding schemes of $S_{1^{(j)}}$s, $S_{2^{(j)}}$s, and $S_{3^{(j)}}$s are the same respectively, $Y_{1^{(j)}}[i]$s, $Y_{2^{(j)}}[i]$s, and $Y_{3^{(j)}}[i]$s have the same probability distribution respectively.

\section{Title of Appendix A} \label{App:AppendixA}
\label{apb32}
\noindent \textbf{Proof of part (b) for bound (\ref{ineq23}).}
\small
\begin{equation} \label{eq3434}
\begin{aligned}
&H(Y_{{\Omega}_1^c}^n)+\sum_{l=2}^{5k+2}{H(Y_{{\Omega}_l^c\bigcap {\Omega}_{l-1}}^n|W_{\mathcal S\backslash{{\Omega}_l}},Y_{{\Omega}_{l-1}^c}^n)}
\\& \leq \sum_{j}{H(Y_{2^{(2j+1)}}^n|Y_{1^{(2j+3)}}^nX_{1^{(2j+3)}}^n)}
\\&+ \sum_{j}{H(Y_{1^{(2j+1)}}^n|Y_{2^{(2j+1)}}^nX_{2^{(2j+1)}}^n)}
\\&+ \sum_{j}{H(Y_{3^{(j)}}^n|Y_{2^{(2j+1)}}^nX_{2^{(2j+1)}}^n)}
\\&+ \sum_{j}{H(Y_{1^{(2j)}}^n|Y_{3^{(j)}}^nX_{3^{(j)}}^n)}
\\&+ \sum_{j}{H(Y_{2^{(2j)}}^n|Y_{1^{(2j)}}^nX_{1^{(2j)}}^nX_{3^{(j)}}^nX_{1^{(2j+1)}}^n)}
\\& \overset{(a)}{\leq} \sum_{j}{\sum_{i=1}^{n}{H(Y_{2^{(2j+1)}}[i]|V_{2^{(2j+1)}}[i])}}
\\&+ \sum_{j}{\sum_{i=1}^{n}{H(Y_{1^{(2j+1)}}[i]|V_{1^{(2j+1)}}[i]V_{3^{(j)}}[i])}}
\\&+ \sum_{j}{\sum_{i=1}^{n}{H(Y_{3^{(j)}}[i]|V_{2^{(2j+1)}}[i]V_{3^{(j)}}[i])}}
\\&+ \sum_{j}{\sum_{i=1}^{n}{H(Y_{1^{(2j)}}[i]|V_{1^{(2j)}}[i]V_{3^{(j)}}[i])}}
\\&+\sum_{j}{\sum_{i=1}^{n}{H(Y_{2^{(2j)}}[i]|V_{1^{(2j+1)}}[i],V_{2^{(2j)}}[i],V_{3^{(j)}}[i])}}
\\& \overset{(b)}{\leq} \sum_{i=1}^{n}{kH(Y_{2^{(1)}}[i]|V_{2^{(1)}}[i])+kH(Y_{1^{(1)}}[i]|V_{1^{(1)}}[i]V_{3^{(1)}}[i])}
\\&+ \sum_{i=1}^{n}{kH(Y_{3^{(1)}}[i]|V_{2^{(1)}}[i]V_{3^{(1)}}[i])}\nonumber
\\&+ \sum_{i=1}^{n}{kH(Y_{1^{(1)}}[i]|V_{1^{(1)}}[i]V_{3^{(1)}}[i])}
\end{aligned}
\end{equation}
\begin{equation} \label{eq343413}
\begin{aligned}
\\&+ \sum_{i=1}^{n}{kH(Y_{2^{(1)}}[i]|V_{1^{(1)}}[i],V_{2^{(1)}}[i],V_{3^{(1)}}[i])}
\end{aligned}
\end{equation}
\normalsize
where step $(a)$ follows from $(V_{2^{(2j)}}[i],V_{3^{(j)}}[i])=h_1(X_{1^{(2j)}}[i],Y_{1^{(2j)}}[i])$,  $(V_{1^{(2j)}}[i],V_{2^{(2j+1)}}[i])=h_3(X_{3^{(j)}}[i],Y_{3^{(j)}}[i])$, $(V_{1^{(2j+1)}}[i],V_{3^{(j)}}[i])=h_2(X_{2^{(2j+1)}}[i],Y_{2^{(2j+1)}}[i])$, $(V_{2^{(2j-1)}}[i],V_{3^{(j)}}[i])=h_1(X_{1^{(2j+1)}}[i],Y_{1^{(2j+1)}}[i])$, $V_{3^{(j)}}[i]=g_3(X_{3^{(j)}}[i])$, $V_{1^{(2j+1)}}[i]=g_1(X_{1^{(2j+1)}}[i])$, and $V_{2^{(2j+1)}}[i]=g_2(X_{2^{(2j+1)}}[i])$. Step $(b)$ follows, since the coding schemes of $S_{1^{(j)}}$s, $S_{2^{(j)}}$s, and $S_{3^{(j)}}$s are the same respectively, $Y_{1^{(j)}}[i]$s, $Y_{2^{(j)}}[i]$s, and $Y_{3^{(j)}}[i]$s have the same probability distribution respectively.

\section{Title of Appendix A} \label{App:AppendixA}
\label{apb33}
\noindent \textbf{Proof of part (b) for bound (\ref{ineq24}).}
\small
\begin{equation} \label{eq3434}
\begin{aligned}
&H(Y_{{\Omega}_1^c}^n)+\sum_{l=2}^{6}{H(Y_{{\Omega}_l^c\bigcap {\Omega}_{l-1}}^n|W_{\mathcal S\backslash{{\Omega}_l}},Y_{{\Omega}_{l-1}^c}^n)}
\\ &\leq H(Y_{1^{(3)}}^n)+H(Y_{3^{(1)}}^n|Y_{1^{(3)}}^nX_{1^{(3)}}^n)
\\&+H(Y_{2^{(2)}}^n|Y_{1^{(3)}}^n,X_{1^{(3)}}^n,X_{3^{(1)}}^n)
\\&+H(Y_{1^{(2)}}^n|Y_{2^{2}}^n,X_{2^{2}}^n,X_{3^{1}}^n)+H(Y_{2^{(1)}}^n|Y_{3^{(1)}}^n,X_{3^{(1)}}^n)
\\&+H(Y_{1^{(1)}}^n|Y_{2^{(1)}}^n,X_{2^{(1)}}^n,X_{3^{(1)}}^n)
\\&\overset{(a)}{\leq} \sum_{i=1}^{n}{H(Y_{1^{(3)}}[i])+H(Y_{3^{(1)}}[i]|V_{3^{(1)}}[i])}
\\&+ \sum_{i=1}^{n}{H(Y_{2^{(2)}}[i]|V_{2^{(2)}}[i]V_{3^{(1)}}[i])}
\\&+\sum_{i=1}^{n}{H(Y_{1^{(2)}}[i]|V_{1^{(2)}}[i]V_{2^{(2)}}[i]V_{3^{(1)}}[i])}
\\&+ \sum_{i=1}^{n}{H(Y_{2^{(1)}}[i]|V_{1^{(1)}}[i]V_{2^{(1)}}[i]V_{3^{(1)}}[i])}
\\&+ \sum_{i=1}^{n}{H(Y_{1^{(1)}}[i]|V_{1^{(1)}}[i]V_{2^{(1)}}[i]V_{3^{(1)}}[i])
+n\epsilon_n}
\\&\overset{(b)}{\leq} \sum_{i=1}^{n}{H(Y_{1^{(1)}}[i])+H(Y_{3^{(1)}}[i]|V_{3^{(1)}}[i])}
\\&+ \sum_{i=1}^{n}{H(Y_{2^{(1)}}[i]|V_{2^{(1)}}[i]V_{3^{(1)}}[i])}
\\&+\sum_{i=1}^{n}{H(Y_{1^{(1)}}[i]|V_{1^{(1)}}[i]V_{2^{(1)}}[i]V_{3^{(1)}}[i])}
\\&+\sum_{i=1}^{n}{H(Y_{2^{(1)}}[i]|V_{1^{(1)}}[i]V_{2^{(1)}}[i]V_{3^{(1)}}[i])}
\\&+\sum_{i=1}^{n}{H(Y_{1^{(1)}}[i]|V_{1^{(1)}}[i]V_{2^{(1)}}[i]V_{3^{(1)}}[i])
+n\epsilon_n}
\end{aligned}
\end{equation}
\normalsize
where $(a)$ follows from $V_{2^{(2)}}[i]=g_2(X_{2^{(2)}}[i])$, $V_{3^{(1)}}[i]=g_3(X_{3^{(1)}}[i])$, $V_{2^{(1)}}[i]=g_2(X_{2^{(1)}}[i])$, $(V_{2^{(2)}}[i],V_{3^{(1)}}[i])=h_1(X_{1^{(3)}}[i],Y_{1^{(3)}}[i])$, $(V_{1^{(1)}}[i],V_{2^{(1)}}[i])=h_3(X_{3^{(1)}}[i],Y_{3^{(1)}}[i])$, $(V_{1^{(2)}}[i],V_{3^{(1)}}[i])=h_2(X_{2^{(2)}}[i],Y_{2^{(2)}}[i])$, $(V_{1^{(1)}}[i],V_{3^{(1)}}[i])=h_2(X_{2^{(1)}}[i],Y_{2^{(1)}}[i])$, and $(V_{1^{(1)}}[i],V_{3^{(1)}}[i])=h_2(X_{2^{(1)}}[i],Y_{2^{(1)}}[i])$. Step $(b)$ follows, since the coding schemes of $S_{1^{(j)}}$ and $S_{2^{(j)}}$ are the same respectively, $Y_{1^{(j)}}[i]$s and $Y_{2^{(j)}}[i]$s have the same probability distribution respectively.

\section{Title of Appendix A} \label{App:AppendixA}
\label{apb34}
\noindent \textbf{Proof of part (b) for bound (\ref{ineq25}).}
\small
\begin{equation} \label{eq3434}
\begin{aligned}
&H(Y_{{\Omega}_1^c}^n)+\sum_{l=2}^{6}{H(Y_{{\Omega}_l^c\bigcap {\Omega}_{l-1}}^n|W_{\mathcal S\backslash{{\Omega}_l}},Y_{{\Omega}_{l-1}^c}^n)}
\\&\leq H(Y_{1^{(3)}}^n)+H(Y_{3^{(1)}}^n|Y_{1^{(3)}}^n,X_{1^{(3)}}^n)
\\&+H(Y_{2^{(2)}}^n|Y_{1^{(3)}}^n,X_{1^{(3)}}^n,X_{3^{(1)}}^n)
\\&+H(Y_{1^{(2)}}^n|Y_{2^{2}}^n,X_{2^{2}}^n,X_{3^{1}}^n)+H(Y_{2^{(1)}}^n|Y_{3^{(1)}}^n,X_{3^{(1)}}^n)
\\&+H(Y_{1^{(1)}}^n|Y_{2^{(1)}}^n,X_{2^{(1)}}^n,X_{3^{(1)}}^n)
\\&\overset{(a)}{\leq} \sum_{i=1}^{n}{H(Y_{1^{(3)}}[i])+H(Y_{3^{(1)}}[i]|V_{1^{(3)}}[i],V_{3^{(1)}}[i])}
\\&+ \sum_{i=1}^{n}{H(Y_{2^{(2)}}[i]|V_{2^{(2)}}[i],V_{3^{(1)}}[i])}
\\&+\sum_{i=1}^{n}{H(Y_{1^{(2)}}[i]|V_{1^{(2)}}[i],V_{2^{(2)}}[i],V_{3^{(1)}}[i])}
\\&+\sum_{i=1}^{n}{H(Y_{2^{(1)}}[i]|V_{2^{(1)}}[i]V_{3^{(1)}}[i])}
\\&+\sum_{i=1}^{n}{H(Y_{1^{(1)}}[i]|V_{1^{(1)}}[i],V_{2^{(1)}}[i],V_{3^{(1)}}[i]}
\\&\overset{(b)}{\leq} \sum_{i=1}^{n}{H(Y_{1^{(1)}}[i])+H(Y_{3^{(1)}}[i]|V_{1^{(1)}}[i],V_{3^{(1)}}[i])}
\\&+\sum_{i=1}^{n}{H(Y_{2^{(1)}}[i]|V_{2^{(1)}}[i],V_{3^{(1)}}[i])}
\\&+\sum_{i=1}^{n}{H(Y_{1^{(1)}}[i]|V_{1^{(1)}}[i],V_{2^{(1)}}[i],V_{3^{(1)}}[i])}
\\&+\sum_{i=1}^{n}{H(Y_{2^{(1)}}[i]|V_{2^{(1)}}[i]V_{3^{(1)}}[i])}
\\&+\sum_{i=1}^{n}{H(Y_{1^{(1)}}[i]|V_{1^{(1)}}[i],V_{2^{(1)}}[i],V_{3^{(1)}}[i])}
\end{aligned}
\end{equation}
\normalsize
where $(a)$ follows from $V_{2^{(2)}}[i]=g_2(X_{2^{(2)}}[i])$, $V_{3^{(1)}}[i]=g_3(X_{3^{(1)}}[i])$, $V_{2^{(1)}}[i]=g_2(X_{2^{(1)}}[i])$, $(V_{2^{(2)}}[i],V_{3^{(1)}}[i])=h_1(X_{1^{(3)}}[i],Y_{1^{(3)}}[i])$, $(V_{1^{(1)}}[i],V_{2^{(1)}}[i])=h_3(X_{3^{(1)}}[i],Y_{3^{(1)}}[i])$, $(V_{1^{(2)}}[i],V_{3^{(1)}}[i])=h_2(X_{2^{(2)}}[i],Y_{2^{(2)}}[i])$, $(V_{1^{(1)}}[i],V_{3^{(1)}}[i])=h_2(X_{2^{(1)}}[i],Y_{2^{(1)}}[i])$, and $(V_{1^{(3)}}[i],V_{2^{(1)}}[i])=h_3(X_{3^{(1)}}[i],Y_{3^{(1)}}[i])$. Step $(b)$ follows, since the coding schemes of $S_{1^{(j)}}$ and $S_{2^{(j)}}$ are the same respectively, $Y_{1^{(j)}}[i]$s and $Y_{2^{(j)}}[i]$s have the same probability distribution respectively.

\section{Title of Appendix A} \label{App:AppendixA}
\label{apb35}
\noindent \textbf{Proof of part (b) for bound (\ref{ineq26}).}
\small
\begin{equation} \label{eq3434}
\begin{aligned}
&H(Y_{{\Omega}_1^c}^n)+\sum_{l=2}^{6k+2}{H(Y_{{\Omega}_l^c\bigcap {\Omega}_{l-1}}^n|W_{\mathcal S\backslash{{\Omega}_l}},Y_{{\Omega}_{l-1}^c}^n)}
\\& \leq \sum_{j}{H(Y_{1^{(3j+2)}}^n|Y_{3^{(j+1)}}^nX_{3^{(j+1)}}^n)+H(Y_{2^{(2j+1)}}^n|Y_{1^{(3j+2)}}^nX_{1^{(3j+2)}}^n)}
\\&+ \sum_{j}{H(Y_{3^{(j)}}^n|Y_{1^{(3j+2)}}^nX_{1^{(3j+2)}}^n)}
\\&+ \sum_{j}{H(Y_{1^{(3j+1)}}^n|Y_{3^{(j)}}^nX_{3^{(j)}}^nY_{2^{(2j+1)}}^nX_{2^{(2j+1)}}^n)}\nonumber
\end{aligned}
\end{equation}
\begin{equation} \label{eq343413}
\begin{aligned}
\\&+ \sum_{j}{H(Y_{2^{(2j)}}^n|Y_{1^{(3j+1)}}^nX_{1^{(3j+1)}}^nX_{3^{(j)}}^n)}
\\&+ \sum_{j}{H(Y_{1^{(3j)}}^n|Y_{2^{(2j)}}^nX_{2^{(2j)}}^nX_{3^{(j)}}^n)}
\\& \overset{(a)}{\leq} \sum_{j}{\sum_{i=1}^{n}{H(Y_{1^{(3j+2)}}[i]|V_{1^{(3j+2)}}[i])}}
\\&+ \sum_{j}{\sum_{i=1}^{n}{H(Y_{2^{(2j+1)}}[i]|V_{2^{(2j+1)}}[i]V_{3^{(j)}}[i])}}
\\&+ \sum_{j}{\sum_{i=1}^{n}{H(Y_{3^{(j)}}[i]|V_{2^{(2j+1)}}[i]V_{3^{(j)}}[i])}}
\\&+ \sum_{j}{\sum_{i=1}^{n}{H(Y_{1^{(3j+1)}}[i]|V_{1^{(3j+1)}}[i]V_{2^{(2j)}}[i]V_{3^{(j)}}[i])}}
\\&+\sum_{j}{\sum_{i=1}^{n}{H(Y_{2^{(2j)}}[i]|V_{2^{(2j)}}[i]V_{3^{(j)}}[i])}}
\\&+\sum_{j}{\sum_{i=1}^{n}{H(Y_{1^{(3j)}}[i]|V_{1^{(3j)}}[i]V_{2^{(2j)}}[i]V_{3^{(j)}}[i])}}
\\& \overset{(b)}{\leq} \sum_{i=1}^{n}{kH(Y_{1^{(1)}}[i]|V_{1^{(1)}}[i])}
\\&+ \sum_{i=1}^{n}{kH(Y_{3^{(1)}}[i]|V_{2^{(1)}}[i]V_{3^{(1)}}[i])}
\\&+ \sum_{i=1}^{n}{2kH(Y_{1^{(1)}}[i]|V_{1^{(1)}}[i]V_{2^{(1)}}[i]V_{3^{(1)}}[i])}
\\&+ \sum_{i=1}^{n}{2kH(Y_{2^{(1)}}[i]|V_{2^{(1)}}[i]V_{3^{(1)}}[i])}
\end{aligned}
\end{equation}
\normalsize
where step $(a)$ follows from $(V_{1^{(3j)}}[i],V_{3^{(j)}}[i])=h_2(X_{2^{(2j)}}[i],Y_{2^{(2j)}}[i])$,  $(V_{2^{(2j)}}[i],V_{3^{(j)}}[i])=h_1(X_{1^{(3j+1)}}[i],Y_{1^{(3j+1)}}[i])$, $(V_{1^{(3j-1)}}[i],V_{2^{(2j)}}[i])=h_3(X_{3^{(j)}}[i],Y_{3^{(j)}}[i])$, $(V_{1^{(3j+1)}}[i],V_{3^{(j)}}[i])=h_2(X_{2^{(2j+1)}}[i],Y_{2^{(2j+1)}}[i])$, $(V_{2^{(2j+1)}}[i],V_{3^{(j)}}[i])=h_1(X_{1^{(3j+2)}}[i],Y_{1^{(3j+2)}}[i])$, $V_{3^{(j)}}[i]=g_3(X_{3^{(j)}}[i])$, and $V_{2^{(2j)}}[i]=g_2(X_{2^{(2j)}}[i])$. Step $(b)$ follows, since the coding schemes of $S_{1^{(j)}}$s, $S_{2^{(j)}}$s, and $S_{3^{(j)}}$s are the same respectively, $Y_{1^{(j)}}[i]$s, $Y_{2^{(j)}}[i]$s, and $Y_{3^{(j)}}[i]$s have the same probability distribution respectively.

\section{Title of Appendix A} \label{App:AppendixA}
\label{apb36}
\noindent \textbf{Proof of part (b) for bound (\ref{ineq27}).}
\small
\begin{equation} \label{eq3434}
\begin{aligned}
&H(Y_{{\Omega}_1^c}^n)+\sum_{l=2}^{6}{H(Y_{{\Omega}_l^c\bigcap {\Omega}_{l-1}}^n|W_{\mathcal S\backslash{{\Omega}_l}},Y_{{\Omega}_{l-1}^c}^n)}
\\ &\leq H(Y_{2^{(2)}}^n)+H(Y_{1^{(3)}}^n|Y_{2^{(2)}}^n,X_{2^{(2)}}^n)+H(Y_{3^{(1)}}^n|Y_{1^{(3)}}^n,X_{1^{(3)}}^n)
\\&+H(Y_{1^{(2)}}^n|Y_{3^{1}}^n,X_{3^{1}}^n,X_{2^{2}}^n)+H(Y_{2^{(1)}}^n|Y_{3^{(1)}}^n,X_{3^{(1)}}^n)\\&+H(Y_{1^{(1)}}^n|Y_{2^{(1)}}^n,X_{2^{(1)}}^n,X_{3^{(1)}}^n)
\\&\overset{(a)}{\leq} \sum_{i=1}^{n}{H(Y_{2^{(2)}}[i])+H(Y_{1^{(3)}}[i]|V_{1^{(3)}}[i],V_{2^{(2)}}[i]V_{3^{(1)}}[i])}
\\&+\sum_{i=1}^{n}{H(Y_{3^{(1)}}[i]|V_{3^{(1)}}[i])+H(Y_{1^{(2)}}[i]|V_{1^{(2)}}[i]V_{2^{(2)}}[i]V_{3^{(1)}}[i])}
\\&+\sum_{i=1}^{n}{H(Y_{2^{(1)}}[i]|V_{2^{(1)}}[i]V_{3^{(1)}}[i])+H(Y_{1^{(1)}}[i]|V_{1^{(1)}}[i]V_{2^{(1)}}[i]V_{3^{(1)}}[i])}\nonumber
\end{aligned}
\end{equation}
\begin{equation} \label{eq343413}
\begin{aligned}
\\&\overset{(b)}{\leq} \sum_{i=1}^{n}{H(Y_{2^{(1)}}[i])+H(Y_{1^{(1)}}[i]|V_{1^{(1)}}[i]V_{2^{(1)}}[i]V_{3^{(1)}}[i])}
\\&+\sum_{i=1}^{n}{H(Y_{3^{(1)}}[i]|V_{3^{(1)}}[i])+H(Y_{1^{(1)}}[i]|V_{1^{(1)}}[i]V_{2^{(1)}}[i]V_{3^{(1)}}[i])}
\\&+\sum_{i=1}^{n}{H(Y_{2^{(1)}}[i]|V_{2^{(1)}}[i]V_{3^{(1)}}[i])+H(Y_{1^{(1)}}[i]|V_{1^{(1)}}[i]V_{2^{(1)}}[i]V_{3^{(1)}}[i])}
\end{aligned}
\end{equation}
\normalsize
where $(a)$ follows from $V_{2^{(2)}}[i]=g_2(X_{2^{(2)}}[i])$, $V_{3^{(1)}}[i]=g_3(X_{3^{(1)}}[i])$, $V_{2^{(1)}}[i]=g_2(X_{2^{(1)}}[i])$, $(V_{1^{(1)}}[i],V_{3^{(1)}}[i])=h_2(X_{2^{(1)}}[i],Y_{2^{(1)}}[i])$, $(V_{1^{(2)}}[i],V_{2^{(1)}}[i])=h_3(X_{3^{(1)}}[i],Y_{3^{(1)}}[i])$, $(V_{2^{(2)}}[i],V_{3^{(1)}}[i])=h_1(X_{1^{(3)}}[i],Y_{1^{(3)}}[i])$, and $(V_{1^{(3)}}[i],V_{3^{(1)}}[i])=h_2(X_{2^{(2)}}[i],Y_{2^{(2)}}[i])$. Step $(b)$ follows, since the coding schemes of $S_{1^{(j)}}$ and $S_{2^{(j)}}$ are the same respectively, $Y_{1^{(j)}}[i]$s and $Y_{2^{(j)}}[i]$s have the same probability distribution respectively.
\\
\section{Title of Appendix A} \label{App:AppendixA}
\label{apb37}
\noindent \textbf{Proof of part (b) for bound (\ref{ineq28}).}
\small
\begin{equation} \label{eq3434}
\begin{aligned}
&H(Y_{{\Omega}_1^c}^n)+\sum_{l=2}^{7}{H(Y_{{\Omega}_l^c\bigcap {\Omega}_{l-1}}^n|W_{\mathcal S\backslash{{\Omega}_l}},Y_{{\Omega}_{l-1}^c}^n)}
\\ &\leq H(Y_{1^{(4)}}^n)+H(Y_{2^{(2)}}^n|Y_{1^{(4)}}^n,X_{1^{(4)}}^n)+H(Y_{1^{(3)}}^n|Y_{2^{(2)}}^n,X_{2^{(2)}}^n)
\\&+H(Y_{3^{(1)}}^n|Y_{1^{3}}^n,X_{1^{3}}^n)+H(Y_{1^{(2)}}^n|Y_{3^{(1)}}^n,X_{2^{(2)}}^n,X_{3^{(1)}}^n)
\\&+H(Y_{2^{(1)}}^n|Y_{3^{(1)}}^n,X_{3^{(1)}}^n)+H(Y_{1^{(1)}}^n|Y_{2^{1}}^n,X_{2^{1}}^n,X_{3^{(1)}}^n)
\\&\overset{(a)}{\leq} \sum_{i=1}^{n}{H(Y_{1^{(4)}}[i])+H(Y_{2^{(2)}}[i]|V_{2^{(2)}}[i],V_{3^{(1)}}[i])}
\\&+\sum_{i=1}^{n}{H(Y_{1^{(3)}}[i]|V_{1^{(3)}}[i],V_{2^{(2)}}[i],V_{3^{(1)}}[i])+H(Y_{3^{(1)}}[i]|V_{3^{(1)}}[i])}
\\&+\sum_{i=1}^{n}{H(Y_{1^{(2)}}[i]|V_{1^{(2)}}[i],V_{2^{(2)}}[i],V_{3^{(1)}}[i])+H(Y_{2^{(1)}}[i]|V_{2^{(1)}}[i],V_{3^{(1)}}[i])}
\\&+\sum_{i=1}^{n}{H(Y_{1^{(1)}}[i]|V_{1^{(1)}}[i],V_{2^{(1)}}[i],V_{3^{(1)}}[i])}
\\&\overset{(b)}{\leq} \sum_{i=1}^{n}{H(Y_{1^{(1)}}[i])+2H(Y_{2^{(1)}}[i]|V_{2^{(1)}}[i],V_{3^{(1)}}[i])}
\\&+\sum_{i=1}^{n}{3H(Y_{1^{(1)}}[i]|V_{1^{(1)}}[i],V_{2^{(1)}}[i],V_{3^{(1)}}[i])+H(Y_{3^{(1)}}[i]|V_{3^{(1)}}[i])}
\end{aligned}
\end{equation}
\normalsize
where $(a)$ follows from $V_{2^{(2)}}[i]=g_2(X_{2^{(2)}}[i])$, $V_{3^{(1)}}[i]=g_3(X_{3^{(1)}}[i])$, $V_{2^{(1)}}[i]=g_2(X_{2^{(1)}}[i])$, $(V_{1^{(1)}}[i],V_{3^{(1)}}[i])=h_2(X_{2^{(1)}}[i],Y_{2^{(1)}}[i])$, $(V_{1^{(2)}}[i],V_{2^{(1)}}[i])=h_3(X_{3^{(1)}}[i],Y_{3^{(1)}}[i])$, $(V_{2^{(2)}}[i],V_{3^{(1)}}[i])=h_1(X_{1^{(3)}}[i],Y_{1^{(3)}}[i])$, $(V_{1^{(3)}}[i],V_{3^{(1)}}[i])=h_2(X_{2^{(2)}}[i],Y_{2^{(2)}}[i])$, and $(V_{2^{(2)}}[i],V_{3^{(1)}}[i])=h_1(X_{1^{(4)}}[i],Y_{1^{(4)}}[i])$. Step $(b)$ follows, since the coding schemes of $S_{1^{(j)}}$ and $S_{2^{(j)}}$ are the same respectively, $Y_{1^{(j)}}[i]$s and $Y_{2^{(j)}}[i]$s have the same probability distribution respectively.

\bibliographystyle{IEEEtran}
{\footnotesize
\bibliography{Allertonref}}
\end{document}